\documentclass[prd,twocolumn,showpacs,showkeys,floatfix,
  superscriptaddress]{revtex4-2}
\usepackage{verbatim}
\usepackage{amsfonts} % AMS
\usepackage{amssymb} % AMS
\usepackage{amsmath} % AMS
\usepackage{amsthm}  % AMS
\usepackage{graphicx} % Include figure files
\usepackage{subfigure}
\usepackage{array} % array
\usepackage{dcolumn} % Align table columns on decimal point
\usepackage{bm} % bold math
\usepackage{latexsym} % latex symbols
\usepackage{longtable} % long tables
\usepackage{hyperref} % hypertext links 
\usepackage[usenames,dvipsnames]{xcolor}
\usepackage{mathrsfs}
\usepackage{comment}
\usepackage{rotating}
\usepackage{here}
\usepackage{float}
\usepackage{multirow}
\usepackage{diagbox}
\usepackage{verbatim}
\usepackage{stmaryrd}

\graphicspath{{./figs/}}

\allowdisplaybreaks

\renewcommand{\Re}{\mathrm{Re}}
\renewcommand{\Im}{\mathrm{Im}}

\newtheorem{thm}{Theorem}[section]

\providecommand{\GeV}{\;\mathrm{GeV}}

\providecommand{\fm}{\;\mathrm{fm}}

\definecolor{HLBlue}{HTML}{6599FF}
\definecolor{HLOrange}{HTML}{FF6600}

\newcommand{\Tr}{\mathrm{Tr}}

\newcommand{\com}[1]{} % delete comments

\begin{document}

\title{Chiral symmetry and taste symmetry from the eigenvalue spectrum
  of staggered Dirac operators}

%\author{Jon Bailey}
%
%\affiliation{
%  ISED, UIC, Yonsei University,
%  Incheon 21983, South Korea 
%}

\author{Hwancheol Jeong}
\affiliation{
  Lattice Gauge Theory Research Center, FPRD, and CTP, \\
  Department of Physics and Astronomy,
  Seoul National University, Seoul 08826, South Korea
}

\author{Chulwoo Jung}
%
%\email[E-mail: ]{chulwoo@bnl.gov}
%
\affiliation{
  Physics Department, Brookhaven National Laboratory,
  Upton, NY11973, USA
}

\author{Seungyeob Jwa}
\affiliation{
  Lattice Gauge Theory Research Center, FPRD, and CTP, \\
  Department of Physics and Astronomy,
  Seoul National University, Seoul 08826, South Korea
}

\author{Jangho Kim}
\affiliation{
  Institut f\"ur Theoretische Physik,
  Goethe University Frankfurt am Main,
  Max-von-Laue-Str.~1,
  60438 Frankfurt am Main, Germany 
}

\author{Jeehun Kim}
\affiliation{
  Lattice Gauge Theory Research Center, FPRD, and CTP, \\
  Department of Physics and Astronomy,
  Seoul National University, Seoul 08826, South Korea
}

\author{Nam Soo Kim}
\affiliation{
  Department of Electrical and Computer Engineering and the Institute of
  New Media and Communications \\
  Seoul National University, Seoul 08826, South Korea
}

\author{Sunghee Kim}
\affiliation{
  Lattice Gauge Theory Research Center, FPRD, and CTP, \\
  Department of Physics and Astronomy,
  Seoul National University, Seoul 08826, South Korea
}

\author{Sunkyu Lee}
\affiliation{
  Lattice Gauge Theory Research Center, FPRD, and CTP, \\
  Department of Physics and Astronomy,
  Seoul National University, Seoul 08826, South Korea
}

\author{Weonjong Lee}
\email[E-mail: ]{wlee@snu.ac.kr}
\affiliation{
  Lattice Gauge Theory Research Center, FPRD, and CTP, \\
  Department of Physics and Astronomy,
  Seoul National University, Seoul 08826, South Korea
}

\author{Youngjo Lee}
\affiliation{
  Department of Statistics,
  Seoul National University, Seoul 08826, South Korea
}

\author{Jeonghwan Pak}
\affiliation{
  Lattice Gauge Theory Research Center, FPRD, and CTP, \\
  Department of Physics and Astronomy,
  Seoul National University, Seoul 08826, South Korea
}

%\author{Stephen R. Sharpe}
%
%\email[E-mail: ]{srsharpe@uw.edu}
%
%\affiliation{
%  Physics Department,
%  University of Washington,
%  Seattle, WA 98195-1560, USA
%}
%

\collaboration{SWME Collaboration}

\date{\today}

\begin{abstract}
  We investigate general properties of the eigenvalue spectrum for
  improved staggered quarks.
  We introduce a new chirality operator $[\gamma_5 \otimes 1]$ and a
  new shift operator $[1 \otimes \xi_5]$, which respect the same
  recursion relation as the $\gamma_5$ operator in the continuum.
  Then we show that matrix elements of the chirality operator
  sandwiched between two eigenstates of the staggered Dirac operator
  are related to those of the shift operator by the Ward identity of
  the conserved $U(1)_A$ symmetry of staggered fermion actions.
  We perform a numerical study in quenched QCD using HYP staggered
  quarks to demonstrate the Ward identity.
  We introduce a new concept of leakage patterns which collectively
  represent the matrix elements of the chirality operator and the
  shift operator sandwiched between two eigenstates of the staggered
  Dirac operator.
  The leakage pattern provides a new method to identify zero modes
  and non-zero modes in the Dirac eigenvalue spectrum.
  This method is as robust as the spectral flow method but requires
  much less computing power.
  Analysis using a machine learning technique confirms that the
  leakage pattern is universal, since the staggered Dirac eigenmodes
  on normal gauge configurations respect it.
  In addition, the leakage pattern can be used to determine a ratio
  of renormalization factors as a by-product.
  We conclude that it might be possible and realistic to measure the
  topological charge $Q$ using the Atiya-Singer index theorem and the
  leakage pattern of the chirality operator in the staggered fermion
  formalism.
\end{abstract}

\maketitle

%-----
% SECTION 1.
%\input{sec.intro.tex}
%-----
%
\section{Introduction}
\label{sec:intr}
%
%---------------
% motivation
%---------------
It is important to understand the low-lying eigenvalue spectrum of the
Dirac operator, which exhibits the topological Ward identity of the
Atiya-Singer index theorem \cite{ Atiyah:1963zz}, the Banks-Casher
relationship \cite{ Banks:1979yr}, and the universality of the
distribution of the near-zero modes for fixed topological charge
sectors \cite{ Shuryak:1992pi, Leutwyler:1992yt}.
Study on the eigenvalue spectrum of the Dirac operator is, by nature,
highly non-perturbative.
Hence, numerical tools available in lattice gauge theory provide a
perfect playground to study diverse properties of the Dirac eigenvalue
spectrum.

In lattice QCD, there are a number of popular methods to implement a
discrete version of the continuum Dirac operator.
We are interested in one particular class of lattice fermions that are
widely used in the lattice QCD community: improved staggered quarks
\cite{ Hasenfratz:2001hp, Lee:2002ui, Follana:2006rc}.
Here we study the eigenvalue spectrum of staggered Dirac operators in
quenched QCD to show that the small eigenvalues near zero modes of the
staggered Dirac operators reproduce the continuum properties very
closely, which was originally noticed in Refs.~\cite{ Follana:2004sz,
  Follana:2005km, Durr:2004as}.
To reach this conclusion, the authors of Refs.~\cite{ Follana:2004sz,
  Follana:2005km} performed a number of tests, verifying consistency
of lattice data with (1) the Atiya-Singer index theorem that describes
the chiral Ward identity relating the zero modes to the topological
charge; (2) the Banks-Casher relationship that relates the chiral
condensate to the density of eigenvalues at the zero modes; and (3)
the universality of the small eigenvalue spectrum in the
$\varepsilon$-regime predicted by random matrix theory.
In addition, the authors of Refs.~\cite{ Azcoiti:2014pfa, Durr:2013gp}
used the spectral flow method of Adams \cite{ Adams:2009eb} to
identify the zero modes from a mixture with non-zero modes.
The spectral flow method is robust but highly expensive in a
computational sense.

Here, we introduce a new, advanced chirality operator $[\gamma_5
  \otimes 1]$, which respects the continuum algebra of $\gamma_5$.
We show that matrix elements of this chirality operator
between eigenstates are related to those of the shift operator $[1
  \otimes \xi_5]$ through the Ward identity of the conserved $U(1)_A$
symmetry of staggered fermions.
In addition, we introduce a new concept of leakage patterns to
distinguish zero modes from non-zero modes.
Using the leakage pattern of the chirality and shift operators, we
show that one can measure the zero modes as reliably as when using the
spectral flow method.
Hence, one could determine the topological charge $Q$ using the
leakage pattern with much smaller computational cost than by using the
spectral flow.
We also show that it is possible to determine the ratio of
renormalization constants $Z_{P\times S}/ Z_{P\times P}$ using the
leakage pattern.
%

%---------------------
% section description
%---------------------
% Section 2.
In Section \ref{sec:qc}, we briefly review the continuum theory of the
eigenvalue spectrum and its relation to the quark condensate $\langle
\bar\psi\psi \rangle$. We also review the Atiya-Singer index theorem
in brief.
% Section 3.
In Section \ref{sec:spec-decom}, we briefly review the eigenvalue
spectrum of staggered Dirac operators that is obtained using the
Lanczos algorithm.
% Section 4.
In Section \ref{sec:chiral}, we briefly review the conserved $U(1)_A$
symmetry in the staggered fermion formalism and explain its role in
the eigenvalue spectrum of staggered Dirac operators.  We also present
numerical examples to help readers to understand basic concepts and
notation.
% Section 5.
In Section \ref{sec:chirality}, we define the chirality operator
$[\gamma_5 \otimes 1]$ and the shift operator $[1 \otimes \xi_5]$.
We show that they respect the continuum recursion relation of
$\gamma_5$.
Then we derive the chiral Ward identity of the $U(1)_A$ symmetry
to show that the matrix elements of the chirality operator are
related to those of the shift operator through the Ward identity.
We discuss the eigenvalue spectrum in the continuum limit and
introduce a new notation of quartet indices.
Then we introduce the concept of leakage patterns for the chirality
operator and the shift operator.
We also present numerical examples to demonstrate that the leakage
patterns of zero modes are completely different from those of non-zero
modes.
%
% Section 6.
In Section \ref{sec:machine:learn}, we review a machine learning
technique and describe how to apply it to extract efficiently the
quartet structure of non-zero modes using leakage patterns.
%
% Section 7.
%
In Section \ref{sec:renorm}, we explain how the leakage pattern of the
zero modes can be used to determine the ratio of the renormalization
factors non-perturbatively.
%
% Section 8.
%
In Section \ref{sec:conclude}, we conclude.
%
% Appendix 
%
The appendices contain technical details on Lanczos algorithms and
mathematical proofs, and more plots of leakage patterns for diverse
topological charge values.

Preliminary results of this paper are published in Refs.~\cite{
  Cundy:2016tmw, Jeong:2017kst, Jeong:2020map}.

%-------------------
% END of SECTION 1.
%-------------------

%-----
% SECTION 2.
%\input{sec.condensate.tex}
%-----

%
\section{Quark condensate in the continuum}
\label{sec:qc}
In the continuum the quark condensate is given by
\begin{align} \label{eq:qc_naive}
  \left\langle \bar{\psi}\psi \right\rangle 
  &= \frac{1}{N_f} \sum_{f} \left\langle 0|
  \bar{\psi}_f \psi_f |0 \right\rangle 
  \\ 
  &= - \frac{1}{V N_f} \int d^4 x
  \ \Tr \left( \frac{1}{ D + m } \right) \,,
\end{align}
where $D$ is the Dirac operator, $m$ is the quark mass, $x$ is the
space-time coordinate, $V$ is the volume, and $N_f$ is the number of
flavors with the same mass $m$.
The trace is a sum over spin and color.
Let us think of the eigenvalues of the Dirac operator.
$D$ is anti-Hermitian, so its eigenvalues are purely imaginary or zero.
\begin{align}
& D^\dagger = -D 
\\
& D u_\lambda(x) = i\lambda u_\lambda(x)
\end{align}
where $\lambda$ is a real eigenvalue, and $u_\lambda(x)$ is the
corresponding eigenvector.

By spectral decomposition \cite{Leutwyler:1992yt},
\begin{align}
  S_f(x,y) &= \langle \psi_f(x) \bar{\psi}_f(y) \rangle
  \nonumber \\
  &= \sum_{\lambda} \frac{1}{i\lambda + m} u_\lambda(x) u^\dagger_\lambda(y)
  \\
  \langle \bar{\psi}\psi \rangle &= - \frac{1}{V}
  \sum_\lambda \frac{1}{i\lambda + m} \int d^4 x \  \Tr ( u_\lambda(x)
  u_\lambda^\dagger(x) )
  \\
  &= - \frac{1}{V} \sum_\lambda \frac{1}{i\lambda + m}
\end{align}
where we adopt the normalization convention
\begin{align}
\langle u_a | u_b
\rangle = \displaystyle\int d^4x \; u^\dagger_a(x) u_b(x) =
\delta_{ab} \,.
\end{align}
Thanks to the chiral symmetry,
\begin{align}
\gamma_5 D &= - D \gamma_5
\\
D \gamma_5 | u_\lambda \rangle &= -i \lambda \gamma_5 | u_\lambda \rangle 
\end{align}
Let us define $u_{-\lambda} \equiv \gamma_5 u_\lambda$, so that $D
u_{-\lambda} = -i\lambda u_{-\lambda}$.
Hence, if there exists $u_\lambda$ for $\lambda \ne 0$, then the
parity partner eigenstate $u_{-\lambda}$ with negative eigenvalue
$-i\lambda$ must also exist.
Now let us separate the zero mode contribution from the spectral
decomposition.
\begin{equation}
  \langle \bar{\psi}\psi \rangle = - \frac{1}{V} 
  \sum_{\lambda > 0} \frac{2m}{\lambda^2+m^2} - \frac{n_+ + n_-}{m V} .
\end{equation}
Here, $n_+$ ($n_-$) is the number of right-handed (left-handed) zero
modes per flavor.
Let us define the subtracted quark condensate
$\langle\bar{\psi}\psi\rangle_{\textrm{sub}}$:
\begin{align} \label{eq:qc_sub}
  \langle \bar{\psi}\psi \rangle_{\textrm{sub}}
  &= \left\langle
  \bar{\psi}\psi \right\rangle + \frac{n_+ + n_-}{m V}
  = - \frac{1}{V} \sum_{\lambda > 0} \frac{2m}{\lambda^2+m^2} \,.
  \\
  &= - \frac{1}{V}\sum_{n} \frac{2m}{\lambda_n^2+m^2}
  \quad \text{with} \quad \lambda_n > 0
  \\
  &= - \int_{-\infty}^{+\infty} d\lambda \, \frac{m}{\lambda^2 + m^2}
  \,\rho_s(\lambda)  \,,
\end{align}
where the spectral density $\rho_s(\lambda)$ is
\begin{align}
  \rho_s(\lambda) &= \frac{1}{V} \sum_n \delta(\lambda - \lambda_n) \,.
\end{align}
Here, $\rho_s$ is a spectral density on a single gauge configuration
with volume $V$.
Now let us average over a full ensemble of gauge field configurations
and take the limit of infinite volume ($V \to \infty$).
Then, in that limit, the spectral density $\rho(\lambda) = \langle
\rho_s(\lambda) \rangle$ has a well defined (smooth and continuous)
value as $\lambda \to 0$.
We can define the chiral condensate as
\begin{align}
  \Sigma &= - \langle 0 | \bar{\psi} \psi | 0 \rangle_{\textrm{sub}}(m=0)
  \nonumber \\
  &= \lim_{m \to 0} \int_{-\infty}^{+\infty} d\lambda \,
  \frac{m}{\lambda^2 + m^2} \,\rho(\lambda)
  = \pi \rho(0) \,,
\end{align}
which is the Banks-Casher relation.
The subtracted quark condensate $\langle \bar{\psi}\psi
\rangle_{\textrm{sub}}$ is expected to behave well in the chiral limit,
even though the contribution from the zero modes is divergent as a simple
pole in the chiral limit.
Hence, in the numerical study on the lattice, it is important to
identify the would-be zero modes which correspond to the zero modes in
the continuum limit, and to remove them in the calculation of the quark
condensate.
Before proceeding, let us briefly go through the index theorem.
In the continuum theory in Euclidean space, the axial Ward identity
\cite{Weinberg:1996kr} is
\begin{equation}
  \partial_\mu A_\mu (x) = 2 m P(x) - 2 N_f q(x) \,.
  \label{eq:anom-WI}
\end{equation}
Here $A_\mu \equiv \bar{\psi} \gamma_\mu\gamma_5 \psi$ is the axial
vector current in the flavor singlet representation, $P \equiv
\bar{\psi} \gamma_5 \psi$ is the corresponding pseudo-scalar operator,
and $q \equiv \cfrac{1}{32\pi^2} \Tr [ F_{\mu\nu}
  \widetilde{F}_{\mu\nu} ]$ is the topological charge density (=
winding number density).
Now the topological charge $Q$ is
\begin{align}
  Q &\equiv \int d^4 x \left< q(x) \right> 
  \\
  &= - \frac{1}{2 N_f} \int d^4 x \left<
  \partial_\mu A_\mu (x) - 2 m P(x) \right> 
  \\
  &= \frac{m}{N_f} \int d^4 x
  \left< \bar{\psi} \gamma_5 \psi \right>
\end{align}
Using the spectral decomposition, we can rewrite $Q$ as follows,
\begin{align}
Q &= - m \sum_\lambda \frac{1}{i\lambda+m}
\int d^4 x \  \left[ u_\lambda^\dagger(x) \gamma_5 u_\lambda(x)
 \right] \,.
\end{align}
Noting that $\gamma_5 u_\lambda(x) = u_{-\lambda}(x)$ for $\lambda \ne
0$,
\begin{align}
\int d^4 x \  \left[ u_\lambda^\dagger(x) \gamma_5 u_\lambda(x)
  \right] = \langle u_\lambda | u_{-\lambda} \rangle = 0 \,. 
\end{align}
Hence, only zero modes contribute to $Q$.
For the zero modes, it is convenient to choose the helicity
eigenstates as the basis vectors so that $\langle u_0^L | \gamma_5 |
u_0^L \rangle = -1 $ and $\langle u_0^R | \gamma_5 | u_0^R \rangle =
+1 $, where the superscripts $L,R$ represent left-handed and
right-handed helicity, respectively.
Then deriving the index theorem is straightforward \cite{
  Atiyah:1963zz}:
\begin{equation}
  Q = n_- - n_+ \,,
  \label{eq:indexThm}
\end{equation}
where $n_+$ ($n_-$) is the number of the right-handed (left-handed)
zero modes.

%-------------------
% END of SECTION 2.
%-------------------

%-----
% SECTION 3.
%\input{sec.spec.tex}
%-----

%
\section{Spectral decomposition with staggered fermions}
\label{sec:spec-decom}
A number of improved versions of staggered fermions exist, such as
HYP-smeared staggered fermions \cite{ Hasenfratz:2001hp}, asqtad
improved staggered fermions \cite{ Bazavov:2009bb}, and highly
improved staggered quarks (HISQ) \cite{ Follana:2006rc}.
Here we refer to all of them collectively as ``staggered fermions.''
Staggered fermions have four tastes per flavor by construction \cite{
  Golterman:1984cy}.
Hence, the quark condensate for staggered fermions is defined as
\begin{equation}
  \langle \bar{\chi}\chi \rangle = - \frac{1}{V N_t} \left\langle
  \Tr \left( \frac{1}{ D_s + m } \right) \right\rangle_U \,,
   \label{eq:qc_stag}
\end{equation}
where $\chi$ represents a staggered quark field, $D_s$ is the staggered
Dirac operator for a single valence flavor, $V$ is the lattice volume, and
$N_t$ is the number of tastes.
We measure the quark condensate using a stochastic method.
\begin{align}
  (D_s+m)_{x,y} & \chi(y) = \xi(x)
  \\
  \chi(x) = & \left[ \frac{1}{D_s+m} \right]_{x,y} \xi(y)
  \\
  \Tr \left( \frac{1}{D_s+m} \right) &= \lim_{N_\xi \to \infty}
  \frac{1}{N_\xi} \sum_\xi \sum_y \xi^{\dagger}(y) \chi(y) \,,
  \label{eq:qc_gs}
\end{align}
where $x,y$ are representative indices which represent the space-time
coordinates, taste, and color indices collectively.
Here $\xi(x)$ represents either Gaussian random numbers or $U(1)$
noise random numbers which satisfy a simple identity:
\[
\lim_{N_\xi \to \infty} \frac{1}{N_\xi} \sum_\xi \xi^\dagger(x)
\xi(y) = \delta_{xy} \,,
\]
where $N_\xi$ is the number of random vector samples.

Staggered fermions have a taste symmetry $SU(4)_L \otimes SU(4)_R
\otimes U(1)_V$ in the continuum limit at $a=0$ \cite{ Lee:1999zxa}.
However, this symmetry breaks down to a subgroup $U(1)_V \otimes
U(1)_A$ on the lattice with $a \ne 0$ \cite{ Lee:1999zxa,
  Golterman:1984cy}.
The remaining axial symmetry $U(1)_A$ plays an important role in
protecting the quark mass from receiving an additive renormalization.
In addition, it does not have any axial anomaly.

\com{explain how to obtain $\lambda_i^2$ using Lanczos algorithm.}

The Dirac operator ($D_s$) of staggered fermions is anti-Hermitian:
$D_s^\dagger = - D_s$.
Hence, its eigenvalues are purely imaginary:
\begin{align}
  D_s | f^s_\lambda \rangle &= i\lambda | f^s_\lambda \rangle \,,
  \label{eq:stag-dirac-spec-1}
\end{align}
where $\lambda$ is real.
Here, the subscript $s$ and superscript $s$ represent staggered quarks.

In practice, when we obtain eigenvalues of $D_s$ numerically, we use
the following relationship instead of
Eq.~\eqref{eq:stag-dirac-spec-1}:
\begin{align}
  D_s^\dagger D_s | g^s_{\lambda^2} \rangle
  &= \lambda^2 | g^s_{\lambda^2} \rangle
  \label{eq:stag-dirac-spec-2}
\end{align}
where the $| g^s_{\lambda^2} \rangle$ state is a mixture of the two
eigenvectors $| f^s_{+\lambda} \rangle$ and $| f^s_{-\lambda}
\rangle$.
In other words,
\begin{align}
  | g^s_{\lambda^2} \rangle
  &= c_1 | f^s_{+\lambda} \rangle + c_2 | f^s_{-\lambda} \rangle
\end{align}
where the $c_i$ are complex numbers that satisfy the normalization
condition
\begin{align}
|c_1|^2 + |c_2|^2 &= 1 \,.
\end{align}
The numerical algorithm is a variation of a Lanczos algorithm adapted
for lattice QCD \cite{ Lanczos:1950zz}.
Details on the numerical algorithms as well as comprehensive
references are given in Appendix \ref{app:lanczos}.

\com{describe the advantage of obtaining $\lambda_i^2$ instead of
  $\lambda_i$.}

Why do we obtain $\lambda^2$ instead of $i\lambda$?
The first reason is that doing so allows us to use even-odd
preconditioning \cite{ DeGrand:1990dk}, which makes Lanczos run on
only even or odd sites on the lattice.
This leads to two benefits: One is that there is a substantial gain in
the speed of the code, and the other is that the code uses only half
the memory otherwise required.
Details on even-odd preconditioning are described in Appendix
\ref{app:phase}.
The second reason is that obtaining $\lambda^2$ instead of $i\lambda$
allows us to implement polynomial acceleration algorithms
\cite{Youcef:1984} into Lanczos more easily, since the eigenvalues of
$D_s^\dagger D_s$ are positive definite and have a lower bound
$\lambda^2 > 0$.
Note that staggered fermions have would-be zero modes whose
eigenvalues are small and positive ($\lambda^2 > 0$) on rough gauge
configurations.
There are no exact zero modes ($\lambda=0$) with staggered fermions on
rough gauge configurations \cite{Smit:1986fn}.
Details of our implementation of polynomial acceleration are described
in Appendix \ref{app:lanczos}.

\com{explain how to obtain $|f_{+\lambda} \rangle$ and
  $|f_{-\lambda} \rangle$.}

Hence, we use the Lanczos algorithm to solve
Eq.~\eqref{eq:stag-dirac-spec-2} for the eigenvector $|
g^s_{\lambda^2} \rangle$ as well as the corresponding eigenvalue
$\lambda^2$.
We obtain $|f^s_{+\lambda} \rangle$ and $|f^s_{-\lambda} \rangle$ by
using projection operators defined as
\begin{align}
  P_+ &= (D_s + i\lambda)
  \label{eq:proj_p}
  \\
  P_- &= (D_s - i\lambda) \,,
  \label{eq:proj_m}
\end{align}
where $P_+$ is the projection operator that selects only the $|
f^s_{+\lambda} \rangle$ component and removes the $| f^s_{-\lambda}
\rangle$ component.
Then
\begin{align}
  |\chi_+ \rangle &= P_+ | g^s_{\lambda^2} \rangle
  \label{eq:chi_p}
  \\
  |\chi_- \rangle &= P_- | g^s_{\lambda^2} \rangle
  \label{eq:chi_m}
\end{align}
and the orthonormal eigenvectors are 
\begin{align}
  |f^s_{+\lambda} \rangle &= \frac{|\chi_+ \rangle}
  { \sqrt{ \langle \chi_+ | \chi_+ \rangle } }
  \label{eq:f_p}
  \\
  |f^s_{-\lambda} \rangle &= \frac{|\chi_- \rangle}
  { \sqrt{ \langle \chi_- | \chi_- \rangle } } \,.
  \label{eq:f_m}
\end{align}
%

%-------------------
% END of SECTION 3.
%-------------------

%-----------------------
% SECTION 4.
%-----------------------
%\input{sec.chiral.tex}
%-----------------------
%
\section{Chiral symmetry of staggered fermions}
\label{sec:chiral}

\com{explain the chiral operators in the staggered fermion formalism. }

\com{First, describe the exact Ward identity ($U(1)_A$ symmetry)
  which relates $|f_{+\lambda} \rangle$ and $|f_{-\lambda} \rangle$.}

The two vectors $|f^s_{\pm\lambda}\rangle$ are related to each other
through a chiral Ward identity of staggered fermions.
Here we address this issue of the chiral symmetry of staggered
fermions and its consequences.
%

%\wlee{wlee: begin new subsection}

\subsection{Notation and definitions}
\label{ssec:notation}

Let us begin with notation and definitions.
For staggered fermions, there are two independent methods to
transcribe operators to the lattice: One is the Golterman method
\cite{ Golterman:1984cy, Golterman:1985dz, Golterman:1986jf}, and the
other is the Kluberg-Stern method \cite{ KlubergStern:1983dg,
  Verstegen:1985kt, Lee:2001hc, Lee:1994xs}.
%
%\wlee{EDIT by wlee : begin}
%
In Appendix \ref{app:cmp:gol-klu}, we explain how to construct
chirality operators using both the Golterman method and the
Kluberg-Stern method, and we compare the two methods.
The comparison is summarized in Table \ref{tab:cmp-gol-klu-1}
of Appendix \ref{app:cmp:gol-klu}.
Since the Kluberg-Stern method respects the recursion relationship,
uniqueness of chirality, and the Ward identity while the Golterman
method does not, we adopt the former method to construct bilinear
operators.
%
%\wlee{EDIT by wlee : end}
%
Accordingly we define staggered bilinear operators as
\begin{align}
  \mathcal{O}_{S \times T} (x)
  &\equiv \sum_{A,B}
  \bar{\chi}(x_A) [\gamma_S \otimes \xi_T]_{AB} \chi(x_B)
  \nonumber \\
  &= \sum_{A,B} \bar{\chi}_a (x_A) \overline{(\gamma_S \otimes \xi_T)}_{AB}
  U(x_A,x_B)_{ab} \chi_b(x_B)
  \label{eq:bi-op-1}
\end{align}
where $\chi_b$ are staggered quark fields, and $a,b$ are color
indices.
Here the coordinate $x_A = 2x + A$, $x$ is a coordinate of the
hypercube, and $A,B$ are hypercubic vectors with $A_\mu, B_\mu \in
\{0,1\}$.  The spin-taste matrices are
\begin{align}
  \overline{(\gamma_S \otimes \xi_T)}_{AB} &\equiv
    \frac{1}{4} \Tr(\gamma_A^\dagger \gamma_S \gamma_B \gamma_T^\dagger)
    \label{eq:s-t-1}
\end{align}
where $\gamma_S$ represents the Dirac spin matrix, and
$\xi_T$ represents the $4 \times 4$ taste matrix.
In addition,
\begin{align}
  U(x_A,x_B) &\equiv
  \mathbb{P}_{SU(3)} \bigg[ \sum_{p \in \mathcal{C}}
    V(x_A,x_{p_1}) V(x_{p_1},x_{p_2})
    \nonumber \\
    & \qquad \qquad \qquad \cdots V(x_{p_n},x_{B})
    \bigg]
  \label{eq:U-1}
\end{align}
where $\mathbb{P}_{SU(3)}$ represents the $SU(3)$ projection, and
$\mathcal{C}$ represents the complete set of shortest paths from $x_A$
to $x_B$.
$V(x,y)$ represents the HYP-smeared fat link \cite{ Lee:2002ui,
  Hasenfratz:2001hp} for HYP staggered fermions, the Fat7 fat link
\cite{ Lee:2002ui, Lee:2002fj, Orginos:1999cr, Lepage:1998vj} for
asqtad staggered fermions or HISQ, and the thin gauge link for
unimproved staggered fermions.
%
%\wlee{wlee: end new subsection}

The conserved $U(1)_A$ axial symmetry transformation is
\begin{align}
  \Gamma_\epsilon(A,B,a,b) &\equiv [\gamma_5 \otimes \xi_5]_{AB;ab}
  \nonumber \\
  &= \overline{ (\gamma_5 \otimes \xi_5) }_{AB} \cdot \delta_{ab}
  \nonumber \\  
  &= \epsilon(A) \cdot \delta_{AB} \cdot \delta_{ab} 
\end{align}
where $\Gamma_\epsilon$ is often called ``distance parity,'' and
\begin{align}
  \epsilon(A) &\equiv  (-1)^{S_A}
  \\
  S_A &\equiv \sum_{\mu=1}^4 A_\mu \,.
\end{align}
Under the $U(1)_A$ transformation, the staggered Dirac operator
transforms as follows,
\begin{align}
  & \Gamma_\epsilon D_s \Gamma_\epsilon = D_s^\dagger = -D_s
  \\
  & \Gamma_\epsilon D_s = -D_s \Gamma_\epsilon \,.
\end{align}
Therefore,
\begin{align}
  D_s | f^s_{+\lambda} \rangle &= +i \lambda | f^s_{+\lambda} \rangle
  \nonumber \\
  D_s \Gamma_\epsilon | f^s_{+\lambda} \rangle &=
  -i \lambda \Gamma_\epsilon | f^s_{+\lambda} \rangle \,,
\end{align}
and $f^s_{-\lambda}$ can be obtained from $f^s_{+\lambda}$ through
$\Gamma_\epsilon$ transformation as follows,
\begin{align}
  \Gamma_\epsilon | f^s_{+\lambda} \rangle &=
  e^{+i\theta} | f^s_{-\lambda} \rangle
  \nonumber \\
  \Gamma_\epsilon | f^s_{-\lambda} \rangle &=
  e^{-i\theta} | f^s_{+\lambda} \rangle \,.
  \label{eq:eps-WI}
\end{align}
In general, there is no constraint for the real phase $\theta$, so 
we expect its probability distribution to be random.
In practice, however, we make use of even-odd preconditioning, and
we obtain the odd site fermion fields ($| g_o \rangle$) from the
even site fermion fields ($| g_e \rangle$) with the relation $| g_o
\rangle = \eta\, D_{oe} | g_e \rangle$, where $D_{oe}$ is a sector of
$D_s$ that connects even site fields to odd site fields, and $\eta$
is a random complex number.
Hence the distribution of $\theta$ depends on our choice of $\eta$.
In our numerical study, we set $\eta=1$.
Then $\theta$ is given by
\begin{align}
  \theta &= \pi + 2\beta \,,
  \qquad
  \beta = \arctan(\lambda)
  \,.
  \label{eq:theta2}
\end{align}
Details on the even-odd preconditioning and the derivation of
Eq.~\eqref{eq:theta2} are explained in Appendix \ref{app:phase}.
We expect that if there exists an eigenvector $| f^s_{+\lambda}
\rangle$, there must be a corresponding parity partner $|
f^s_{-\lambda} \rangle$ due to the exact chiral symmetry
$\Gamma_\epsilon$.
In other words, the Ward identity of Eq.~\eqref{eq:eps-WI} comes
directly from the conserved $U(1)_A$ axial symmetry.

\subsection{Numerical Examples}
\label{ss:num-ex}

\com{Show examples.}

We now use numerical examples to demonstrate how the above theory
works in quenched QCD.
In Table \ref{tab:in-para}, details of the gauge configurations are
presented.
%----------------------
% TABLE 1.
%----------------------
%\input{table/tab_para}
%----------------------
%
\begin{table}[!h]
  \caption{Input parameters for numerical study in quenched QCD.  For
    more details, refer to Ref.~\cite{Follana:2005km}.  The
    relationship between sample sizes in our study and the number of
    the gauge configurations is non-trivial and discussed later. }
  \label{tab:in-para}
  \renewcommand{\arraystretch}{1.2}
  \begin{ruledtabular}
    \begin{tabular}{r | l}
      parameters & values \\ \hline
      gluon action & tree level Symanzik
      \cite{ Luscher:1984xn, Luscher:1985zq, Alford:1995hw} \\
      tadpole improvement & yes \\
      $\beta$ & 5.0 \\
      geometry & $20^4$ \\
      $a$ & $0.077(1)\fm$ \cite{Bonnet:2001rc} \\
      $1/a$ & $2.6\GeV$ \\
      \hline
      valence quarks & HYP staggered fermions
      \cite{ Lee:2002ui, Kim:2011pz, Kim:2010fj} \\
      $N_f$ & $N_f=0$ (quenched QCD)
    \end{tabular}
  \end{ruledtabular}
\end{table}
%----------------
% END of TABLE 1.
%----------------

%----------
% FIGURE 1.
%----------
%\input{figs/fig_q0_ev}
%----------
\begin{figure}[t]
%  \vspace*{-5mm}
%  \centering
%  \footnotesize
%  \renewcommand{\arraystretch}{1.2}
%  \renewcommand{\subfigcapskip}{-0.55em}
%
%  \hspace*{-7mm}
  \subfigure[]{
    \includegraphics[width=\linewidth]{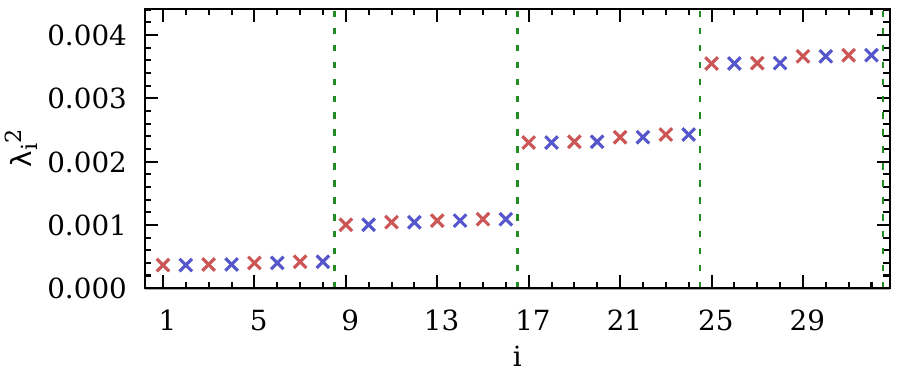}
    \label{sfig:q0-abs-ev}
  }
  \\
  \subfigure[]{
    \includegraphics[width=\linewidth]{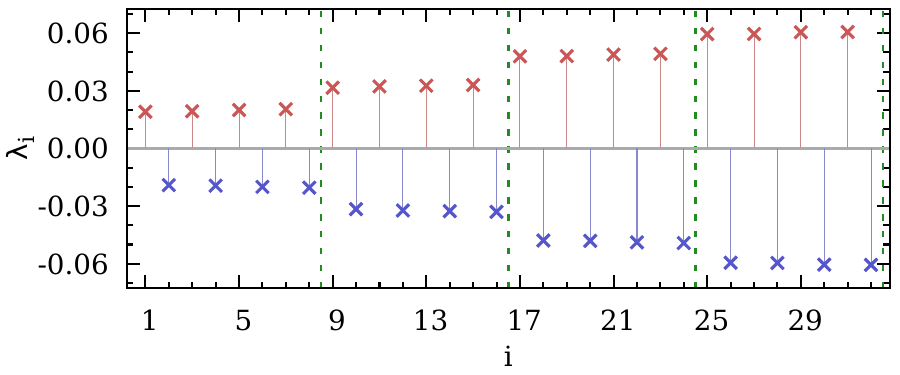}
    \label{sfig:q0-ev}
  }
  \caption{Eigenvalue spectrum of staggered Dirac operator on
  a $Q=0$ gauge configuration.}
  \label{fig:q0-ev}
\end{figure}
%------------------
% END of FIGURE 1.
%------------------
We measure the topological charge $Q$ using gauge links.
We use the $Q(\text{5Li})$ operator defined in Ref.~\cite{
  deForcrand:1997esx, deForcrand:1995qq} after $10 \sim 30$ iterations
of APE smearing with $\alpha = 0.45$ \cite{ Cichy:2014qta,
  Hasenfratz:1998qk, Falcioni:1984ei}.
We show an example of the eigenvalue spectrum for $Q=0$ in
Fig.~\ref{fig:q0-ev}.
Since $Q=0$, we do not expect to find any zero modes for this gauge
configuration.
In Fig.~\ref{fig:q0-ev}\,\subref{sfig:q0-abs-ev}, we show eigenvalues
$\lambda^2$ for the eigenvectors $| g^s_{\lambda^2} \rangle$ defined
in Eq.~\eqref{eq:stag-dirac-spec-2}.
Here we observe eight-fold degeneracy for non-zero eigenmodes due to
the conserved $U(1)_A$ axial symmetry.
Here $\lambda_2 = -\lambda_1$ and, in general, $\lambda_{2n} =
-\lambda_{2n-1}$ for integer $n > 0$.
In other words, $\lambda_{2n}$ is the parity partner of
$\lambda_{2n-1}$.
For each $\lambda_i$, there exists four-fold degeneracy due to
approximate $SU(4)$ taste symmetry.
For each of these four-fold degenerate eigenvalues (for example
$\lambda_1, \lambda_3, \lambda_5, \lambda_7$ in
Fig.~\ref{fig:q0-ev}\,\subref{sfig:q0-abs-ev}), there exists a parity
partner eigenvalue due to the $U(1)_A$ symmetry: $\lambda_2 = -
\lambda_1$, $\lambda_4 = - \lambda_3$, $\lambda_6 = - \lambda_5$, and
$\lambda_8 = - \lambda_7$ (refer to
Fig.~\ref{fig:q0-ev}\,\subref{sfig:q0-ev}).

Let us turn to an example with $Q=-1$.
Since $Q=-1$, we expect to observe four-fold would-be zero modes.
The gauge configurations are so rough that we expect
to observe not exact zero modes but would-be zero modes.
In Fig.~\ref{fig:qm1-ev}, we demonstrate how the would-be zero modes
behave on a gauge configuration with $Q=-1$.
As one can see in Figs.~\ref{fig:qm1-ev}\,\subref{sfig:qm1-abs-ev} and
\ref{fig:qm1-ev}\,\subref{sfig:qm1-ev}, we find four-fold degenerate
would-be zero modes: $\lambda_1, \lambda_2, \lambda_3, \lambda_4$.
Thanks to the $U(1)_A$ chiral Ward identity in Eq.~\eqref{eq:eps-WI},
we find that $\lambda_2 = - \lambda_1$ and $\lambda_4 = -\lambda_3$.
As in the case with $Q=0$, we find that the non-zero eigenmodes are
eight-fold degenerate.
This pattern of four-fold degeneracy for would-be zero modes
and eight-fold degeneracy for non-zero modes is also observed
in the cases with $Q=-2$ and $Q=-3$, which are presented in
Appendix \ref{app:qm2-qm3-dat}.
%-----------------------
% FIGURE 2.
%-----------------------
%\input{figs/fig_qm1_ev}
%-----------------------
\begin{figure}[h]
%  \vspace*{-5mm}
%  \centering
%  \footnotesize
%  \renewcommand{\arraystretch}{1.2}
%  \renewcommand{\subfigcapskip}{-0.55em}
%
%  \hspace*{-7mm}
  \subfigure[]{
    \includegraphics[width=\linewidth]{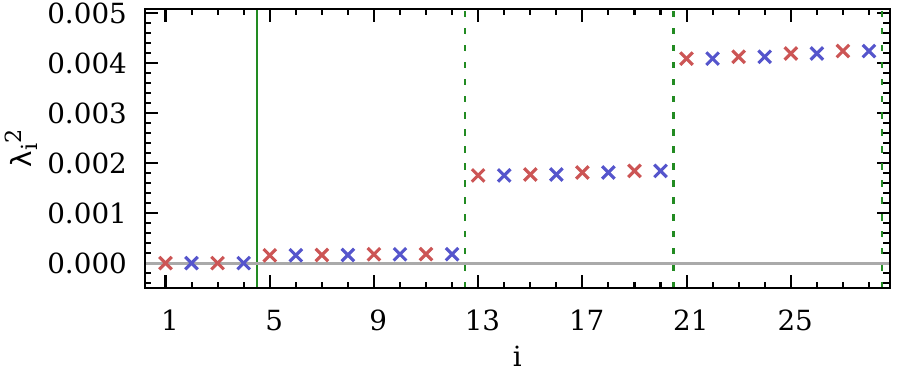}
    \label{sfig:qm1-abs-ev}
  }
  \\
  \subfigure[]{
    \includegraphics[width=\linewidth]{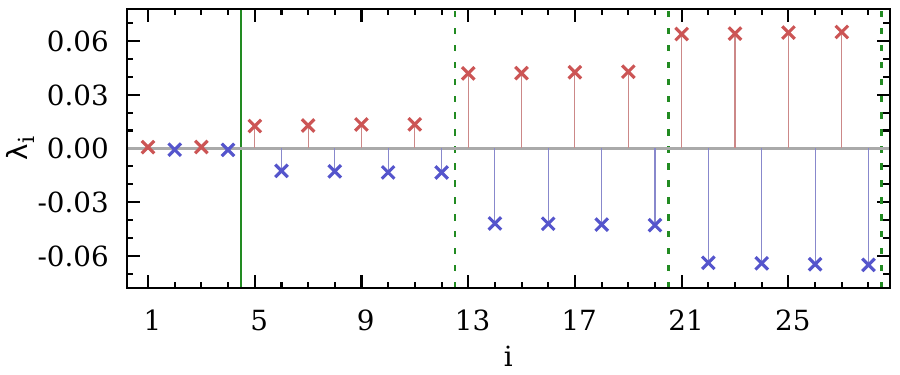}
    \label{sfig:qm1-ev}
  }
  \caption{Eigenvalue spectrum of staggered Dirac operator on
  a $Q=-1$ gauge configuration.}
  \label{fig:qm1-ev}
\end{figure}
%-----------------
% END of FIGURE 2.
%-----------------

At this point, the reader may have already concluded that we can
distinguish would-be zero modes of staggered quarks from non-zero
modes by counting the degeneracy of the eigenvalues \cite{
  Follana:2004sz, Follana:2005km, Donald:2011if}.
This conclusion is true but can lead to wrong answers in practice.
The reason is that, on large lattices, the eigenvalues are so dense
that visually distinguishing between 4-fold and 8-fold degeneracies
is typically impossible.
Hence, we need a practical and robust method to identify would-be
zero modes and non-zero modes of staggered fermions.
The introduction of such a method is the main subject of the next
section, Sec.~\ref{sec:chirality}.

Using the chiral Ward identity of Eq.~\eqref{eq:eps-WI}, we can
measure the phase $\theta$ numerically.
In Fig.~\ref{fig:theta}, we show numerical results (red circles) for
$\theta$.
Here, the blue line represents the theoretical prediction given in
Eq.~\eqref{eq:theta2}.
We find the results are consistent with the prediction within
numerical precision.
%----------------
% FIGURE 3.
%---------------
%\input{figs/fig_theta}
%---------------
\begin{figure}[ht]
  \centering
  \includegraphics[width=\linewidth]{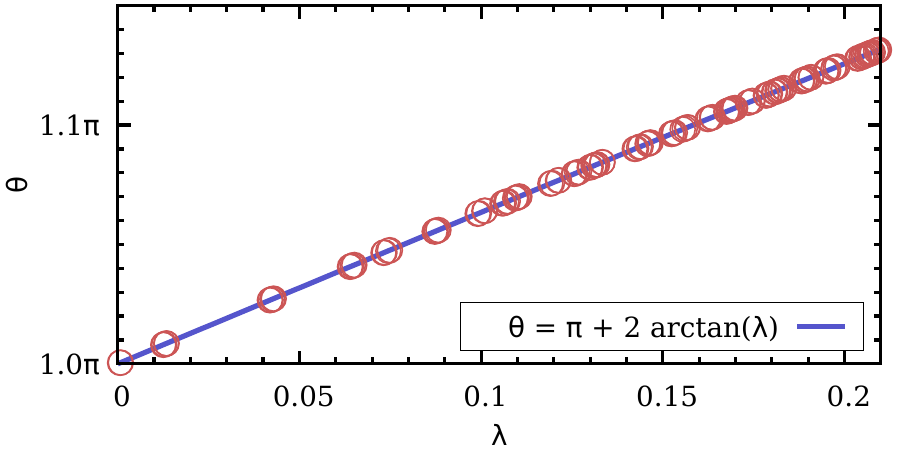}
  \caption{The phase $\theta$ as a function of $\lambda$.  Red circles
    represent numerical results for $\theta$.  The blue line
    represents the prediction from the theory.  Here we use a gauge
    configuration with $Q=-1$ for the measurement. }
  \label{fig:theta}
\end{figure}
%-----------------
% END of FIGURE 3
%-----------------

%------------------
% END of SECTION 4
%------------------

%------------
% SECTION 5.
%------------
%\input{sec.chirality.tex}
%------------
%
\section{Chirality Measurement}
\label{sec:chirality}

%\wlee{wlee: begin editing}

\com{Second, describe chirality $[\gamma_5 \otimes 1]$ in zero modes
  and in non-zero modes.}

To simplify the notation, we introduce the following
convention for eigenvalue indices,
\begin{align}
  D_s | f_j \rangle &= i \lambda_j | f_j \rangle \,,
\end{align}
where $|f_j \rangle = |f^s_{\lambda_j} \rangle$ is
defined in Eq.~\eqref{eq:stag-dirac-spec-1}.
%
%\wlee{EDIT by wlee : begin}
%
Using the Kluberg-Stern method explained in Appendix
\ref{app:cmp:gol-klu}, we define the chirality operator and
abbreviations as follows.
%
%\wlee{EDIT by wlee : end}
%
\begin{align}
  \Gamma_5 (\lambda_i,\lambda_j) &\equiv
  \langle f_i | [\gamma_5 \otimes 1] | f_j \rangle
  \nonumber \\
  &\equiv \sum_{x} \sum_{A,B} \, [f^s_{\lambda_i}(x_A)]^\dagger
  [\gamma_5 \otimes 1]_{x;AB} \, f^s_{\lambda_j}(x_B) \,,
  \label{eq:chirality-1}
  \\
  (\Gamma_5)^i_j &\equiv \Gamma_5 (\lambda_i,\lambda_j) \,,
  \label{eq:chirality-2}
  \\
  |\Gamma_5|^i_j &\equiv |\Gamma_5 (\lambda_i,\lambda_j)| \,,
  \label{eq:chirality-3}
\end{align}
where $x_{A}$, $x_{B}$ and $[\gamma_5 \otimes 1]$ are defined in
Eqs.~\eqref{eq:bi-op-1}-\eqref{eq:U-1} in Subsection
\ref{ssec:notation}, and $\lambda_i$ and $\lambda_j$ represent
eigenvalues of $D_s$.
The chirality operator $[\gamma_5 \otimes 1]$ satisfies the same
relationships as the continuum chirality operator $\gamma_5$.
\begin{align}
  & [\gamma_5 \otimes 1]^{2n+1} = [\gamma_5 \otimes 1] \,,
  \label{eq:recur-1}
  \\
  & [\gamma_5 \otimes 1]^{2n} = [1 \otimes 1] \,,
  \label{eq:recur-2}
  \\
  & [\frac{1}{2}(1\pm\gamma_5) \otimes 1]^n
  = [\frac{1}{2}(1\pm\gamma_5) \otimes 1] \,,
  \label{eq:recur-3}
  \\
  &  [\frac{1}{2}(1+\gamma_5) \otimes 1]
    [\frac{1}{2}(1-\gamma_5) \otimes 1] = 0\,,
  \label{eq:recur-4}
\end{align}
where $n$ is a non-negative integer.
A rigorous proof of Eqs.~\eqref{eq:recur-1}-\eqref{eq:recur-4}
is given in Appendix \ref{app:recur}.
%

%\wlee{EDIT by wlee : begin}
%
Our definition of the chirality operator $[\gamma_5 \otimes 1]$ uses
the Kluberg-Stern method, and is
different from that used in Refs.~\cite{ Smit:1986fn, Follana:2004sz,
  Adams:2009eb}, which adopt the Golterman method.
The old chirality operator (the Golterman method) of Refs.~\cite{
  Smit:1986fn, Follana:2004sz, Adams:2009eb} does not satisfy the
recursion relations of Eqs.~\eqref{eq:recur-1}-\eqref{eq:recur-4}.
In addition, it does not satisfy the chiral Ward identities of
Eqs.~\eqref{eq:wi:u(1)a-1}-\eqref{eq:wi:u(1)a-3}.
This difference is addressed in Appendices \ref{app:cmp:gol-klu} and
\ref{app:recur}.
The bottom line is that the conventional chirality operator (the
Golterman method) does not satisfy the recursion relationships in
Eqs.~\eqref{eq:recur-1}-\eqref{eq:recur-4}, even though it is
classified according to the true irreducible representations of the
lattice rotational symmetry group \cite{ Golterman:1985dz,
  Golterman:1986jf, Verstegen:1985kt}.
%
%\wlee{EDIT by wlee : end}

\com{Third, describe chirality $[1 \otimes \xi_5]$ in zero modes and in
  non-zero modes.}

For further discussion we define another operator $[1 \otimes \xi_5]$,
which we call the ``(maximal) shift operator,''
\begin{align}
  \Xi_5 (\lambda_i,\lambda_j) &\equiv
  \langle f_i | [1 \otimes \xi_5] | f_j \rangle
  \nonumber \\
  &\equiv \sum_x \sum_{A,B} \, [f^s_{\lambda_i}(x_A)]^\dagger
  [1 \otimes \xi_5]_{x;AB} \, f^s_{\lambda_j}(x_B) \,,
  \label{eq:pp-op-1}
  \\
  (\Xi_5)^i_j &\equiv \Xi_5 (\lambda_i,\lambda_j) \,,
  \label{eq:pp-op-2}
  \\
  |\Xi_5|^i_j &\equiv |\Xi_5 (\lambda_i,\lambda_j)| \,.
  \label{eq:pp-op-3}
\end{align}
where $x_{A}$, $x_{B}$ and $[1 \otimes \xi_5 ]$ are defined in
Eqs.~\eqref{eq:bi-op-1}-\eqref{eq:U-1} in Subsection
\ref{ssec:notation}, and $\lambda_i$ and $\lambda_j$ represent
eigenvalues of $D_s$.
This shift operator satisfies the following recursion relations:
\begin{align}
  &[1 \otimes \xi_5]^{2n+1} = [1 \otimes \xi_5]\,,
  \label{eq:rr-xi5-1}
  \\
  &[1 \otimes \xi_5]^{2n} = [1 \otimes 1]\,,
  \label{eq:rr-xi5-2}
\end{align}
where $n$ is a non-negative integer.
The conserved $U(1)_A$ symmetry transformation can be expressed in
terms of the chirality operator and the shift operator as follows,
\begin{align}
  \Gamma_\epsilon &\equiv [\gamma_5 \otimes \xi_5]
  \nonumber \\
  &= [\gamma_5 \otimes 1 ][1 \otimes \xi_5]
  \nonumber \\
  &= [1 \otimes \xi_5] [\gamma_5 \otimes 1] \,.
  \label{eq:wi:u(1)a-1}
\end{align}
A rigorous proof of Eq.~\eqref{eq:wi:u(1)a-1} is given in Appendix
\ref{app:recur}.
In addition, the chirality and shift operators satisfy the
following relations:
\begin{align}
  \Gamma_\epsilon [\gamma_5 \otimes 1 ]
  &= [\gamma_5 \otimes 1 ] \Gamma_\epsilon = [1 \otimes \xi_5] \,,
  \label{eq:wi:u(1)a-2}
  \\
  \Gamma_\epsilon [1 \otimes \xi_5 ]
  &= [1 \otimes \xi_5 ] \Gamma_\epsilon
  = [\gamma_5 \otimes 1 ] \,.
  \label{eq:wi:u(1)a-3}
\end{align}
A rigorous proof of Eqs.~\eqref{eq:wi:u(1)a-2}-\eqref{eq:wi:u(1)a-3}
is also given in Appendix \ref{app:recur}.
%
%\wlee{wlee: end of editing}
%
Therefore, we obtain the Ward identities:
\begin{align}
  e^{+i\theta} [\gamma_5 \otimes 1] | f_{-i} \rangle
  &= [1 \otimes \xi_5] | f_{+i} \rangle \,,
  \nonumber \\
  e^{-i\theta} [\gamma_5 \otimes 1] | f_{+i} \rangle
  &= [1 \otimes \xi_5] | f_{-i} \rangle \,,
  \label{eq:wi-3}
\end{align}
where
\begin{align}
  | f_{\pm i} \rangle &\equiv | f^s_{\pm \lambda_i} \rangle \,.
\end{align}

\com{Fourth, explain the concept of leakage.}

Hence, we define the spectral decomposition by
\begin{align}
  [\gamma_5 \otimes 1] | f_j \rangle
  &= \sum_i (\Gamma_5)^i_j | f_i \rangle
  \label{eq:spec-decomp-3}
\end{align}
where we use Eqs.~\eqref{eq:chirality-1} and \eqref{eq:chirality-2}.
Similarly,
\begin{align}
  [1 \otimes \xi_5] | f_j \rangle
  &= \sum_i (\Xi_5)^i_j | f_i \rangle
\end{align}
where we use Eqs.~\eqref{eq:pp-op-1} and \eqref{eq:pp-op-2}.
Thanks to the Ward identities of Eq.~\eqref{eq:wi-3}, we obtain
\begin{align}
  & e^{-i\theta} \Gamma_5(\lambda_i,+\lambda_j)
  = \Xi_5(\lambda_i, -\lambda_j)
  \nonumber \\
  & e^{-i\theta} (\Gamma_5)^i_{+j} = (\Xi_5)^i_{-j} 
  \nonumber \\
  & |\Gamma_5|^i_{+j} = |\Xi_5|^i_{-j} \,.
  \label{eq:wi-5}
\end{align}
Similarly,
\begin{align}
  & e^{+i\theta} \Gamma_5(\lambda_i,-\lambda_j) = \Xi_5(\lambda_i, +\lambda_j)
  \nonumber \\
  & e^{+i\theta} (\Gamma_5)^i_{-j} = (\Xi_5)^i_{+j}
  \nonumber \\
  & |\Gamma_5|^i_{-j} = |\Xi_5|^i_{+j} \,.
  \label{eq:wi-6}
\end{align}

Applying $\Gamma_\epsilon$ to both sides of
Eq.~\eqref{eq:spec-decomp-3}, we obtain
\begin{align}
  [1 \otimes \xi_5] | f_{j} \rangle
  &= \sum_\ell (\Gamma_5)^\ell_j
  e^{i\theta_\ell} | f_{-\ell} \rangle
  \\
  &= \sum_i (\Xi_5)^i_j | f_{i} \rangle \,.
\end{align}
Hence, we obtain another Ward identity:
\begin{align}
  |\Gamma_5|^{-i}_j &= |\Xi_5|^{+i}_j \,.
  \label{eq:wi-7}
\end{align}
Similarly, we can obtain the Ward identities:
\begin{align}
  |\Gamma_5|^{-i}_{-j} &= |\Xi_5|^{+i}_{-j} \,,  \\
  |\Gamma_5|^{+i}_{j}  &= |\Xi_5|^{-i}_{j} \,.
  \label{eq:wi-8}
\end{align}
We can summarize all the results of
Eqs.~\eqref{eq:wi-5}-\eqref{eq:wi-8} in the following form:
\begin{align}
  |\Gamma_5|^i_j &= |\Xi_5|^{-i}_j = |\Xi_5|^i_{-j}
  = |\Gamma_5|^{-i}_{-j} \,,
  \label{eq:wi-10} \\
  \Leftrightarrow\;
  |\Gamma_5(\lambda_i, \lambda_j)| &= |\Xi_5(-\lambda_i, \lambda_j)| =
  |\Xi_5(\lambda_i, -\lambda_j)| \nonumber\\
  &= |\Gamma_5(-\lambda_i, -\lambda_j)| \,.
  \label{eq:wi-11}
\end{align}
In addition, Hermiticity ensures that we can interchange $\lambda_i$
and $\lambda_j$, which gives the final form of the chiral Ward
identities.
\begin{widetext}
\begin{align}
  |\Gamma_5|^i_j &= |\Xi_5|^{-i}_j = |\Xi_5|^i_{-j} = |\Gamma_5|^{-i}_{-j}
  = |\Gamma_5|^j_i = |\Xi_5|^{-j}_i = |\Xi_5|^j_{-i} = |\Gamma_5|^{-j}_{-i}
  \label{eq:wi-12}
  \\
  \Leftrightarrow\;
  |\Gamma_5(\lambda_i, \lambda_j)|
  &= |\Xi_5(-\lambda_i, \lambda_j)|
  = |\Xi_5(\lambda_i, -\lambda_j)|
  = |\Gamma_5(-\lambda_i, -\lambda_j)|
  = |\Gamma_5(\lambda_j, \lambda_i)|
  = |\Xi_5(-\lambda_j, \lambda_i)|
  \nonumber \\
  &= |\Xi_5(\lambda_j, -\lambda_i)|
  = |\Gamma_5(-\lambda_j, -\lambda_i)| \,.
  \label{eq:wi-13}
\end{align}
\end{widetext}
The quantity $(|\Gamma_5|^i_j)^2$ for $i \ne j$ represents the leakage
probability of the chirality operator.
We call $|\Gamma_5|^i_j$ the leakage parameter for the chirality
operator.
Similarly, the quantity $(|\Xi_5|^i_j)^2$ for $i \ne j$ represents the
leakage probability of the shift operator, and we call $|\Xi_5|^i_j$
the leakage parameter for the shift operator.
By examining the leakage pattern, we can distinguish zero modes and
non-zero modes, which is the main subject of the next subsections.

\subsection{Eigenvalue spectrum of $D_s$ in the continuum}
\label{subsec:eigen-spec}
Here we consider staggered quark actions at $a=0$.
We define a general form of the shift operator corresponding to a
generator of the $SU(4)$ taste symmetry:
\begin{align}
  \Xi_F &= [1 \otimes \xi_F] \,,
  \\
  \xi_F &\in \{\;\xi_5, \;\xi_\mu, \;\xi_{\mu 5}, \;\xi_{\mu\nu} \;\}
  \qquad \text{for $\mu \ne \nu$} \,,
\end{align}
where $\xi_\mu$ respects the Clifford algebra $\{ \xi_\mu, \xi_\nu \}
= 2 \delta_{\mu\nu}$ in Euclidean space.

Let us consider the following quantity $W_1$ in the continuum:
\begin{align}
  W_1 &\equiv \langle f_\ell |\; \Xi_F \; D_s | f_n \rangle
  \\
  D_s | f_n \rangle &= i \lambda_n |f_n \rangle \,.
\end{align}
Since the $SU(4)$ taste symmetry is exactly conserved in the continuum,
we know that
\begin{align}
  [\Xi_F, D_s] &= 0
\end{align}
Hence, we find the following Ward identity:
\begin{align}
  W_1 &= \langle f_\ell | \; \Xi_F \; D_s | f_n \rangle
  = i \lambda_n \langle f_\ell | \; \Xi_F | f_n \rangle
  \\
  &= \langle f_\ell | \; D_s \; \Xi_F | f_n \rangle
  = i \lambda_\ell \langle f_\ell | \; \Xi_F | f_n \rangle
\end{align}
and therefore
\begin{align}
  i(\lambda_\ell -\lambda_n) \cdot \langle f_\ell | \; \Xi_F | f_n \rangle
  &= 0 \,.
\end{align}
Hence, in the continuum, we find the following properties
of the eigenvalue spectrum:
\begin{itemize}
\item If $\lambda_\ell \ne \lambda_n$, $(\Xi_F)^\ell_n = \langle
  f_\ell | \; \Xi_F | f_n \rangle = 0$.  In other words, if the
  eigenvalues are different, there is no leakage ($(\Xi_F)^\ell_n =
  0$) between the two eigenmodes.
\item If $\lambda_j \equiv \lambda_\ell =\lambda_n $, $(\Xi_F)^\ell_n
  \ne 0$ is possible. In this case, the eigenvalues are degenerate
  and belong to a quartet such that
  \begin{align}
    D_s | f_{j,m} \rangle &= i \lambda_j | f_{j,m} \rangle
    \\
    |f_\ell \rangle, |f_n \rangle &\in \{ | f_{j,m} \rangle
    \; \text{with $m=1,2,3,4$} \} \,,
  \end{align}
  where $m$ is a taste index which represents the four-fold degeneracy
  for the eigenvalue $\lambda_j$, $|f_\ell \rangle$ and $|f_n \rangle$
  are linear combinations of the quartet $\{ | f_{j,m} \rangle \}\,,$
  and the eigenvectors for different $m$ are orthogonal to each other
  by construction due to the Lanczos algorithm.
\item We know that the staggered fermion field $\chi^{c}(x_A)$ is
  mapped into the continuum fermion field $\psi_{\alpha;t}^{c}(x)$,
  where $\alpha$ represents a Dirac spinor index, $c$ represents a
  color index, and $t=1,2,3,4$ represents a taste index.
  Hence, for a given eigenvalue $\lambda_j$, there remain four
  degrees of freedom which come from the taste index.
  Accordingly, for a given eigenvalue $\lambda_j$, there are four
  degenerate eigenstates $|f_{j,m} \rangle$ with $m=1,2,3,4$.
\item If we know all four eigenstates $\{ | f_{j,m} \rangle \}$ for a
  certain eigenvalue $\lambda_j$, we find that
  \begin{align}
    \Tr (\Xi_F) &= \sum_{m=1}^4 (\Xi_F)^{j,m}_{j,m}
    \nonumber \\
    &= \sum_{m=1}^4 \langle f_{j,m} | \Xi_F | f_{j,m} \rangle = 0
  \end{align}
  This is because the $SU(4)$ group generators are traceless in the
  fundamental representation.
\end{itemize}

However, on the lattice at $a \ne 0$, the taste symmetry is broken
by terms of order $a^2 \alpha_s^n$ with $n \ge 1\,,$ which is
explained in Ref.~\cite{Lee:1999zxa}.
In addition, for $a \ne 0$,
\begin{align}
  D_s | f_{j,m} \rangle &= i\lambda_{j,m} | f_{j,m} \rangle
\end{align}
and $\lambda_{j,m} \ne \lambda_{j,m'}$ in general for $m \ne m'$,
which reflects the taste symmetry breaking effect at $a \ne 0$.
We know that $\lambda_{j,m} = \lambda_{j,m'}$ for all $m,m'$
in the continuum ($a=0$) due to the exact taste symmetry.
Hence, on the finite lattice, we expect a small deviation from
the above continuum properties.
A good barometer to measure this effect is to monitor $T_5$
\begin{align}
  T_5 &\equiv \frac{1}{4} \Tr (\Xi_5)
  = \frac{1}{4} \sum_m (\Xi_5)^{j,m}_{j,m} 
  \label{eq:t5}
\end{align}
and measure how much it deviates from zero (the continuum value).
Another direct barometer to measure effects of taste symmetry breaking
is the leakage $S_5$ from one quartet ($\lambda_\ell$) to another
quartet ($\lambda_j$) with $\lambda_\ell \ne \lambda_j$.
\begin{align}
  S_5 &\equiv \frac{1}{16} \sum_{m,m'} |\Xi_5|^{\ell,m}_{j,m'}
  = \frac{1}{16} \sum_{m,m'}
  | \langle f_{\ell,m} | \; \Xi_5 | f_{j,m'} \rangle |
  \label{eq:s5}
\end{align}
The size of the leakage $S_5$ indicates directly how much the taste
symmetry is broken at $a \ne 0$, since $S_5 = 0$ in the continuum.
We present numerical results for $T_5$ and $S_5$ in the next
subsection.
\subsection{Numerical study on chirality and leakage}
\label{subsec:leak}

Here we use dual notations for the eigenmodes: One is the serial
index $i$ for $\lambda_i$ and the other is the quartet index $j$ with
taste index $m$ for $\lambda_{j,m}$ .
The serial index is convenient for the plots, tables, and leakage
patterns such as $|\Gamma_5|^a_b$, while the quartet index is
convenient to explain the eigenstates classified by the taste symmetry
group.
The one-to-one mapping from the serial index $i$ to the quartet
indices $j,m$ is given in Table \ref{tab:map-i-j} for the quartet
index $j=0, \pm 1$ when $Q = \pm 1$.
The mapping for the quartet index $j=\pm 2$ (non-zero modes) is
similar.
%------------
% TABLE 2.
%------------
%\input{table/tab_quartet}
%------------
\begin{table}[!tbhp]
  \caption{ One-to-one mapping of serial index $i$ of the $\lambda_i$
    eigenstate into a quartet index $j$ and taste index $m$ for
    $\lambda_{j,m}=\lambda_i$.  Here $\lambda_{2n} = -\lambda_{2n-1}$
    and $\lambda_{-j,m} = -\lambda_{+j,m}$. Here we assume $Q=\pm1$.}
  \label{tab:map-i-j}
  \renewcommand{\arraystretch}{1.2}
  \begin{ruledtabular}
    \begin{tabular}{l | l | c | r  r | l}
      $\lambda_i$ & $\lambda_{j,m}$ & $i$ & $j$ & $m$ & mode
      \\ \hline
      $\lambda_1$ & $\lambda_{0,1}$ & 1   & 0   & 1  & zero \\
      $\lambda_2$ & $\lambda_{0,2}$ & 2   & 0   & 2  & zero \\
      $\lambda_3$ & $\lambda_{0,3}$ & 3   & 0   & 3  & zero \\
      $\lambda_4$ & $\lambda_{0,4}$ & 4   & 0   & 4  & zero 
      \\ \hline
      $\lambda_5$ & $\lambda_{+1,1}$ & 5   & $+1$ & 1 & non-zero \\
      $\lambda_7$ & $\lambda_{+1,2}$ & 7   & $+1$ & 2 & non-zero \\
      $\lambda_9$ & $\lambda_{+1,3}$ & 9   & $+1$ & 3 & non-zero \\
      $\lambda_{11}$ & $\lambda_{+1,4}$ & 11   & $+1$ & 4 & non-zero 
      \\ \hline
      $\lambda_6$ & $\lambda_{-1,1}$ & 6   & $-1$ & 1 & non-zero \\
      $\lambda_8$ & $\lambda_{-1,2}$ & 8   & $-1$ & 2 & non-zero \\
      $\lambda_{10}$ & $\lambda_{-1,3}$ & 10   & $-1$ & 3 & non-zero \\
      $\lambda_{12}$ & $\lambda_{-1,4}$ & 12   & $-1$ & 4 & non-zero 
    \end{tabular}
  \end{ruledtabular}
\end{table}
%----------------
% END of TABLE 2.
%----------------

%------------
% FIGURE 4.
%------------
%\input{figs/fig_qm1_leak_f1}
%------------
\begin{figure}[t]
%  \vspace*{-5mm}
%  \centering
%  \footnotesize
%  \renewcommand{\arraystretch}{1.2}
%  \renewcommand{\subfigcapskip}{-0.55em}
%
%  \hspace*{-7mm}
  \subfigure[$|\Gamma_5|^i_1 = |\Gamma_5( \lambda_i, \lambda_1)|$]{
    \includegraphics[width=\linewidth]{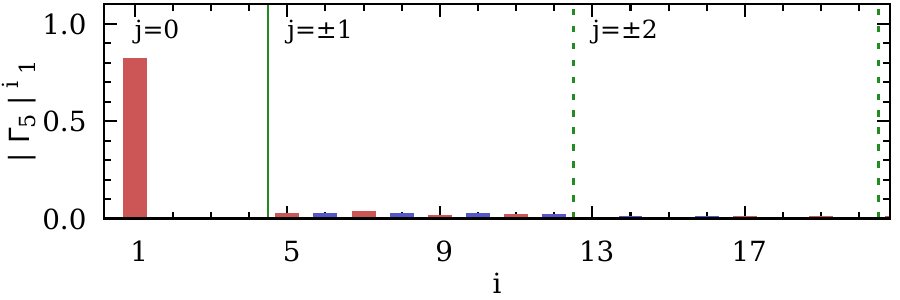}
    \label{sfig:qm1-leak-g5-f1}
  }
  \\
  \subfigure[$|\Xi_5|^i_1 = |\Xi_5( \lambda_i, \lambda_1)|$]{
    \includegraphics[width=\linewidth]{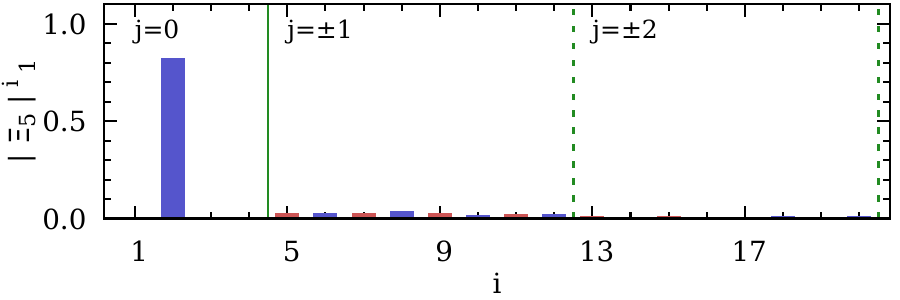}
    \label{sfig:qm1-leak-xi5-f1}
  }
  \\
  \subfigure[$|\Xi_5|^i_2 = |\Xi_5( \lambda_i, \lambda_2 = -\lambda_1)|$]{
    \includegraphics[width=\linewidth]{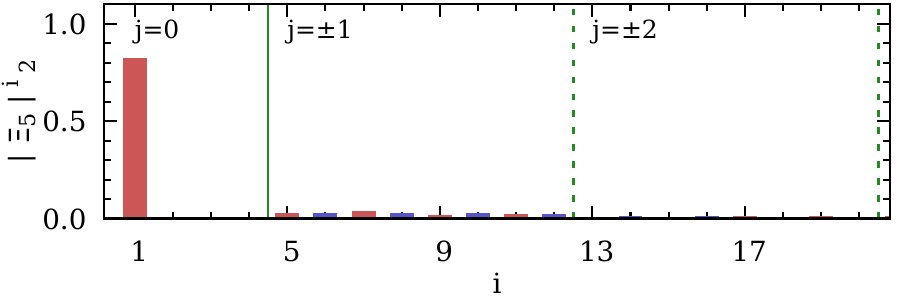}
    \label{sfig:qm1-leak-xi5-f2}
  }
  \\
  \subfigure[$|\Gamma_5|^i_2 = |\Gamma_5( \lambda_i, \lambda_2 = -\lambda_1)|$]{
    \includegraphics[width=\linewidth]{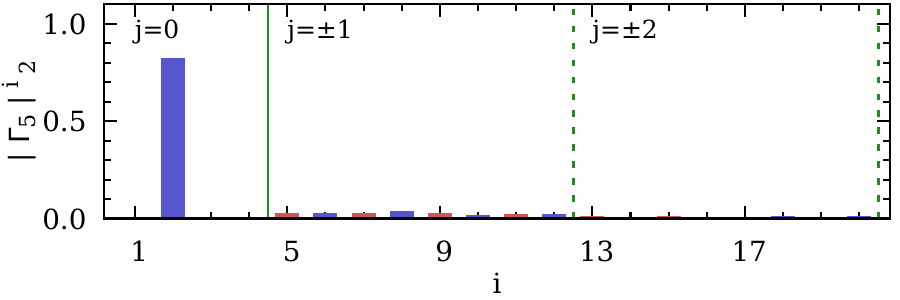}
    \label{sfig:qm1-leak-g5-f2}
  }
  \caption{Leakage pattern for would-be zero modes at $Q=-1$. Here,
    the red bar represents leakage to $\lambda_{i=2n-1} > 0$ with $i$
    odd, and the blue bar represents leakage to its parity partner
    $\lambda_{i=2n} = -\lambda_{2n-1}$ with $i$ even. }
  \label{fig:qm1-leak-f1}
\end{figure}
%------------------
% END of FIGURE 4.
%------------------
In Fig.~\ref{fig:qm1-leak-f1}, we present the leakage pattern of
the zero mode of $\lambda_1$ and its parity partner $\lambda_2 =
-\lambda_1$.
Since $Q=-1$ in Fig.~\ref{fig:qm1-leak-f1}, we expect to observe
four-fold degenerate would-be zero modes within a single quartet
(quartet index $j=0$).
\begin{align}
\lim_{a \to 0} \lambda_i = 0 \qquad \text{for $i=1,2,3,4$} \,.
\end{align}
In the continuum limit ($a = 0$), the $SU(4)$ taste symmetry becomes
exactly conserved and so would-be zero modes become exact zero modes.
However, at finite lattice spacing $a \ne 0$, the gauge configuration
is sufficiently rough that would-be zero modes have non-zero
eigenvalues: $\lambda_2 = -\lambda_1$, $\lambda_4 = -\lambda_3 $, and
$\lambda_1 \ne \lambda_3$ for $\lambda_1, \lambda_3 > 0$.
%

%-------------
% TABLE 3.
%-------------
% \input{table/tab_wi_1}
%-------------
\begin{table}[!tbhp]
  \caption{ Numerical values for leakage patterns from the $\lambda_1$
    eigenstate to the $\lambda_i$ eigenstate in
    Fig.~\ref{fig:qm1-leak-f1}. Here, $j$ represents a quartet index
    for the $\lambda_i$ eigenstate. The leakage represents
    $|\mathcal{O}|^i_1 = |\mathcal{O}(\lambda_i, \lambda_1)| = |
    \langle f_i | \mathcal{O} | f_1 \rangle |$ for $\mathcal{O} =
    \Gamma_5, \Xi_5$. }
  \label{tab:wi-1}
  \renewcommand{\arraystretch}{1.2}
  \begin{ruledtabular}
    \begin{tabular}{r | c | l | l }
      $j$ & leakage & value & Ward id. \\ \hline
      0   &  $| \Gamma_5 |^1_1$ & 0.82382566818582
          & $= |\Xi_5|^2_1$ \\
      0   & $| \Xi_5|^2_1$ & 0.82382566818581
          & $= |\Xi_5|^1_2$ \\
      0   & $| \Xi_5|^1_2$ & 0.82382566818580
          & $= |\Gamma_5|^2_2$  \\
      0   & $| \Gamma_5|^2_2$ & 0.82382566818579
          & $= |\Gamma_5|^1_1$
      \\ \hline
      0   & $| \Gamma_5|^2_1$ & $6.67 \times 10^{-4}$
          & \\
      0   & $| \Gamma_5|^3_1$ & $1.34 \times 10^{-3}$
          & \\
      0   & $| \Gamma_5|^4_1$ & $1.79 \times 10^{-3}$
          &
      \\ \hline
      $+1$  & $| \Gamma_5 |^5_1$ & $2.56 \times 10^{-2}$
          & \\
      $-1$  & $| \Gamma_5 |^6_1$ & $2.54 \times 10^{-2}$
          &
      \\ \hline
      $+2$  & $| \Gamma_5 |^{13}_1$ & $5.77 \times 10^{-3}$
          & \\
      $-2$  & $| \Gamma_5 |^{14}_1$ & $1.18 \times 10^{-2}$
          & \\
    \end{tabular}
  \end{ruledtabular}
\end{table}
%----------------
% END of TABLE 3.
%----------------
In Fig.~\ref{fig:qm1-leak-f1}\,\subref{sfig:qm1-leak-g5-f1}, we show
the leakage pattern of $|\Gamma_5|^i_1 =
|\Gamma_5(\lambda_i,\lambda_1)| = | \langle f_i | \Gamma_5 | f_1
\rangle |$.
We find that there is, in practice, no leakage, and so the only
non-zero component is $|\Gamma_5|^1_1 =
|\Gamma_5(\lambda_1,\lambda_1)|$.
The other components are practically zero.
In Figs.~\ref{fig:qm1-leak-f1}\,\subref{sfig:qm1-leak-xi5-f1},
\ref{fig:qm1-leak-f1}\,\subref{sfig:qm1-leak-xi5-f2}, and
\ref{fig:qm1-leak-f1}\,\subref{sfig:qm1-leak-g5-f2}, we find that the
Ward identity of Eqs.~\eqref{eq:wi-12} and \eqref{eq:wi-13} is
well-respected by the numerical results.
In other words, the Ward identity $|\Gamma_5|^1_1 = |\Xi_5|^2_1 =
|\Xi_5|^1_2 = |\Gamma_5|^2_2$ is satisfied within the numerical
precision of the computer.
Please refer to Table \ref{tab:wi-1} for numerical details.
This confirms that the theoretical prediction from the Ward identity
in Eqs.~\eqref{eq:wi-12} and \eqref{eq:wi-13} is correct.

In Fig.~\ref{fig:qm1-leak-f1}\,\subref{sfig:qm1-leak-g5-f1}, we find
that there is small leakage into other quartets ($j=\pm 1, \pm 2$).
The typical size of leakage between off-diagonal elements of the
would-be zero modes, the $j=0$ quartet, (\textit{e.g.}
$|\Gamma_5|^3_1$) is of order $10^{-3}$.
We also observe small leakage patterns of order $10^{-2} \sim 10^{-3}$
from the would-be zero modes, the $j=0$ quartet, to the non-zero
modes, the $j=\pm 1, \pm 2$ quartets (\textit{e.g.} $|\Gamma_5|^5_1$).

Now let us consider non-zero modes in the $j=+1$ quartet.
In Fig.~\ref{fig:qm1-leak-f5}, we present the leakage pattern for the
non-zero modes of $\lambda_5$ and its parity partner $\lambda_6 =
-\lambda_5$.
Even in the continuum limit ($a=0$), $\lambda_5 \ne 0$, and so it is a
non-zero mode.
Thanks to the approximate $SU(4)$ taste symmetry and the exact
$U(1)_A$ axial symmetry, there will be eight-fold degeneracy in the
family of eight eigenstates composed of the $j=+1$ quartet, to which
$\lambda_5$ belongs, and the $j=-1$ quartet (the parity partners).
These eight-fold degenerate modes are designated together as $j=\pm 1$
quartets in Fig.~\ref{fig:qm1-leak-f5}, where they are a set of
$\{\lambda_i\}$ with $5 \le i \le 12$.
Let us scrutinize the leakage pattern of the non-zero mode $\lambda_5
= \lambda_{j=+1,m=1}$.
In Fig.~\ref{fig:qm1-leak-f5}\,\subref{sfig:qm1-leak-g5-f5}, first
note that there is practically no leakage in the $\Gamma_5$ chirality
measurement from $\lambda_5$ into $\lambda_{2n-1}$ with $n>0$.
In other words, $|\Gamma_5|^{2n-1}_5 = |\Gamma_5(\lambda_{2n-1},
\lambda_5) | \cong 0$.
This implies that the chirality operator on the non-zero mode with
$\lambda > 0$ leaks into only the parity partner modes with $\lambda <
0$.
Second, note that the nontrivial leakage goes to those eigenstates in
the $j=-1$ quartet such as $\{\lambda_6, \lambda_8, \lambda_{10},
\lambda_{12} \} = \{ \lambda_{j,m} |\ j=-1,\ m=1,2,3,4\} $.
In addition, we find that the Ward identity of Eqs.~\eqref{eq:wi-12}
and \eqref{eq:wi-13} is well-respected within the numerical precision
in Figs.~\ref{fig:qm1-leak-f5}\,\subref{sfig:qm1-leak-g5-f5},
\ref{fig:qm1-leak-f5}\,\subref{sfig:qm1-leak-xi5-f5},
\ref{fig:qm1-leak-f5}\,\subref{sfig:qm1-leak-xi5-f6}, and
\ref{fig:qm1-leak-f5}\,\subref{sfig:qm1-leak-g5-f6}.
In Table \ref{tab:wi-2}, we present numerical values of
the $|\Gamma_5|^i_5$ shown in
Fig.~\ref{fig:qm1-leak-f5}\,\subref{sfig:qm1-leak-g5-f5} .
%----------
% FIGURE 5.
%----------
%\input{figs/fig_qm1_leak_f5}
%----------
\begin{figure}[t]
%  \vspace*{-5mm}
%  \centering
%  \footnotesize
%  \renewcommand{\arraystretch}{1.2}
%  \renewcommand{\subfigcapskip}{-0.55em}
%
%  \hspace*{-7mm}
  \subfigure[$|\Gamma_5|^i_5 = |\Gamma_5 (\lambda_i, \lambda_5)| $]{
    \includegraphics[width=\linewidth]{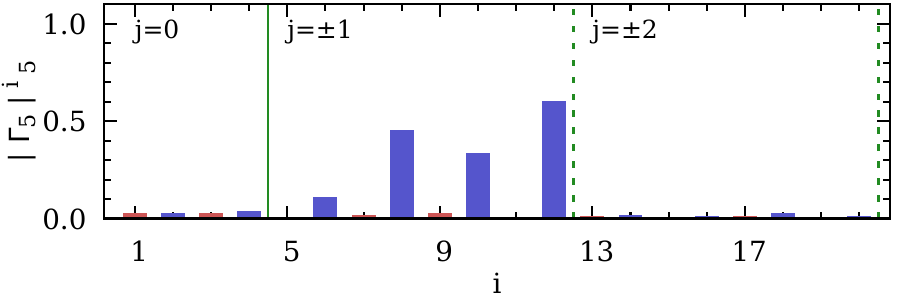}
    \label{sfig:qm1-leak-g5-f5}
  }
  \\
  \subfigure[$|\Xi_5|^i_5 = |\Xi_5 (\lambda_i, \lambda_5)| $]{
    \includegraphics[width=\linewidth]{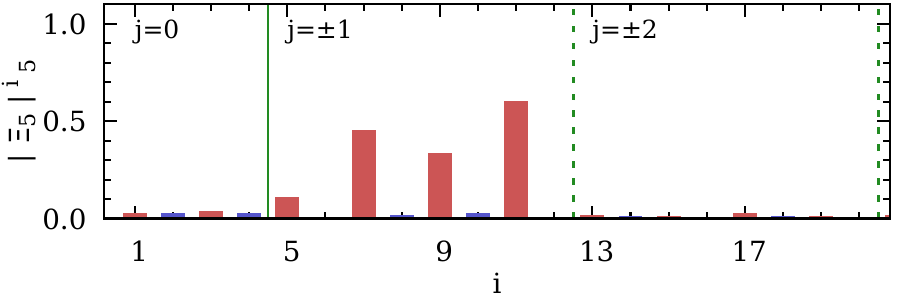}
    \label{sfig:qm1-leak-xi5-f5}
  }
  \\
  \subfigure[$|\Xi_5|^i_6 = |\Xi_5 (\lambda_i, \lambda_6 = -\lambda_5)| $]{
    \includegraphics[width=\linewidth]{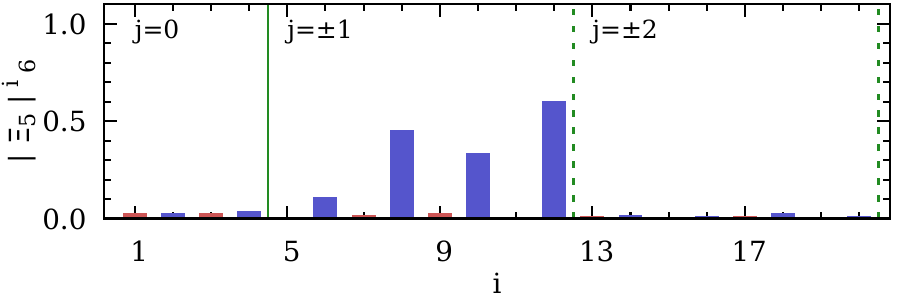}
    \label{sfig:qm1-leak-xi5-f6}
  }
  \\
  \subfigure[$|\Xi_5|^i_6 = |\Gamma_5 (\lambda_i, \lambda_6 = -\lambda_5)|$]{
    \includegraphics[width=\linewidth]{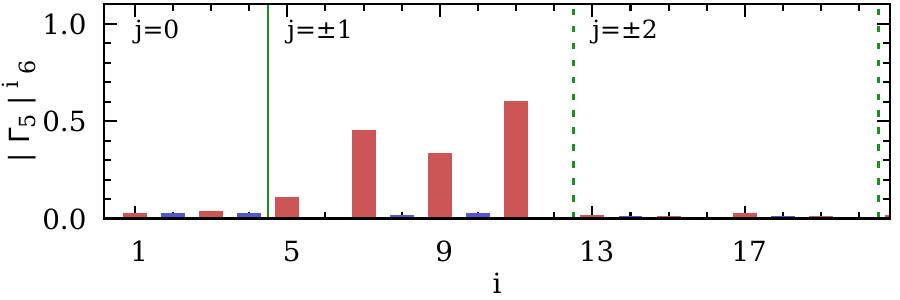}
    \label{sfig:qm1-leak-g5-f6}
  }
  \caption{Leakage pattern for non-zero modes at $Q=-1$.}
  \label{fig:qm1-leak-f5}
\end{figure}
%------------------
% END of FIGURE 5.
%------------------
%
%------------------
% TABLE 4
%------------------
%\input{table/tab_wi_2}
%------------------
\begin{table*}[!tbhp]
  \caption{ Numerical values for data in Fig.~\ref{fig:qm1-leak-f5}. }
  \label{tab:wi-2}
  \renewcommand{\arraystretch}{1.2}
  \begin{ruledtabular}
    \begin{tabular}{r c c l}
       $j$ & leakage & value & Ward identities \\
      \hline
      $-1$  & $|\Gamma_5|^6_5$ & 0.110
      & $= |\Xi_5|^5_5 = |\Xi_5|^6_6 = |\Gamma_5|^5_6$
      \\
      $-1$  & $|\Gamma_5|^8_5$ & 0.452
      & $= |\Xi_5|^7_5 = |\Xi_5|^8_6 = |\Gamma_5|^7_6
      = |\Gamma_5|^5_8 = |\Xi_5|^5_7 = |\Xi_5|^6_8 = |\Gamma_5|^6_7$
      \\
      $-1$  & $|\Gamma_5|^{10}_5$ & 0.334
      & $= |\Xi_5|^9_5 = |\Xi_5|^{10}_6 = |\Gamma_5|^9_6
      = |\Gamma_5|^5_{10} = |\Xi_5|^5_9 = |\Xi_5|^6_{10}
      = |\Gamma_5|^6_9$
      \\
      $-1$  & $|\Gamma_5|^{12}_5$ & 0.601
      & $= |\Xi_5|^{11}_5 = |\Xi_5|^{12}_6 = |\Gamma_5|^{11}_6
      = |\Gamma_5|^5_{12} = |\Xi_5|^5_{11} = |\Xi_5|^6_{12}
      = |\Gamma_5|^6_{11}$ \\
      \hline
      $+1$  & $|\Gamma_5|^5_5$ & $2.05 \times 10^{-3}$
      & $= |\Xi_5|^6_5 = |\Xi_5|^5_6 = |\Gamma_5|^6_6$
      \\
      $+1$  & $|\Gamma_5|^7_5$ & $16.7 \times 10^{-3}$
      & $= |\Xi_5|^8_5 = |\Xi_5|^7_6 = |\Gamma_5|^8_6
      = |\Gamma_5|^5_7 = |\Xi_5|^5_8 = |\Xi_5|^6_7
      = |\Gamma_5|^6_8$
      \\
      $+1$  & $|\Gamma_5|^9_5$ & $25.6 \times 10^{-3}$
      & $= |\Xi_5|^{10}_5 = |\Xi_5|^9_6 = |\Gamma_5|^{10}_6
      = |\Gamma_5|^5_9 = |\Xi_5|^5_{10} = |\Xi_5|^6_9
      = |\Gamma_5|^6_{10}$
      \\
      $+1$  & $|\Gamma_5|^{11}_5$ & $7.32 \times 10^{-3}$
      & $= |\Xi_5|^{12}_5 = |\Xi_5|^{11}_6 = |\Gamma_5|^{12}_6
      = |\Gamma_5|^5_{11} = |\Xi_5|^5_{12} = |\Xi_5|^6_{11}
      = |\Gamma_5|^6_{12}$ \\
      \hline
      $0$   & $|\Gamma_5|^3_5$ & $2.52 \times 10^{-2}$ &
      \\
      $0$   & $|\Gamma_5|^4_5$ & $3.43 \times 10^{-2}$ &
      \\
      \hline
      $+2$   & $|\Gamma_5|^{13}_5$ & $1.02 \times 10^{-2}$ &
      \\
      $-2$   & $|\Gamma_5|^{14}_5$ & $1.38 \times 10^{-2}$ &
      \\
    \end{tabular}
  \end{ruledtabular}
\end{table*}
%----------------
% END of TABLE 4
%----------------

Let us examine the $\Gamma_5 =[\gamma_5 \otimes 1]$ leakage pattern of
the $j=+1$ quartet of the non-zero modes $\{\lambda_5, \lambda_7,
\lambda_9, \lambda_{11}\}$.
In Fig.~\ref{fig:qm1-leak-g5-f5set}, we find that the chirality
measurement vanishes; $(\Gamma_5)^i_i = \Gamma_5(\lambda_i, \lambda_i)
= 0$ for $\lambda_i$ in the $j=+1$ quartet of non-zero modes.
We also find that the $\Gamma_5$ leakage of $\lambda_{+1,m} > 0$ of
the $j=+1$ quartet goes to the parity partners with $\lambda_{-1,m'} <
0$ of the $j=-1$ quartet, and the leakage to other quartets such as
$j=\pm2$ is negligibly small compared to the leakage to the $j=-1$
quartet.
The numerical values of $|\Gamma_5|^{-1,m}_{+1,m'}$ are summarized in
Table \ref{tab:leak_g5}.
%---------
% FIGURE 6
%---------
%\input{figs/fig_qm1_leak_g5_f5set}
%---------
\begin{figure}[t]
%  \vspace*{-5mm}
%  \centering
%  \footnotesize
%  \renewcommand{\arraystretch}{1.2}
%  \renewcommand{\subfigcapskip}{-0.55em}
%
%  \hspace*{-7mm}
  \subfigure[$|\Gamma_5|^i_5 = |\Gamma_5 (\lambda_i, \lambda_5)| $]{
    \includegraphics[width=\linewidth]{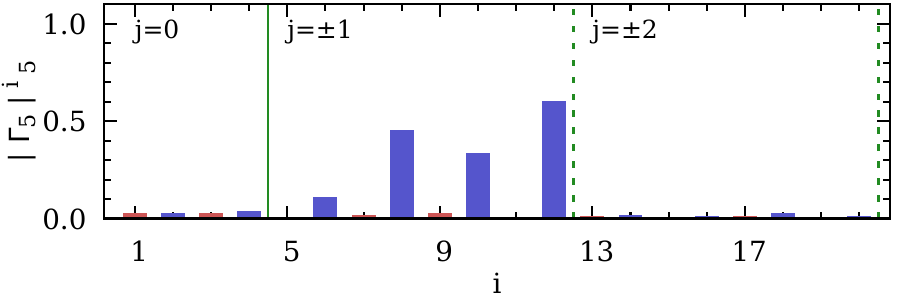}
    \label{sfig:qm1-leak-g5-f5-2}
  }
  \\
  \subfigure[$|\Gamma_5|^i_7 = |\Gamma_5 (\lambda_i, \lambda_7)| $]{
    \includegraphics[width=\linewidth]{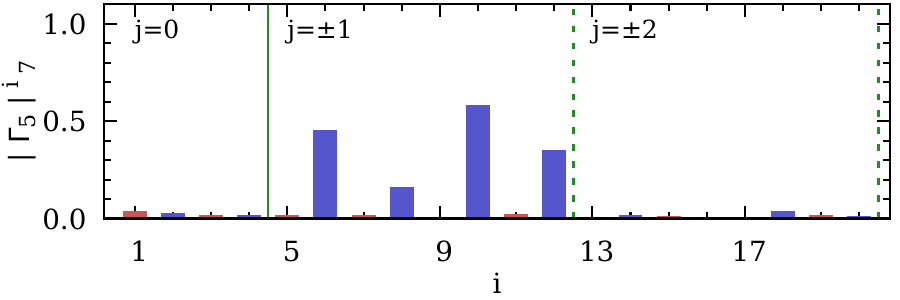}
    \label{sfig:qm1-leak-g5-f7}
  }
  \\
  \subfigure[$|\Gamma_5|^i_9 = |\Gamma_5 (\lambda_i, \lambda_9)| $]{
    \includegraphics[width=\linewidth]{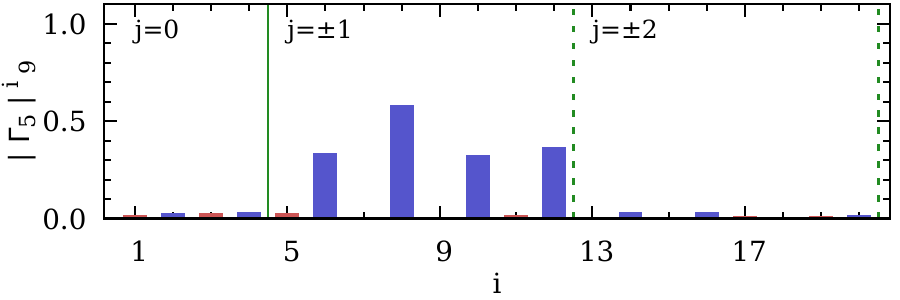}
    \label{sfig:qm1-leak-g5-f9}
  }
  \\
  \subfigure[$|\Gamma_5|^i_{11} = |\Gamma_5 (\lambda_i, \lambda_{11})|$]{
    \includegraphics[width=\linewidth]{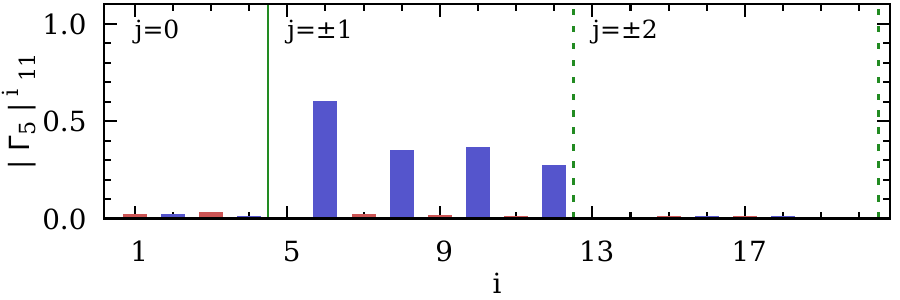}
    \label{sfig:qm1-leak-g5-f11}
  }
  \caption{$[\gamma_5 \otimes 1]$ leakage pattern for non-zero modes
    at $Q=-1$.}
  \label{fig:qm1-leak-g5-f5set}
\end{figure}

%-----------------
% END of FIGURE 6
%-----------------
%
%---------
% TABLE 5.
%---------
%\input{table/tab_leak_g5}
%---------
\begin{table}[!tbhp]
  \caption{ $| \Gamma_5 |^{-1,m}_{+1,m'}$ values in
      Fig.~\ref{fig:qm1-leak-g5-f5set}. }
  \label{tab:leak_g5}
  \renewcommand{\arraystretch}{1.2}
  \begin{ruledtabular}
    \begin{tabular}{c|cccc}
      \diagbox{$\lambda_i$}{$\lambda_j$} & $\lambda_5$ & $\lambda_7$ &
      $\lambda_9$ & $\lambda_{11}$ \\ \hline
      $\lambda_6$ & 0.110 & 0.452 & 0.334 & 0.601 \\
      $\lambda_8$ & 0.452 & 0.161 & 0.582 & 0.349 \\
      $\lambda_{10}$ & 0.334 & 0.582 & 0.323 & 0.366 \\
      $\lambda_{12}$ & 0.601 & 0.349 & 0.366 & 0.271 \\
    \end{tabular}
  \end{ruledtabular}
\end{table}
%----------------
% END of TABLE 5.
%----------------

Let us examine the $\Xi_5=[1 \otimes \xi_5]$ leakage pattern of the
$j=+1$ quartet of the non-zero modes $\{\lambda_5, \lambda_7,
\lambda_9, \lambda_{11}\}$.
In Fig.~\ref{fig:qm1-leak-xi5-f5set}, we find that the $\Xi_5$ leakage
from the $j=+1$ quartet to the $j=-1$ quartet (parity partners)
vanishes in practice.
Since the leakage pattern of $\Xi_5$ is related to the leakage
pattern of $\Gamma_5$ by the Ward identity
\begin{align}
  |\Xi_5|^{j,m}_{j',m'} &= |\Gamma_5|^{-j,m}_{j',m'} \,,
  \label{eq:ward_jm}
\end{align}
Fig.~\ref{fig:qm1-leak-xi5-f5set} can be obtained from
Fig.~\ref{fig:qm1-leak-g5-f5set} using the Ward identity.
We find that the $\Xi_5$ leakage from the $j=+1$ quartet to other
quartets such as $j=\pm 2$ quartets is negligibly small compared to
leakage to itself (the $j=+1$ quartet).
Leakage patterns of the $\Gamma_5$ chirality and $\Xi_5$ shift
operators for diverse topological charges are shown in Appendix
\ref{app:ex-nonzero}.
%------------
% FIGURE 7.
%------------
%\input{figs/fig_qm1_leak_xi5_f5set}
%------------
\begin{figure}[t]
%  \vspace*{-5mm}
%  \centering
%  \footnotesize
%  \renewcommand{\arraystretch}{1.2}
%  \renewcommand{\subfigcapskip}{-0.55em}
%
%  \hspace*{-7mm}
  \subfigure[$|\Xi_5|^i_5 = |\Xi_5 (\lambda_i, \lambda_5)| $]{
    \includegraphics[width=\linewidth]{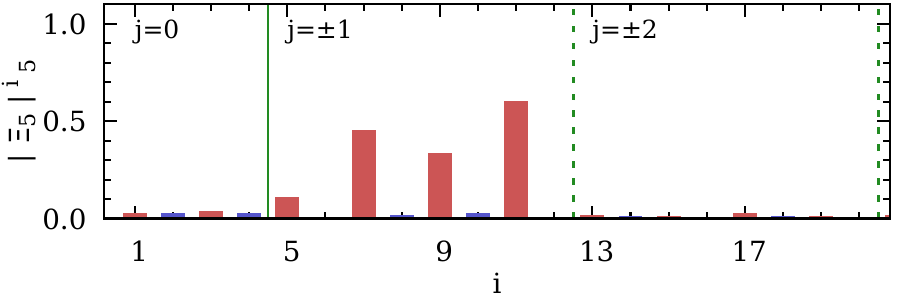}
    \label{sfig:qm1-leak-xi5-f5-2}
  }
  \\
  \subfigure[$|\Xi_5|^i_7 = |\Xi_5 (\lambda_i, \lambda_7)| $]{
    \includegraphics[width=\linewidth]{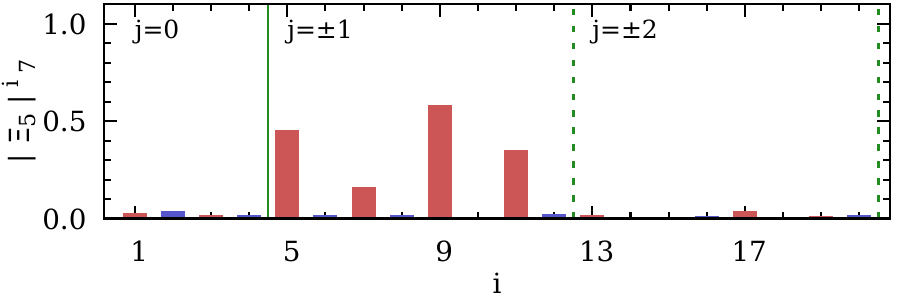}
    \label{sfig:qm1-leak-xi5-f7}
  }
  \\
  \subfigure[$|\Xi_5|^i_9 = |\Xi_5 (\lambda_i, \lambda_9)| $]{
    \includegraphics[width=\linewidth]{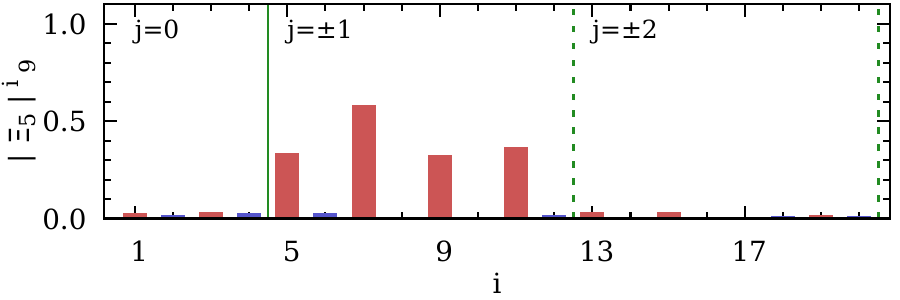}
    \label{sfig:qm1-leak-xi5-f9}
  }
  \\
  \subfigure[$|\Xi_5|^i_{11} = |\Xi_5 (\lambda_i, \lambda_{11})|$]{
    \includegraphics[width=\linewidth]{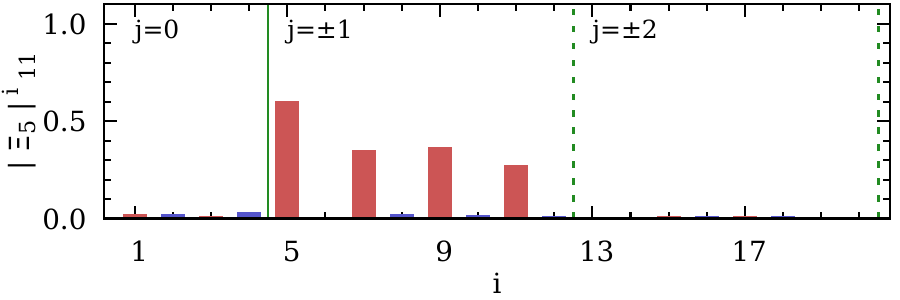}
    \label{sfig:qm1-leak-xi5-f11}
  }
  \caption{$[1 \otimes \xi_5]$ leakage pattern for non-zero modes
    at $Q=-1$.}
  \label{fig:qm1-leak-xi5-f5set}
\end{figure}
%------------------
% END of FIGURE 7.
%------------------

Let us summarize the leakage pattern for would-be zero modes and that
for non-zero modes.
We first begin with the leakage pattern for the zero modes.
\begin{enumerate}
\item A zero mode of staggered fermions appears as a four-fold
  degenerate quartet. In other words, for the topological charge $Q$,
  the number of zero modes is $4 \times (n_+ + n_-)\,,$ and $Q = n_- -
  n_+$ (Atiyah-Singer Index Theorem), where $n_+$ ($n_-$) is the
  number of right-handed (left-handed) zero mode quartets.

\item In the chirality $\Gamma_5 = [\gamma_5 \otimes 1]$ measurement,
  the zero mode has practically no leakage to other eigenstates.

\item In the shift $\Xi_5=[1 \otimes \xi_5]$ measurement, the zero
  mode with eigenvalue $\lambda$ has a full (100\%) leakage into its
  parity partner mode with eigenvalue $-\lambda$, and no leakage into
  any other eigenmodes.
  
\end{enumerate}

The leakage pattern for nonzero modes is 
\begin{enumerate}
\item A non-zero mode of staggered fermions appears as an eight-fold
  degeneracy composed of a quartet ($+j$ quartet) and its parity
  partner quartet ($-j$ quartet). In other words, non-zero eigenmodes
  can be grouped into sets with eight elements in each set.  This is
  due to the approximate $SU(4)$ taste symmetry and the conserved
  $U(1)_A$ axial symmetry.

\item In the chirality $\Gamma_5=[\gamma_5 \otimes 1]$ measurement,
  the non-zero mode with eigenvalue $\lambda_{j,m}$ has no leakage to
  its own quartet ($j$ quartet), but has leakage only to the parity
  partner ($-j$ quartet) with $\lambda_{-j,m'}$. It has no leakage to
  any eigenmode which belongs to other quartets with $\ell \ne \pm j$.

\item In the shift $\Xi_5=[1 \otimes \xi_5]$ measurement, the non-zero
  mode with $\lambda_{j,m}$ has no leakage to its parity partner ($-j$
  quartet) at all. But it has leakage only to the eigenstates in its
  own $(+j)$ quartet. This pattern comes directly from the Ward
  identity.  In other words, the $\Xi_5$ leakage pattern is a mirror
  image reflecting $\Gamma_5$ through the mirror of the Ward identity.
  $\Xi_5$ has no leakage to any eigenmode which belongs to other
  quartets with $\ell \ne \pm j$.

\item Thanks to the Ward identity of the conserved $U(1)_A$ symmetry,
  the leakage pattern of $|\Gamma_5|^{-j,m}_{\ell,m'}$ is identical to
  that of $|\Xi_5|^{+j,m}_{\ell,m'}$.
  
\end{enumerate}
In Appendix \ref{app:ex-zero}, we provide more examples to demonstrate
our claim that the leakage pattern for zero modes holds in general.
In Appendix \ref{app:ex-nonzero}, we give more examples to demonstrate
our claim that the leakage pattern for non-zero modes is valid in
general.
We have repeated numerical tests over hundreds of zero modes and tens
of thousands of nonzero modes.
We performed the numerical study on hundreds of gauge configurations
and find that the above leakage pattern is valid for all of them
except for those gauge configurations with unstable topological
charge.
\begin{enumerate}
\item We find a number of gauge configurations which do not have a
  stable topological charge.

\item We find about 10 gauge configurations with unstable topological
  charge among 100 gauge configurations with $12^4$ lattice geometry
  at $\beta=4.6$.

\item We find about 8 gauge configurations with unstable topological
  charge among 300 gauge configurations with $20^4$ lattice geometry
  at $\beta=5.0$.
  
\end{enumerate}

%---------
% TABLE 6.
%---------
%\input{table/tab_t5}
%---------
\begin{table}[!tbhp]
  \caption{Numerical results for $T_5$. To obtain the results, we use
    292 gauge configurations with the input parameters in Table
    \ref{tab:in-para}. $N_q$ represents the number of quartets used to
    obtain the statistical error. Here $j=0$ represents would-be zero
    modes, and $j >0$ represents non-zero modes.  }
  \label{tab:t5}
  \renewcommand{\arraystretch}{1.2}
  \begin{ruledtabular}
    \begin{tabular}{cccr}
      $j$ & $|\Re(T_5)|$ & $|\Im(T_5)|$ & $N_q$
      \\ \hline
      $j=0$ & $7.2(130) \times 10^{-4}$ & $5.9(46) \times 10^{-12}$ & $490$
      \\
      $j>0$ & $6.2(120) \times 10^{-3}$ & $3.3(25) \times 10^{-12}$ & $7034$
      \\
    \end{tabular}
  \end{ruledtabular}
\end{table}
%----------------
% END of TABLE 6.
%----------------
In Table \ref{tab:t5}, we present results for $T_5$ defined in
Eq.~\eqref{eq:t5}, which is a direct barometer to estimate the effect
of taste symmetry breaking.
If the taste symmetry is exactly conserved, then $T_5$ must vanish.
Hence, a non-trivial value of $T_5$ indicates the size of taste
symmetry breaking.
In Table \ref{tab:t5}, we find that $|\Re(T_5)|$ is of order
$10^{-3}$, while $|\Im(T_5)|$ is essentially zero.
This indicates that the effect of taste symmetry breaking is very
small (in the sub-percent level within each quartet).

%-----------
% FIGURE 8.
%-----------
%\input{figs/fig_s5}
%-----------
\begin{figure}[h]
  \centering
  \includegraphics[width=\linewidth]{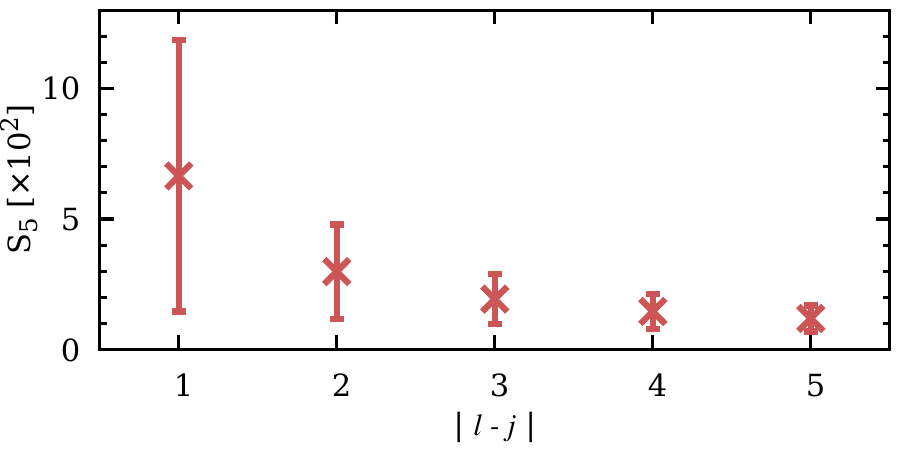}
  \caption{$S_5$ as a function of $|\ell - j|$. Numerical values are
  given in Table \ref{tab:s5}.}
  \label{fig:s5}
\end{figure}
%-----------------
% END of FIGURE 8.
%-----------------
In Fig.~\ref{fig:s5}, we present $S_5$ defined in Eq.~\eqref{eq:s5}
as a function of $|\ell - j|$ with $\ell,\, j \ge 0$.
Here $|\ell - j| = 1$ represents a pair of nearest neighbor quartets,
$|\ell - j| = 2$ represents a pair of next-to-nearest neighbor
quartets, and so on.
The values of $S_5$ are the same size as the statistical errors.
This indicates that the taste symmetry breaking results in simply
random noise added to the physical signal ($S_5 = 0$).
For $|\ell - j| = 1$, the noise is $\approx$ 7\%, and for $|\ell - j|
= 2$, the noise is $\approx$ 3\%.
We find that the noise decreases as $|\ell - j|$ increases.
The numerical values of $S_5$ in Fig.~\ref{fig:s5} are presented in
Table \ref{tab:s5}.
%---------
% TABLE 7.
%---------
%\input{table/tab_s5}
%---------
\begin{table}[tbhp]
  \caption{Numerical results for $S_5$. Here, we measure $S_5$ between
    two different quartets ($\ell \ne j$ and $\ell, j \ge 0$). $N_p$
    represents the number of ($\ell$, $j$) pairs with $\ell \ne j$.}
  \label{tab:s5}
  \renewcommand{\arraystretch}{1.2}
  \begin{ruledtabular}
    \begin{tabular}{ccc}
      $|\ell-j|$ & $S_5$ & $N_p$
      \\ \hline
      1 & $6.6(52) \times 10^{-2}$ & $7185$
      \\
      2 & $3.0(18) \times 10^{-2}$ & $6893$
      \\
      3 & $1.9(10) \times 10^{-2}$ & $6601$
      \\
      4 & $1.5(7) \times 10^{-2}$ & $6309$
      \\
      5 & $1.2(5) \times 10^{-2}$ & $6017$
    \end{tabular}
  \end{ruledtabular}
\end{table}
%
%----------------
% END of TABLE 7.
%----------------

%------------------
% END of SECTION 5.
%------------------

%------------
% SECTION 6.
%------------
%\input{sec.machine.tex}
%------------
\section{Machine Learning}
\label{sec:machine:learn}
In previous sections, we have shown that the $U(1)_A$ symmetry of
staggered fermions induces the chiral Ward identities in
Eq.~\eqref{eq:wi-12}, and we have also noted that the approximate
$SU(4)$ taste symmetry brings in the quartet behavior of the
eigenvalue spectrum.
Furthermore, a combined effect of those symmetries gives us distinctive
leakage patterns for the chirality operator $\Gamma_5$ and the shift
operator $\Xi_5$.
In this section, we apply a machine learning technique to the
following tasks.
\begin{enumerate}
\item We want to know how much the non-zero modes respect the quartet
  classification rules, which come from the $SU(4)$ taste symmetry.
\item We want to know how efficiently we can measure the topological
  charge $Q$ using the index theorem from the quartet structure of the
  non-zero modes.
\item We want to detect any anomalous behavior of the eigenvalue
  spectrum, which does not follow the standard leakage pattern of the
  non-zero modes.
\item We want to figure out what causes the anomalous behavior of the
  eigenvalue spectrum.
\end{enumerate}
%

%----------
% FIGURE 9.
%----------
%\input{figs/fig_g5x1_mat}
%----------
\begin{figure}[t]
  \centering
  \subfigure[$200 \times 200$]{
    \includegraphics[width=\linewidth]{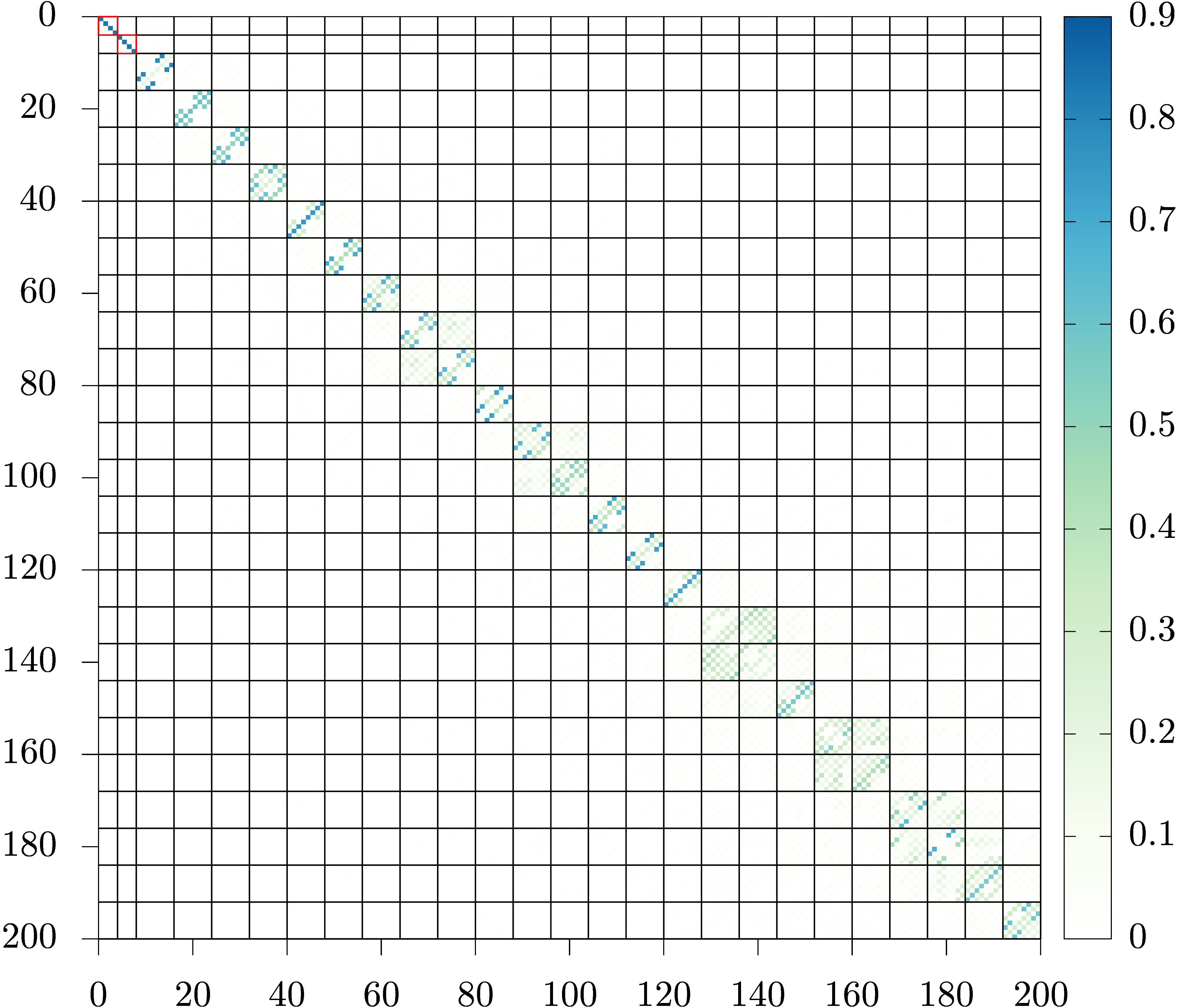}
    \label{sfig:g5x1_mat_200}
  }
  \\
  \subfigure[$32 \times 32$]{
    \includegraphics[width=\linewidth]{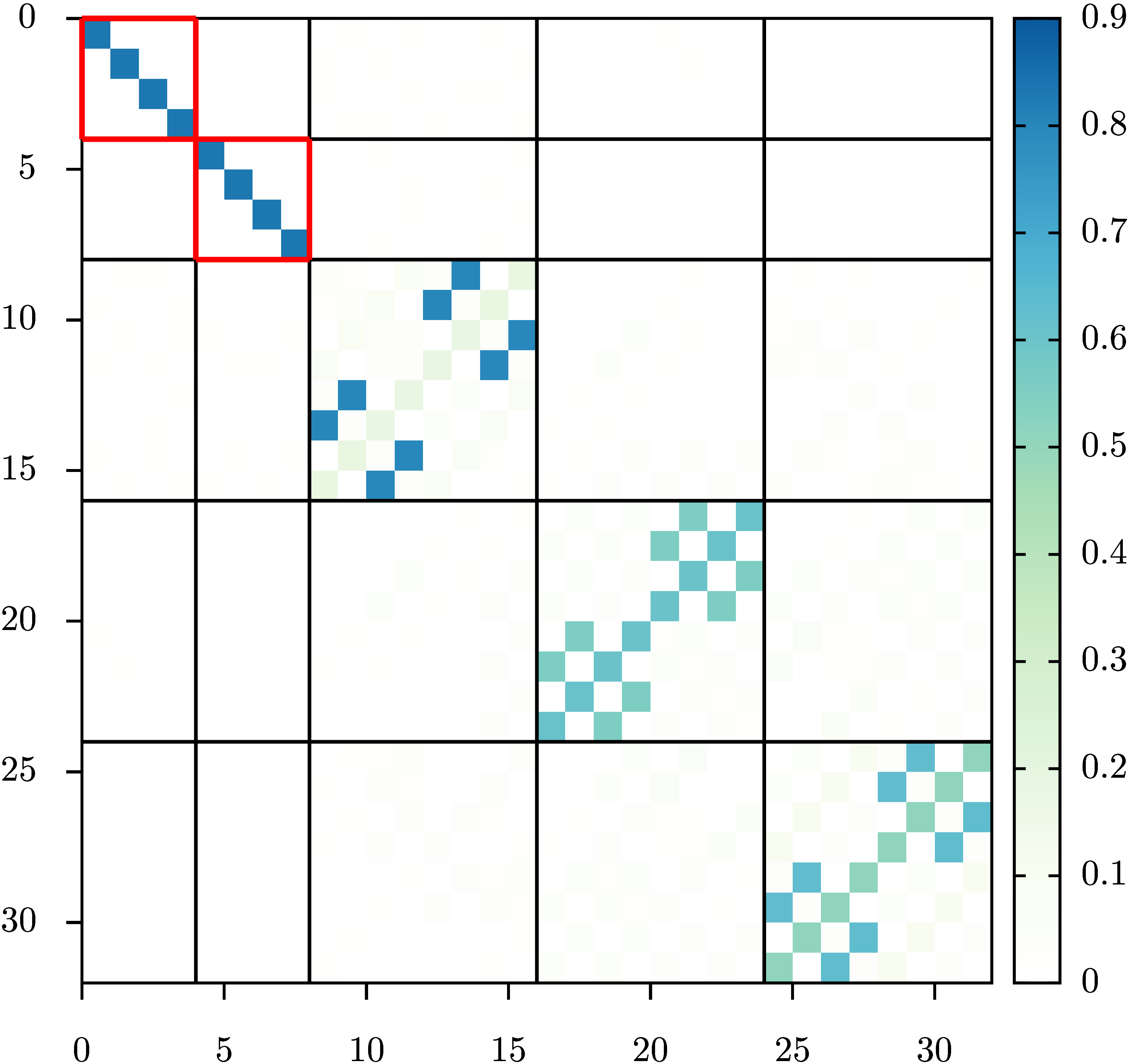}
    \label{sfig:g5x1_mat_32}
  }
  \caption{Matrix elements of $|\Gamma_5|$ for 200 and 32 of the
    lowest eigenmodes on a gauge configuration with $Q = 2$. Here,
    indices on both axes are the eigenvalue index. The color of each
    square represents the magnitude of the corresponding matrix
    element. Black lines indicate borders of non-zero mode quartets,
    and red lines are those of zero mode quartets. }
  \label{fig:g5x1_mat}
\end{figure}
%----------
% END of FIGURE 9.
%----------
%
%-----------
% FIGURE 10.
%-----------
%\input{figs/fig_label}
%-----------
\begin{figure}[t]
  \centering
  \subfigure[class 0]{
    \includegraphics[width=0.46\linewidth]{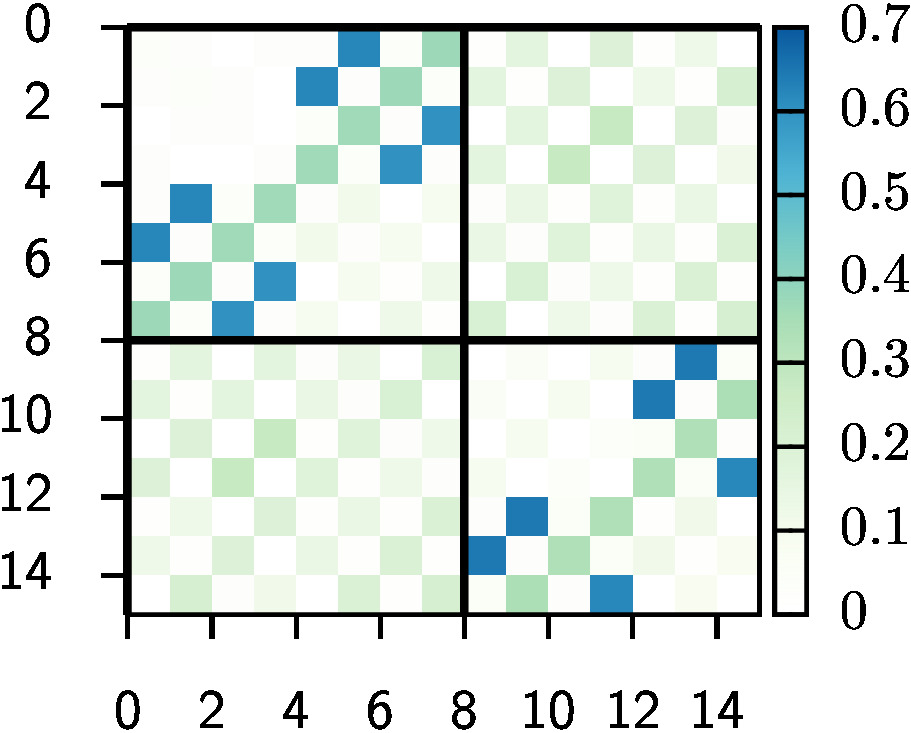}
    \label{sfig:label_0}
  }
  \subfigure[class 1]{
    \includegraphics[width=0.46\linewidth]{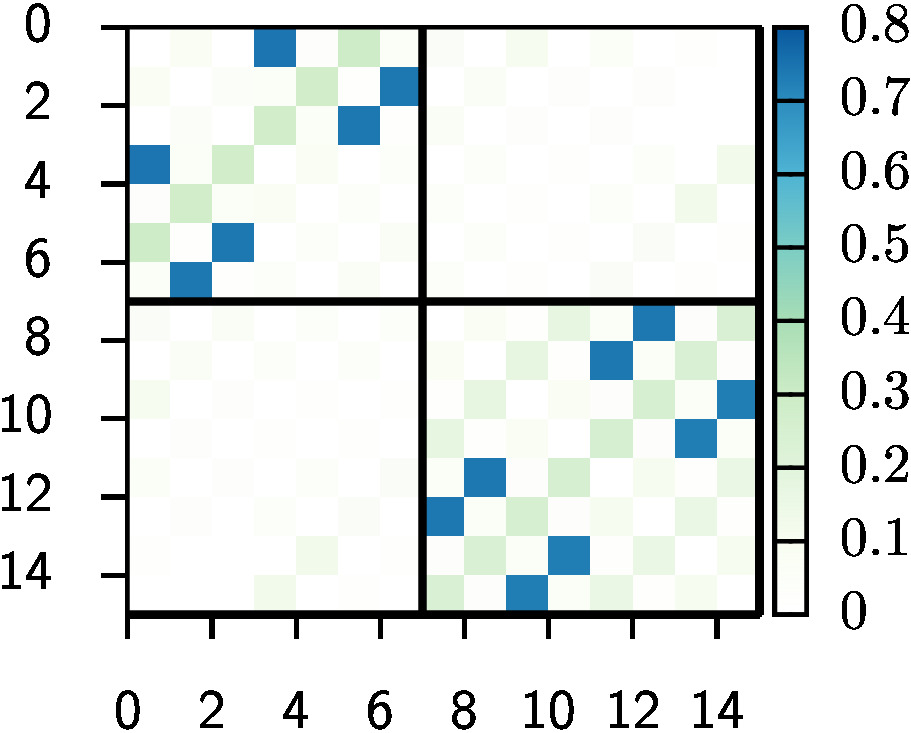}
    \label{sfig:label_1}
  }
  \\
  \subfigure[class 2]{
    \includegraphics[width=0.46\linewidth]{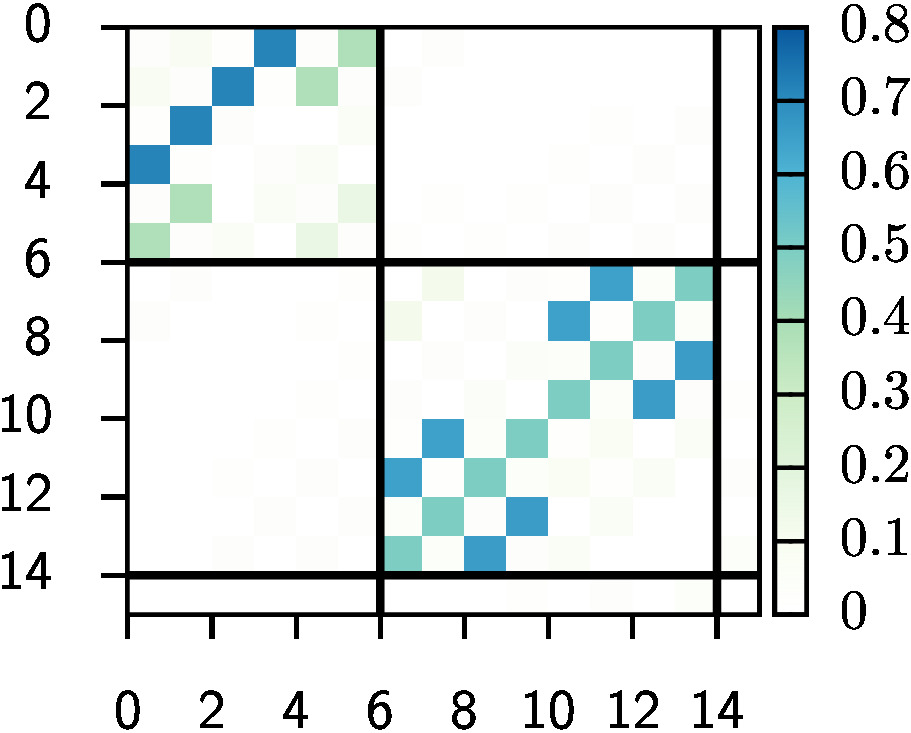}
    \label{sfig:label_2}
  }
  \subfigure[class 3]{
    \includegraphics[width=0.46\linewidth]{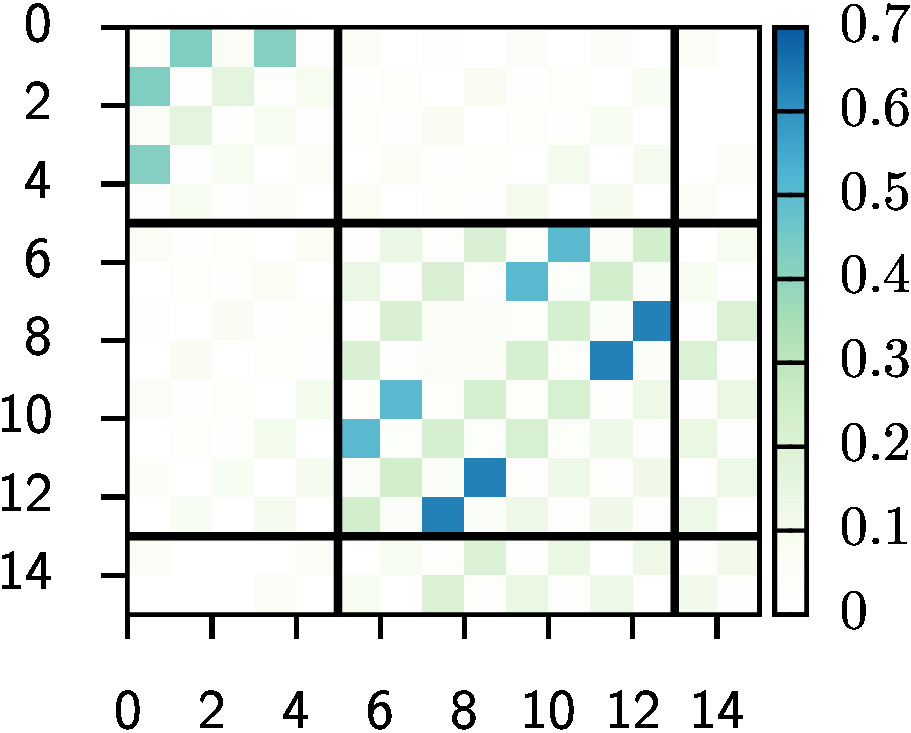}
    \label{sfig:label_3}
  }
  \\
  \subfigure[class 4]{
    \includegraphics[width=0.46\linewidth]{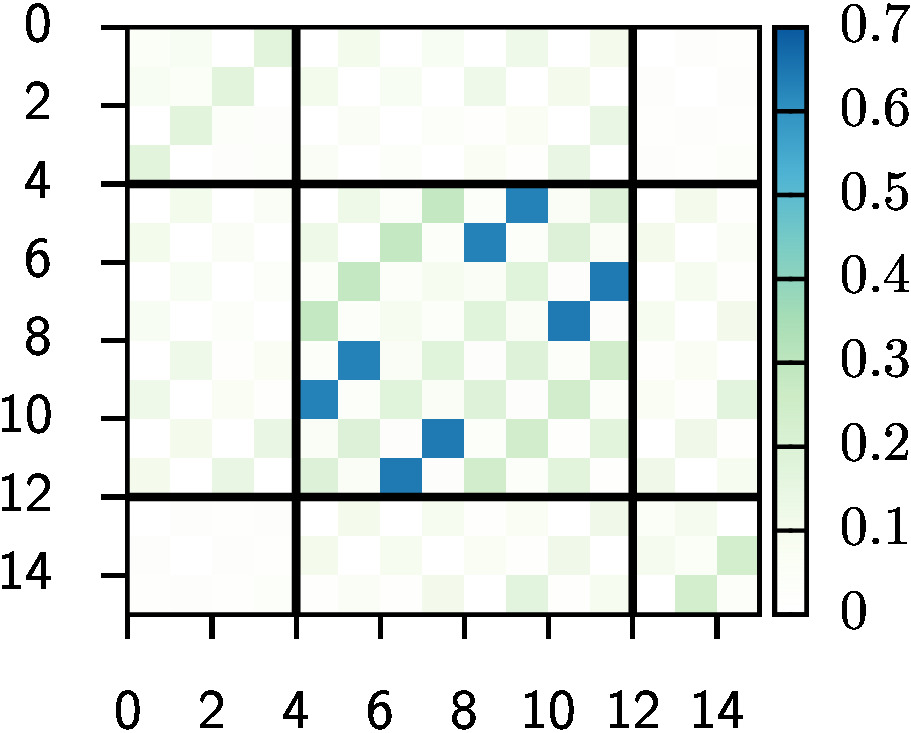}
    \label{sfig:label_4}
  }
  \subfigure[class 5]{
    \includegraphics[width=0.46\linewidth]{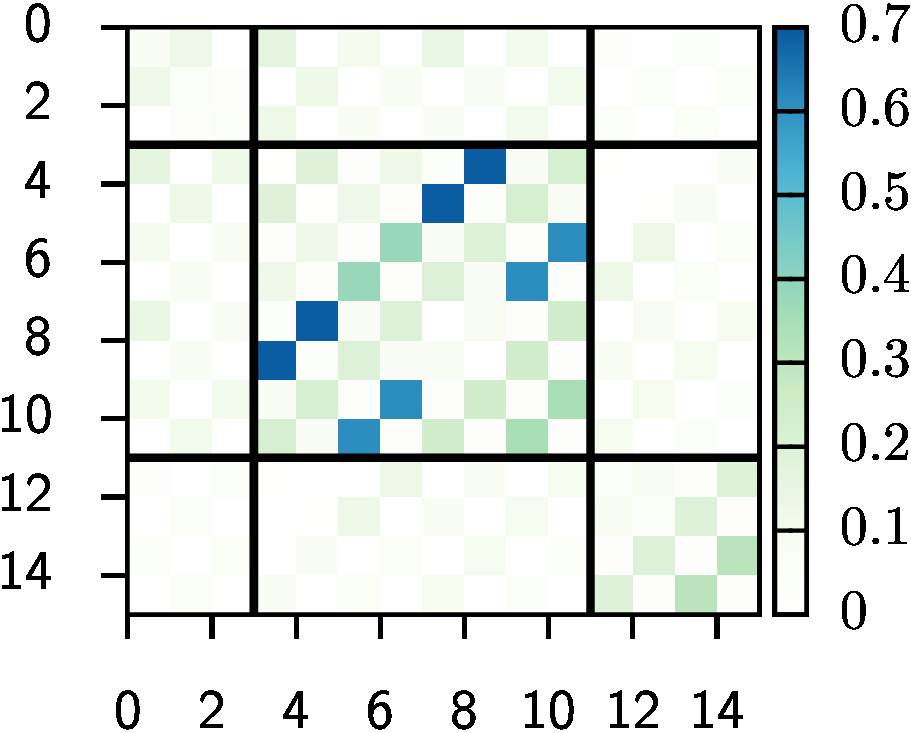}
    \label{sfig:label_5}
  }
  \\
  \subfigure[class 6]{
    \includegraphics[width=0.46\linewidth]{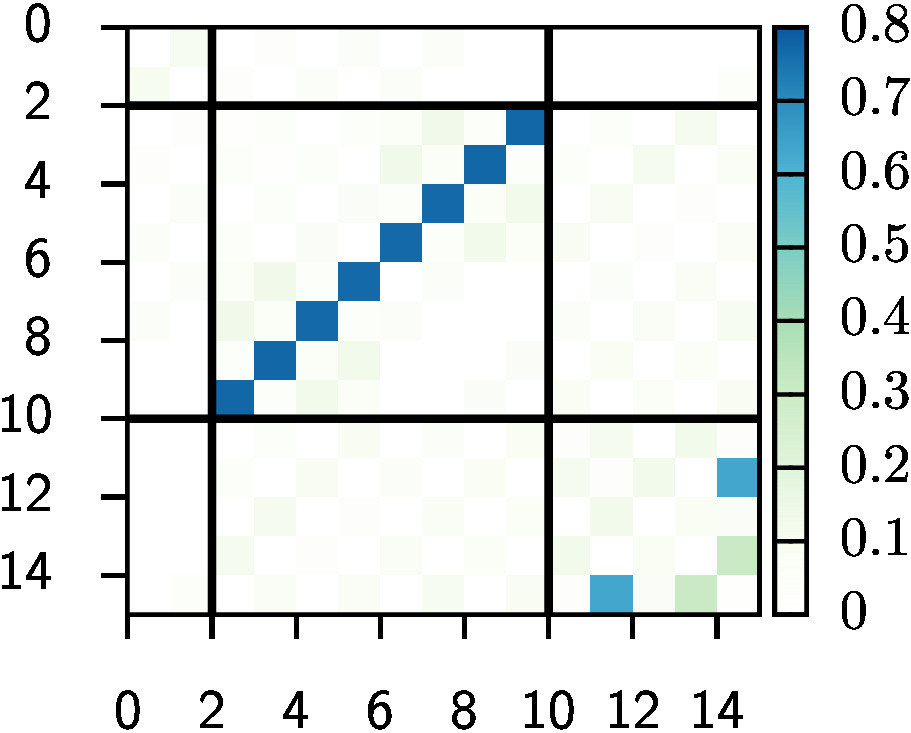}
    \label{sfig:label_6}
  }
  \subfigure[class 7]{
    \includegraphics[width=0.46\linewidth]{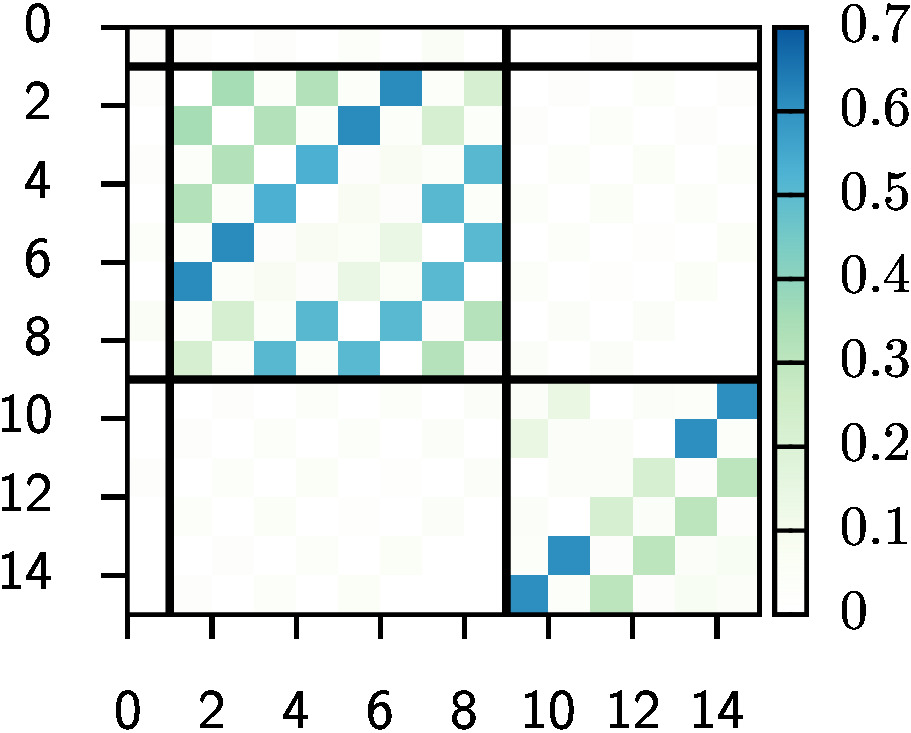}
    \label{sfig:label_7}
  }
  \caption{Examples for our samples. Every sample contains only one
    non-zero mode quartet. There are eight kinds of classes according to
    the location of the borders of the quartet.}
  \label{fig:label}
\end{figure}
%------------------
% END of FIGURE 10.
%------------------
Let us explain our sampling method for the machine learning.
In Fig.~\ref{fig:g5x1_mat}, we show matrix elements $|\Gamma_5|^i_j$ on a
gauge configuration with $Q = 2$.
Fig.~\ref{sfig:g5x1_mat_200} is for the $200$ lowest eigenmodes, and
Fig.~\ref{sfig:g5x1_mat_32} is a zoomed-in version of
Fig.~\ref{sfig:g5x1_mat_200} for the $32$ lowest eigenmodes.
Here the depth of the blue color represents the size of the matrix
element $|\Gamma_5|^i_j$, and $i,j$ run from zero to 199.
We identify two zero mode quartets (red boxes) by looking at the
magnitude of the diagonal components.
These two quartets have the same chirality ($n_- = 2$), which is
consistent with the topological charge $Q = 2$.
Excluding the would-be zero modes, we randomly choose a $15 \times 15$
sub-matrix of $|\Gamma_5|^i_j$ along the diagonal line of
$|\Gamma_5|^i_j$ matrix elements.
This $15 \times 15$ sub-matrix is the smallest square sub-matrix of
$|\Gamma_5|$ which contains all elements of only one quartet of
non-zero modes and its parity partner quartet.

In Fig.~\ref{fig:label}, we present 8 different classes for arbitrary
samples.
Our purpose for the machine learning is to find borders (black lines)
of the non-zero mode quartets (or octets when the parity partners are
included) in each sample.
We classify arbitrary samples into eight different classes according
to the location of the border lines.
Each class is labeled as in Fig.~\ref{fig:label}.
%

%---------
% TABLE 8
%---------
%\input{table/tab_ml_para}
%---------
\begin{table}[th]
  \caption{Parameters for machine learning.}
  \label{tab:ml_para}
  \renewcommand{\arraystretch}{1.2}
  \begin{ruledtabular}
    \begin{tabular}{r | l}
      parameters & values \\ \hline
      number of training configurations   & $120$  \\
      number of training samples          & $1223$ \\
      number of validation configurations & $30$   \\
      number of validation samples        & $308$  \\
      number of test configurations       & $142$  \\
      number of test samples              & $1448$ \\
      \hline
      \multirow{2}{*}{loss function} & categorical \\
      & cross-entropy \cite{Goodfellow-et-al-2016, chollet2015keras} \\
      optimization method & Adam \cite{Kingma:2014vow} \\
      activation function for hidden layers & ReLU
      \cite{Goodfellow-et-al-2016} \\
      activation function for output layer & softmax
      \cite{Goodfellow-et-al-2016} \\
    \end{tabular}
  \end{ruledtabular}
\end{table}
%---------------
% END of TABLE 8
%---------------
%
%---------
% TABLE 9
%---------
%\input{table/tab_nn_para}
%---------
\begin{table}[th]
  \caption{Hyper-parameters for neural networks. Here, we show one of
    the examples of best performance, in which we use only MLP but not
    CNN.}
  \label{tab:nn_para}
  \renewcommand{\arraystretch}{1.2}
  \begin{ruledtabular}
    \begin{tabular}{cccc}
      layer & type & number of units & activation \\
      \hline
      input & - & $225$ & - \\
      hidden \#1 & MLP & $160$ & ReLU \\
      hidden \#2 & MLP & $1210$ & ReLU \\
      hidden \#3 & MLP & $1490$ & ReLU \\
      output & MLP & $8$ & Softmax \\
    \end{tabular}
  \end{ruledtabular}
\end{table}
%---------------
% END of TABLE 9
%---------------
We use a deep learning model which combines the multi-layer perceptron
(MLP) \cite{Goodfellow-et-al-2016} and the convolutional neural
network (CNN) \cite{Goodfellow-et-al-2016}.
In Table \ref{tab:ml_para}, we present our basic setup for the machine
learning.
We use the gauge configuration ensemble described in Table
\ref{tab:in-para}.
The data measured over $292$ gauge configurations are distributed over
a training set, validation set, and test set as in Table \ref{tab:ml_para}.
For each gauge configuration, we generate around ten $15 \times 15$
matrix samples from the $200$ lowest eigenmodes without overlapping.
We make popular and suitable choices for the loss function\footnote{
  Popular and basic loss functions such as the mean squared error
  (MSE) and mean absolute error (MAE) are usually used for regression
  problems. On the contrary, the categorical cross-entropy loss
  function is most applicable to multi-class classification
  problems.},
optimization method\footnote{ Popular optimization methods available
  in the market are stochastic gradient descent, AdaGrad, RMSprop, and
  Adam \cite{Kingma:2014vow}.},
and activation functions\footnote{ Popular activation functions are
  the $\tanh$, sigmoid, and ReLU. Here, we make use of ReLU for the
  hidden layers since it is the simplest and fastest. The softmax
  function is essential for the output layer of the multi-class
  classification.}
relevant to our purpose, which are summarized in Table
\ref{tab:ml_para}.
The best hyper-parameters such as the number of layers and the number
of units for each layer are determined by using the Keras Tuner
\cite{chollet2015keras}.
The accuracy of classification per gauge configuration is obtained by
averaging the accuracies of the machine learning (ML) prediction for
all the samples on a single gauge configuration.
Our best model achieves an average accuracy of $96.5(156)$\% for $142$
test gauge configurations.
The hyper-parameters which represent the structure of the neural
network model are given in Table \ref{tab:nn_para}.
Among the test set, we find five gauge configurations on which the
average accuracy per gauge configuration is lower than $50$\%.
Data show that some ghost (unphysical) eigenvectors are present in the
eigenvalue spectrum on these gauge configurations, so that the ML
prediction gives a wrong answer not due to failure of the ML algorithm
but due to human mistakes in labeling quartet samples based on the
eigenvalue index.
Excluding these five gauge configurations, we achieve the average
accuracy of $99.4(23)$\%.
Considering that all samples generated on the same gauge configuration
are connected by the eigenvalue index (or quartet index), this average
accuracy of $99.4$\% implies that one can in the end find completely
correct quartet groups for all the normal gauge configurations of the
test set.
It also demonstrates our claim that the leakage pattern is universal
over all normal gauge configuration ensembles.
Details of the results of this ML research will be reported separately
in Ref.~\cite{skLee:prep}.
%

%-------------------
% END of SECTION 6.
%-------------------

%-----------
% SECTION 7.
%-----------
%\input{sec.renorm.tex}
%-----------
\section{Zero modes and renormalization}
\label{sec:renorm}

\com{Fifth, explain how we can use zero modes to measure the
  renormalization factor for the bilinear operator: $Z_{[P\times
      S]}$.}

As explained in Sec.~\ref{sec:chirality}, we know that there is
practically no leakage for the zero modes in the chirality
measurement.
Hence, it is possible to determine the renormalization factor
$\kappa_P$ by imposing the index theorem as follows.
For $Q \ne 0$,
\begin{align}
  4 \times Q &= - \kappa_P \times \sum_{\lambda \in S_0}
  \langle f^s_\lambda | [\gamma_5 \otimes 1] | f^s_\lambda \rangle
  \label{eq:Z_P-1}
  \\
  \kappa_P &= - \frac{4 Q}{ C_0 }
  \label{eq:Z_P-2}
  \\
  C_0 &= \displaystyle \sum_{\lambda \in S_0} \Gamma_5(\lambda, \lambda)
\end{align}
where $S_0$ is the entire set of zero modes, and
\begin{align}
  \kappa_P &= \frac{ Z_{P \times S}(\mu) } {Z_{P \times P} (\mu)} \,,
\end{align}
where
\begin{align}
  O_{S} &= \bar{\chi}[\gamma_5 \otimes 1]\chi
  \\
  O_{P} &= \bar{\chi}[\gamma_5 \otimes \xi_5]\chi
  \\ {}
  [O_{S}]_R (\mu) &= Z_{P \times S}(\mu) [O_{S}]_B
  \\ {}
  [O_{P}]_R (\mu) &= Z_{P \times P}(\mu) [O_{P}]_B \,,
\end{align}
and the subscript $[\cdots]_R$ ($[\cdots]_B$) represents a
renormalized (bare) operator.
The $Z_{P \times S}$ and $Z_{P \times P}$ are the renormalization
factors for the bilinear operators $O_S$ and $O_P$, respectively.
One advantage of this scheme is that $\kappa_P$ is independent of valence
quark masses, even though we perform the measurement with arbitrary
masses for valence quarks.
Numerical results for $\kappa_P$ are summarized in Table
\ref{tab:renorm-1}.

There are a few key issues in the physical interpretation of
$\kappa_P$.
\begin{itemize}
\item Since the topological charge $Q$ and sum $C_0$ are independent
  of renormalization scale, $\kappa_P$ must be independent of the
  renormalization scale $\mu$.
  
\item This means that the scale dependence of $Z_{P\times S}(\mu)$
  must cancel off that of $Z_{P \times P}(\mu)$.

\item It would be nice to cross-check this property of $\kappa_P$ in
  the RI-MOM scheme \cite{ Aoki:2007xm} and in the RI-SMOM scheme
  \cite{ Sturm:2009kb}.

\end{itemize}
%-----------
% TABLE 10
%-----------
%\input{table/tab_Z_P}
%-----------
\begin{table}[!tbhp]
  \caption{Numerical results for $\kappa_P$.}
  \label{tab:renorm-1}
  \renewcommand{\arraystretch}{1.2}
  \begin{ruledtabular}
    \begin{tabular}{c | c | r}
      topological charge & number of samples & $\kappa_P$ \\ \hline
      $|Q|=1$ & 72 & 1.26(13) \\
      $|Q|=2$ & 68 & 1.22(3) \\
      $|Q|=3$ & 45 & 1.23(2) \\ \hline
      weighted average & 241 & 1.23(2) \\
    \end{tabular}
  \end{ruledtabular}
\end{table}
%-----------------
% END of TABLE 10
%-----------------

%------------------
% END of SECTION 7.
%------------------

%-----------
% SECTION 8.
%-----------
%\input{sec.conclude.tex}
%-----------
\section{Conclusion}
\label{sec:conclude}
%
% summary of this work:
%
We study general properties of the eigenvalue spectrum of Dirac
operators in the staggered fermion formalism.
As an example, we use the Dirac operator for HYP staggered quarks.
In Section \ref{sec:chirality}, we introduce a new chirality operator
$\Gamma_5$ and a new shift operator $\Xi_5$ and prove that they
respect the continuum recursion relationships, as given in
Eqs.~\eqref{eq:recur-1}-\eqref{eq:recur-4} and
Eqs.~\eqref{eq:rr-xi5-1}-\eqref{eq:rr-xi5-2}.
Using these operators with nice chiral properties, we find that the
leakage pattern of $|\Gamma_5|^{-j,m}_{\ell,m'}$ is related to that of
$|\Xi_5|^{j,m}_{\ell,m'}$ through the Ward identity of the conserved
$U(1)_A$ symmetry.

We find that the leakage pattern of $\Gamma_5$ and $\Xi_5$ for the
zero modes is quite different from that for the non-zero modes.
This difference in leakage pattern allows us to distinguish the zero
modes from the non-zero modes even though we do not know \textit{a
  priori} about the topological charge.
We find that using the leakage pattern of $\Gamma_5$ and $\Xi_5$, one
can determine the topological charge as reliably as when using
standard field theoretical methods such as the cooling method.

We use a machine learning (ML) technique to check the universality
of this leakage pattern over the entire ensemble of gauge
configurations (refer to Table \ref{tab:in-para}).
Our best-trained deep learning model identifies the quartet of
non-zero modes with 98.7(34)\% accuracy using a single normal gauge
configuration.
Choosing the highest probability prediction of the ML and comparing
the prediction with the known answer, we find that the ML can identify
all quartet groups on an eigenvalue spectrum correctly.
In addition, the ML technique detects wrong answers resulting from
human input mistakes since the ML prediction disagrees with a wrong
answer by giving the prediction with low accuracy ($<50\%$).
This reassures us that the ML technique is highly reliable at
identifying anomalous gauge configurations with defects such as
violation of the index theorem and ghost eigenmodes.
Once we identify the zero modes, it is also possible to determine the
ratio of renormalization factors $ \kappa_P = Z_{P\times S}(\mu) /
Z_{P \times P} (\mu)$ from the chirality measurement of $\Gamma_5$.

The leakage pattern is a new concept introduced in this paper.
It can be used to study the low lying eigenvalue spectrum
of staggered Dirac operators systematically.
It helps us understand how to extract the taste symmetry and
chiral symmetry from the staggered eigenvalue spectrum.
The leakage pattern will help us to dig out related physics more
efficiently, such as topological charge, index theorem, Banks-Casher
relation, and non-perturbative renormalization.
%

%------------------
% END of SECTION 8.
%------------------

%---------------------------------------------------------------------
\begin{acknowledgments}
  We would like to express sincere gratitude to Eduardo Follana for
  helpful discussion and providing his code to cross-check results of
  our code.
  We also thank Jon A. \nobreak{Bailey} for helpful comments on the manuscript.
  We also thank Stephen R.~Sharpe for helpful discussion. 
  The research of W.~Lee is supported by the Mid-Career Research
  Program (Grant No.~NRF-2019R1A2C2085685) of the NRF grant funded by
  the Korean government (MOE).
  This work was supported by Seoul National University Research Grant
  in 2019.
  W.~Lee would like to acknowledge the support from the KISTI
  supercomputing center through the strategic support program for the
  supercomputing application research [No.~KSC-2016-C3-0072,
    KSC-2017-G2-0009, KSC-2017-G2-0014, KSC-2018-G2-0004,
    KSC-2018-CHA-0010, KSC-2018-CHA-0043].
  Computations were carried out in part on the DAVID supercomputer at
  Seoul National University.
\end{acknowledgments}

%---------------------------------------------------------------------
\appendix

%------------
% APPENDIX A.
%------------
%\input{app.lanczos.tex}
%------------
\section{Lanczos algorithm}
\label{app:lanczos}
Lanczos is a numerical algorithm for calculating eigenvalues and
eigenvectors of a Hermitian matrix \cite{Lanczos:1950zz}.
It transforms an $n \times n$ Hermitian matrix $H$ to a tridiagonal
matrix $T$ through a unitary transformation $Q$, which is represented
by
\begin{equation}
  T = Q^\dagger H Q \,.
  \label{eq:lanczos}
\end{equation}
Here columns of $Q$ are composed of basis vectors of the $n$th Krylov
subspace $\mathcal{K}_n (H,b)$ generated by $H$ and a starting vector
$b$ of our choice.
Each iteration of Lanczos computes a column of $Q$ and $T$ in sequence.
At the end, diagonalizing the tridiagonal matrix $T$ yields
eigenvalues and eigenvectors of $H$.
In principle, Lanczos is a direct method that takes $n$ iterations to
construct the $n \times n$ tridiagonal matrix $T$.
However, since these columns of $T$ are computed in order, a sequence
of $m < n $ iterations also constructs an $m \times m$ tridiagonal
matrix $T'$ which is a submatrix of $T$.
In practice, the real benefit of Lanczos is that eigenvalues of $T'$
approximate eigenvalues of $T$.
As iteration continues, and the size of the submatrix $T'$ increases,
eigenvalues of $T'$ converge to eigenvalues of $T$.
The convergence behavior is somewhat complicated.
The eigenvalues converge to the largest, the smallest, or the most
sparse eigenvalue first.
The speed of convergence depends on the density of eigenvalues.
The less dense, the faster the convergence.
In this paper, we make use of two popular improvement techniques of
Lanczos: (1) implicit restart \cite{Lehoucq96deflationtechniques}, and
(2) polynomial acceleration with Chebyshev polynomials \cite{MR736453}.
The implicit restart method gets rid of converged eigenvalues in the middle
of the Lanczos iteration.
It takes effect as if we restarted the Lanczos with a shifted matrix $H'$
given by
\begin{equation}
  H' \;\equiv\; H - \sum_i \lambda_i \mathcal{I} \,,
\end{equation}
where $\lambda_i$ are eigenvalues we want to remove.
Then $H'$ is still Hermitian but does not have such eigenvalues
$\lambda_i$.
Hence, Lanczos with $H'$ converges to remaining eigenvalues faster.
The implicit restarting procedure gives us a new submatrix, which has
a dimension ($(m-r) \times (m-r)$) reduced by the number of
eigenvalues we have removed ($r$).
Then we iterate Lanczos $r$ times to refill the submatrix and restore
the structure of the $m \times m$ matrix.
We repeat the implicit restart to obtain a new submatrix of dimension
$(m-r) \times (m-r)$, and so on.
This procedure allows us to control the size of the submatrix, the
computational cost, and the memory usage, while the submatrix $T'$
contains $(m-r)$ eigenmodes that are more precise (much closer to the
true eigenmodes of the full matrix $H$) for each iteration.
A polynomial operation on a matrix changes the eigenvalue spectrum
accordingly while retaining the eigenvectors.
Since the polynomial of a Hermitian matrix is also Hermitian, Lanczos
is still available to calculate its eigenvalues and eigenvectors.
By choosing a proper polynomial, one can manipulate the density of the
eigenvalue spectrum so that the convergence to the desired eigenvalues
is accelerated.
A Chebyshev polynomial is a popular choice for this purpose.
Using the Chebyshev polynomial, we want to map the first region of
eigenmodes of no interest to $[-1, 1]$ and map the second region of
eigenmodes of our interest to $[-\infty, -1]$.
In the interval $[-1,1]$, the eigenvalues are dense enough that
Lanczos does not converge.
In addition, the Chebyshev polynomial rapidly changes in the second
region so that the density of eigenmodes is low enough to more quickly
accelerate the convergence of Lanczos.
Here we apply the Chebyshev polynomial for $D_s^\dagger D_s$, whose
eigenvalues are $\lambda^2 \ge 0$.
We set the lower bound of the first region to a value somewhat greater
than the largest eigenvalue of interest.
This strategy will not only suppress high unwanted eigenmodes, but
also accelerate the speed of Lanczos for the low eigenmodes of
interest.
Numerical stability is essential for the Lanczos algorithm.
Each Lanczos iteration generates Lanczos vectors, which are column
vectors of the unitary matrix $Q$ in Eq.~\eqref{eq:lanczos}.
After several iterations, however, Lanczos vectors lose their
orthogonality due to gradual loss of numerical precision.
If not addressed, this loss would induce spurious ghost eigenvalues
\cite{Cullum:1985aa}.
A straightforward prescription to solve the problem is performing a
reorthogonalization for every calculation of Lanczos vectors.
There are also alternative approaches to eliminate the ghost
eigenvalues without reorthogonalization, such as the Cullum-Willoughby
method \cite{CULLUM1981329, Cullum:2002aa}.
Here we choose the first solution and perform the full
reorthogonalization for each Lanczos iteration.
For a large scale simulation using Lanczos, Multi-Grid Lanczos \cite{
  Clark:2017wom} and Block Lanczos \cite{ Jang:2019roq} are available.
Multi-Grid Lanczos is also based on the implicit restart and Chebyshev
acceleration.
In addition, Multi-Grid Lanczos reduces the memory requirement
significantly by compressing the eigenvectors using their local
coherence \cite{Luscher:2007se}.
A spatially-blocked deflation subspace is constructed from some of the
lowest eigenvectors of the Dirac operator.
Then the coherence of eigenvectors allows us to represent other
eigenvectors on this subspace and to run Lanczos with much less
memory.
Meanwhile, Block Lanczos utilizes the Split Grid method
\cite{Jang:2019roq}.
This algorithm deals with multiple starting vectors for Lanczos, where
the Split Grid method divides the domain of the Dirac operator
application into multiple smaller domains so that each partial domain
runs in parallel on a partial grid (lattice) with a lower surface to
volume ratio than that of the full grid.
Hence, one can optimize the off-node communication by adjusting the block
(grid) size.
This approach would give a significant speed-up compared with our
method.
We plan to implement Multi-Grid Lanczos and Block Lanczos in the near
future.
%

%-------------------
% END of APPENDIX A.
%-------------------

%------------
% APPENDIX B.
%------------
%\input{app.phase.tex}
%------------
\section{Even-odd preconditioning and phase ambiguity}
\label{app:phase}
Even-odd preconditioning reorders a fermion field $\chi(x)$ so that
even site fermion fields are obtained first, and odd site fermion
fields are obtained from them:
\begin{equation}
  \chi(x) =
  \begin{pmatrix}
    \chi_e \\
    \chi_o \\
  \end{pmatrix}
  \,,
\end{equation}
where $\chi_e$ ($\chi_o$) is the fermion field collection on even
(odd) sites.
On this basis, the massless staggered Dirac operator $D_s$ can be
represented as a block matrix:
\begin{equation}
  D_s =
  \begin{pmatrix}
    0 & D_{eo} \\
    D_{oe} & 0 \\
  \end{pmatrix}
  \,,
\end{equation}
where $D_{oe}$ ($D_{eo}$) relates even (odd) site fermion fields to
odd (even) site fermion fields.
Since $D_s^\dagger = -D_s$, we also find that $D_{oe}^\dagger = -D_{eo}$
and $D_{eo}^\dagger = -D_{oe}$.
On this basis, $D_s^\dagger D_s$ is expressed as
\begin{align}
  D_s^\dagger D_s & =
  \begin{pmatrix}
    0 & -D_{eo} \\
    -D_{oe} & 0 \\
  \end{pmatrix}
  \begin{pmatrix}
    0 & D_{eo} \\
    D_{oe} & 0 \\
  \end{pmatrix}
  \\
  & =
  \begin{pmatrix}
    - D_{eo} D_{oe} & 0 \\
    0 & - D_{oe} D_{eo} \\
  \end{pmatrix}
  \,.
\end{align}
Hence, the eigenvalue equation of $D_s^\dagger D_s$
(Eq.~\eqref{eq:stag-dirac-spec-2}) can be divided into two eigenvalue
equations as follows,
\begin{align}
  - D_{eo} D_{oe} | g_e \rangle & = \lambda^2 | g_e \rangle \,,
  \label{eq:g_e}
  \\
  - D_{oe} D_{eo} | g_o \rangle & = \lambda^2 | g_o \rangle \,,
  \label{eq:g_o}
\end{align}
where $| g_{e(o)} \rangle$ is the collection of even (odd) site
components of $| g^s_{\lambda^2} \rangle$.
Here we omit the superscript $s$ and the subscript ${\lambda^2}$ for
notational simplicity.
Now let us multiply $D_{oe}$ from the left on both sides of
Eq.~\eqref{eq:g_e}.
Then we find that
\begin{equation}
  - D_{oe} D_{eo} ( D_{oe} | g_e \rangle ) = \lambda^2 ( D_{oe} | g_e
  \rangle ) \,,
\end{equation}
which is identical to Eq.~\eqref{eq:g_o}.
Hence, we find that $| g_o \rangle = \eta\, D_{oe} | g_e \rangle$
where $\eta = r e^{i\alpha}$ is an arbitrary complex number with $r>0$
and $0 \leq \alpha < 2\pi$.
Here $r$ represents the scaling behavior and $\alpha$ represents a
random phase.
Since $- D_{eo} D_{oe} (=D_{oe}^\dagger D_{oe})$ is Hermitian and
positive semi-definite, one can solve Eq.~\eqref{eq:g_e} using the
Lanczos algorithm introduced in Appendix \ref{app:lanczos}.
From the result for $| g_e \rangle$, it is straightforward to obtain
the eigenvector $| g^s_{\lambda^2} \rangle$ of
Eq.~\eqref{eq:stag-dirac-spec-2} since
\begin{equation}
  | g^s_{\lambda^2} \rangle =
  \begin{pmatrix}
    | g_e \rangle \\
    \eta\, D_{oe} | g_e \rangle
  \end{pmatrix}
  \,.
  \label{eq:g_eo}
\end{equation}
Now we apply the projection operator $P_+$, defined in
Eq.~\eqref{eq:proj_p}, to $| g^s_{\lambda^2} \rangle$.
Using Eq.~\eqref{eq:g_e}, we find that
\begin{align}
  | \chi_+ \rangle = P_+ | g^s_{\lambda^2} \rangle & =
  \begin{pmatrix}
    i \lambda & D_{eo} \\
    D_{oe} & i \lambda
  \end{pmatrix}
  \begin{pmatrix}
    | g_e \rangle \\
    \eta\, D_{oe} | g_e \rangle
  \end{pmatrix}
  \nonumber \\
  & = (1 + i \eta \lambda)
  \begin{pmatrix}
    i \lambda\, | g_e \rangle \\
    D_{oe} | g_e \rangle
  \end{pmatrix}
  \,.
  \label{eq:chi_p_eo}
\end{align}
Similarly, for the projection operator $P_-$, defined in
Eq.~\eqref{eq:proj_m}, we find that
\begin{align}
  | \chi_- \rangle = P_- | g^s_{\lambda^2} \rangle =
  (1 - i \eta \lambda)
  \begin{pmatrix}
    - i \lambda\, | g_e \rangle \\
    D_{oe} | g_e \rangle
  \end{pmatrix}
  \,.
  \label{eq:chi_m_eo}
\end{align}
Since $\eta$ only appears in the overall factor for both cases, it
gives only the relative phase difference between the normalized
eigenvectors $| f^s_{\pm \lambda} \rangle$ defined in
Eqs.~\eqref{eq:f_p} and \eqref{eq:f_m}.
We can proceed further to obtain the eigenvectors $| f^s_{\pm\lambda}
\rangle$.
The norm of $| \chi_+ \rangle$ is given by
\begin{align}
  \langle \chi_+ | \chi_+ \rangle 
  & = [ (1 - i \eta^* \lambda) (1 + i \eta
    \lambda ) ] \cdot 2 \lambda^2 \langle g_e | g_e \rangle \,.
\end{align}
Hence, $| f^s_{+\lambda} \rangle$ is
\begin{align}
  | f^s_{+\lambda} \rangle & = \frac{1}{N} \sqrt{
    \frac{1+i\eta\lambda}{ 1-i\eta^*\lambda } }
  \begin{pmatrix}
    i \lambda\, | g_e \rangle \\
    D_{oe} | g_e \rangle
  \end{pmatrix}
  \,,
\end{align}
where
\begin{align}
  N \equiv \sqrt{ 2\lambda^2 \langle g_e | g_e \rangle} \,.
\end{align}
Similarly,
\begin{align}
  | f^s_{-\lambda} \rangle & = \frac{1}{N} \sqrt{
    \frac{1-i\eta\lambda}{1+i\eta^*\lambda} }
  \begin{pmatrix}
    - i \lambda\, | g_e \rangle \\
    D_{oe} | g_e \rangle
  \end{pmatrix}
  \,.
\end{align}
These results for $| f^s_{\pm\lambda} \rangle$ indicate that the phase
difference $\theta$ for the $\Gamma_\epsilon$ transformation defined
in Eq.~\eqref{eq:eps-WI} depends on the value of $\eta$.

In our numerical study, we set $\eta$ to $\eta = r e^{i\alpha} = 1$:
$r=1$ and $\alpha=0$.
Hence, the relative random phase between $| f^s_{\pm\lambda} \rangle$
states is removed by hand.
Therefore, our value of $\theta$ defined in Eq.~\eqref{eq:eps-WI}
includes a bias from our choice of $\eta=1$. 
For $\eta = 1$ (our choice), $\Gamma_\epsilon | f^s_{+\lambda}
\rangle$ is
\begin{align}
  \Gamma_\epsilon | f^s_{+\lambda} \rangle & = \frac{1}{N} \sqrt{
    \frac{1+i\lambda}{ 1-i\lambda } }
  \begin{pmatrix}
    i \lambda\, | g_e \rangle \\
    - D_{oe} | g_e \rangle
  \end{pmatrix}
  \,,
\end{align}
while $| f^s_{- \lambda} \rangle$ is
\begin{align}
  | f^s_{-\lambda} \rangle & = \frac{1}{N} \sqrt{
    \frac{1-i\lambda}{1+i\lambda} }
  \begin{pmatrix}
    - i \lambda\, | g_e \rangle \\
    D_{oe} | g_e \rangle
  \end{pmatrix}
  \,.
\end{align}
Then we obtain $e^{i\theta}$ from the following matrix element,
\begin{align}
  \langle f^s_{-\lambda} | \Gamma_\epsilon | f^s_{+\lambda} \rangle
  & = \frac{1}{N^2} \sqrt{ \left( \frac{1-i\lambda}{1+i\lambda} \right)^*
    \frac{1+i\lambda}{1-i\lambda} } \;\cdot (-N^2)
  \nonumber \\
  & = - \frac{1+i\lambda}{1-i\lambda}
  \nonumber \\
  & = e^{i (\pi + 2\beta)} = e^{i\theta}\,,
  \label{eq:theta-1}
\end{align}
where $\beta \equiv \arctan(\lambda)$.
From Eqs.~\eqref{eq:eps-WI} and \eqref{eq:theta-1}, we find that
\begin{equation}
  \theta = \pi + 2\beta \,.
  \label{eq:theta}
\end{equation}
In Fig.~\ref{fig:theta}, we show the measurements of the phase
$\theta$ for hundreds of eigenvectors on a gauge configuration with $Q
= -1$.
The results for $\theta$ are consistent with our theoretical
prediction Eq.~\eqref{eq:theta} within numerical precision.
%

%-------------------
% END of APPENDIX B.
%-------------------

%------------
% APPENDIX C.
%------------
%\input{app.data.tex}
%------------
\section{Eigenvalue spectrum for $Q=-2$ and $Q=-3$}
\label{app:qm2-qm3-dat}
%
%-----------
% FIGURE 11.
%-----------
%\input{figs/fig_qm2_ev}
%-----------
\begin{figure}[t]
%  \vspace*{-5mm}
%  \centering
%  \footnotesize
%  \renewcommand{\arraystretch}{1.2}
%  \renewcommand{\subfigcapskip}{-0.55em}
%
%  \hspace*{-7mm}
  \subfigure[$\lambda_i^2$]{
    \includegraphics[width=\linewidth]{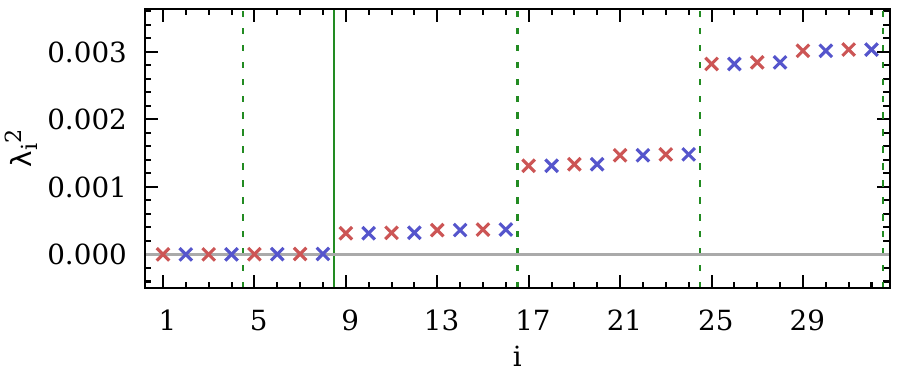}
    \label{sfig:qm2-abs-ev}
  }
  \\
  \subfigure[$\lambda_i$]{
    \includegraphics[width=\linewidth]{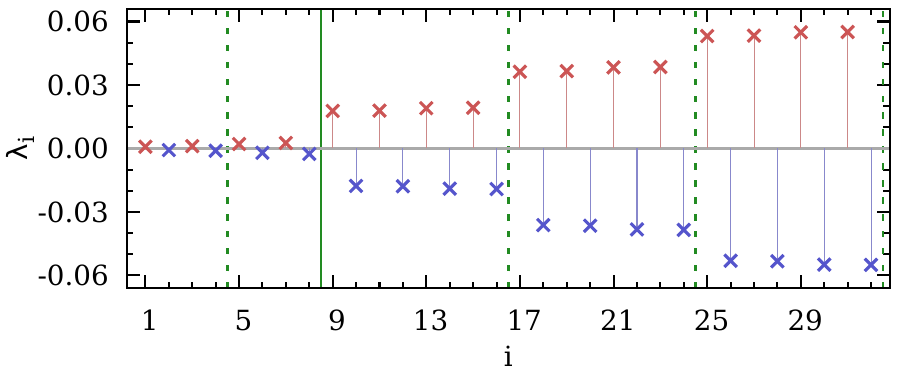}
    \label{sfig:qm2-ev}
  }
  \caption{The same as Fig.~\ref{fig:q0-ev} except for $Q=-2$.}
  \label{fig:qm2-ev}
\end{figure}
%-------------------
% END of FIGURE 11.
%-------------------
%
%-----------
% FIGURE 12.
%-----------
%\input{figs/fig_qm3_ev}
%-----------
\begin{figure}[t]
%  \vspace*{-5mm}
%  \centering
%  \footnotesize
%  \renewcommand{\arraystretch}{1.2}
%  \renewcommand{\subfigcapskip}{-0.55em}
%
%  \hspace*{-7mm}
  \subfigure[$\lambda_i^2$]{
    \includegraphics[width=\linewidth]{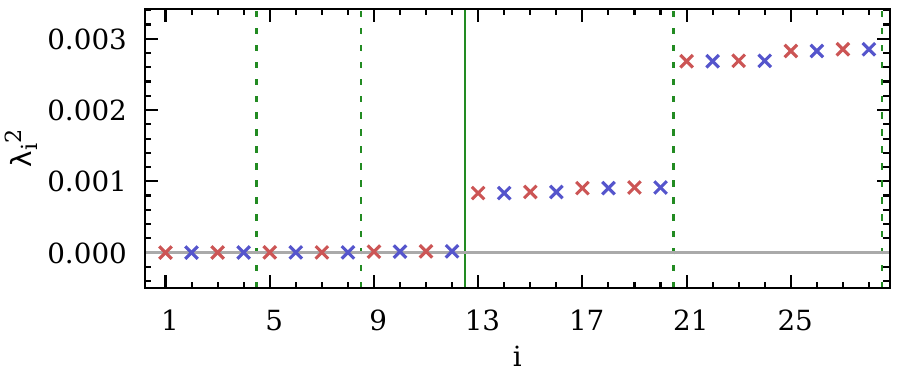}
    \label{sfig:qm3-abs-ev}
  }
  \\
  \subfigure[$\lambda_i$]{
    \includegraphics[width=\linewidth]{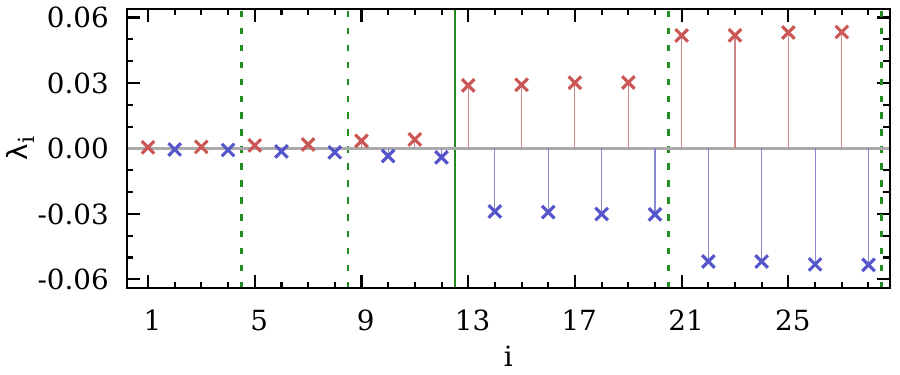}
    \label{sfig:qm3-ev}
  }
  \caption{The same as Fig.~\ref{fig:q0-ev} except for $Q=-3$.}
  \label{fig:qm3-ev}
\end{figure}
%------------------
% END of FIGURE 12.
%------------------
In Figs.~\ref{fig:qm2-ev} and \ref{fig:qm3-ev}, we present examples of
the eigenvalue spectrum for $Q = -2$ and $Q = -3$, respectively.
Figures \ref{sfig:qm2-abs-ev} and \ref{sfig:qm3-abs-ev} show
eigenvalues $\lambda^2$ for the eigenvectors $| g^s_{\lambda^2}
\rangle$ defined in Eq.~\eqref{eq:stag-dirac-spec-2}.
In Fig.~\ref{fig:qm2-ev}, we find two sets of four-fold degenerate
eigenstates, corresponding to $\{ \lambda_1, \lambda_2, \lambda_3,
\lambda_4 \}$ and $\{ \lambda_5, \lambda_6, \lambda_7, \lambda_8 \} $.
Each set of four eigenvalues indicates a quartet of would-be zero
modes.
The number of quartets is related to the topological charge $Q = -2$
by the index theorem of Eq.~\eqref{eq:indexThm} when all would-be zero
modes have the same chirality ($n_-= 0$ and $n_+ = 2$).
Apart from the would-be zero modes, we observe that non-zero modes are
eight-fold degenerate, as in the cases $Q = 0$
(Fig.~\ref{fig:q0-ev}) and $Q = -1$ (Fig.~\ref{fig:qm1-ev}).
Similarly, in Fig.~\ref{fig:qm3-ev}, we find three quartets of
would-be zero modes with $n_-=0$ and $n_+=3$ ($Q=-3$): $\{ \lambda_1,
\lambda_2, \lambda_3, \lambda_4 \}$, $\{ \lambda_5, \lambda_6,
\lambda_7, \lambda_8 \} $, and $\{ \lambda_9, \lambda_{10},
\lambda_{11}, \lambda_{12} \}$.
Because the number of quartets equals the absolute value of the topological
charge $|Q| = 3$, it is possible to deduce that all the would-be zero modes
have the same chirality in accordance with the index theorem of
Eq.~\eqref{eq:indexThm}.
For non-zero modes, we observe the pattern of eight-fold degeneracy as
in other examples for $Q = 0$ in Fig.~\ref{fig:q0-ev}, $Q=-1$
in Fig.~\ref{fig:qm1-ev}, and $Q=-2$ in Fig.~\ref{fig:qm2-ev}.
%

%--------------------
% END of APPENDIX C.
%--------------------

%------------
% APPENDIX D.
%------------
%\wlee{EDIT by wlee : begin}
%
\section{Comparison of the Golterman and Kluberg-Stern methods}
\label{app:cmp:gol-klu}

Using the Golterman method, the chirality operator is defined as follows.
\begin{align}
  \mathcal{O}^{\text{Gol}}_{\gamma_5 \times 1} (x)
  &= \sum_A \bar{\chi}(x_A) [\gamma_5 \otimes 1]^{\text{Gol}} \chi(x_A)
  \nonumber \\
  &= \sum_{A} \bar{\chi}(x_A) [\rho_{\gamma_5 \otimes 1}(A)
  M_{\gamma_5 \otimes 1}] \chi(x_A)
\end{align}
where the coordinate $x_A = 2x + A$, $x$ is a coordinate of the
hypercube, and $A$ is a hypercubic vector with $A_\mu \in \{0,1\}$
with $\mu = 1,2,3,4$.
$M_{\gamma_5 \otimes 1}$ is defined as
\begin{align}
  M_{\gamma_5 \otimes 1} \chi(x_A)
  &\equiv \frac{1}{4!} \sum_{\{p_a\} \in \mathcal{P}} \tilde{D}_{p_4}
  \tilde{D}_{p_3}\tilde{D}_{p_2}\tilde{D}_{p_1} \chi(x_A)
\end{align}
where $p_1 \ne p_2 \ne p_3 \ne p_4$, and $\mathcal{P}$ is the set of all
permutations of \{1,2,3,4\}.
The symmetric shift operator $\tilde{D}_{\mu}$ is defined as
\begin{align}
  \tilde{D}_{\mu} \chi(y) &= \frac{1}{2}
  [ V_\mu(y) \chi(y+\hat{\mu})
  + V^\dagger_\mu(y-\hat\mu) \chi(y-\hat\mu)]
\end{align}
where $\hat\mu$ is a unit vector in the $\mu$ direction of Euclidean
space.
$V_\mu(y)$ is a (smeared) gauge link used in Eq.~\eqref{eq:U-1}.
$\rho_{\gamma_5 \otimes 1}$ is defined as
\begin{align}
  \rho_{\gamma_5 \otimes 1}(A) &=
  \frac{1}{4} \Tr (\gamma_A^\dagger \gamma_5 \gamma_{\bar{A}})
  = (-1)^{A_1 + A_3}
\end{align}
where $\bar{A}_\mu = (A_\mu + 1) \mod 2$.
In the Golterman method, the chirality operator $[\gamma_5 \otimes
  1]^\text{Gol}$ connects a site A=(0,0,0,0) of $\bar\chi$ to the 16
sites B=(1,1,1,1), (-1,1,1,1), ..., (-1,-1,-1,-1) of $\chi$.
As a consequence,
\begin{align}
  ([\gamma_5 \otimes 1]^\text{Gol})^2 \ne 1
  \label{eq:rr-gol-1}
\end{align}
unlike the continuum chirality operator which satisfies
$[\gamma_5 \otimes 1]^2 = [1 \otimes 1] = 1$.
%

%-----------------------
% Kluberg-Stern method
%-----------------------
%
Using the definition of bilinear operators in Eq.~\eqref{eq:bi-op-1},
obtained with the Kluberg-Stern method, we define the chirality operator as
\begin{align}
  \mathcal{O}_{\gamma_5 \times 1}^{\text{Klu}} (x)
  &\equiv \sum_{A}
  \bar{\chi}(x_A) [\gamma_5 \otimes 1]^{\text{Klu}}_{A\bar{A}}
  \chi(x_{\bar{A}})
  \nonumber \\
  &= \sum_{A} \bar{\chi}_a (x_A) \overline{(\gamma_5 \otimes 1)}_{A\bar{A}}
  U(x_A,x_{\bar{A}})_{ab} \chi_b(x_{\bar{A}})
  \label{eq:chiral-op-1}
\end{align}
Using the definition in Eq.~\eqref{eq:s-t-1} with $\gamma_S = \gamma_5$
and $\xi_T = 1$, we find 
\begin{align}
  \overline{(\gamma_5 \otimes 1)}_{A\bar{A}} &\equiv
  \frac{1}{4} \Tr(\gamma_A^\dagger \gamma_5 \gamma_{\bar{A}} 1^\dagger)
  \nonumber \\
  &= \rho_{\gamma_5 \otimes 1}(A) = (-1)^{A_1 + A_3}
  \label{eq:chiral-op-2}
\end{align}
Hence, the phase of the Kluberg-Stern operator is identical to
that of the Golterman operator, which is in general true for 
the whole set of bilinear operators.
We have freedom to choose $U(x_A,x_{\bar{A}})_{ab}$ to make
the chirality operator gauge-invariant.
Here, we set $U(x_A,x_{\bar{A}})_{ab}$ to
\begin{align}
  U(x_A,x_{\bar{A}}) &\equiv
  \mathbb{P}_{SU(3)} \bigg[ \sum_{p \in \mathcal{C}}
    V(x_A,x_{p_1}) V(x_{p_1},x_{p_2})
    \nonumber \\
    & \qquad \qquad \qquad \cdots V(x_{p_n},x_{\bar{A}})
    \bigg]
  \label{eq:U-2}
\end{align}
where $\mathbb{P}_{SU(3)}$ represents the $SU(3)$ projection, and
$\mathcal{C}$ represents the complete set of shortest paths from $x_A$
to $x_{\bar{A}}$.
Here, the $SU(3)$ projection is crucial to make the chirality operator
satisfy the continuum recursion relation:
\begin{align}
  ([\gamma_5 \otimes 1]^{\text{Klu}})^2 &= [1 \otimes 1] = 1
  \label{eq:rec-rel-1}
\end{align}
A rigorous proof of Eq.~\eqref{eq:rec-rel-1} is given in Appendix
\ref{app:recur}, Theorem \ref{thm-1}.
In the Kluberg-Stern method, the chirality operator $[\gamma_5 \otimes
  1]^\text{Klu}$ connects a site $A=(0,0,0,0)$ of $\bar\chi$ to a
single site $\bar{A} = (1,1,1,1)$, which makes it possible to satisfy
the recursion relation of Eq.~\eqref{eq:rec-rel-1}.

The Kluberg-Stern operator without the $SU(3)$ projection
$\mathbb{P}_{SU(3)}$ contains the Golterman operator as a leading
term:
\begin{align}
  \mathcal{O}^{\text{Klu}}_{\gamma_5 \times 1} (x) &=
  \mathcal{O}^{\text{Gol}}_{\gamma_5 \times 1} (x)
  + \mathcal{O}_{\text{irrel}}(x)
\end{align}
where $\mathcal{O}_\text{irrel}$ represents irrelevant operators of
higher dimension.
For example, $\mathcal{O}_\text{irrel}$ includes a 4-dimensional
operator:
\begin{align}
  \mathcal{O}^\text{irrel}_{\gamma_\mu \times \xi_{\mu5}}
  &= \sum_A \bar{\chi}(x_A) [\rho_{\gamma_\mu \otimes \xi_{\mu5} } (A)
    \; M_{\gamma_\mu \otimes \xi_{\mu5}}] \chi(x_A)
  \\
  \rho_{\gamma_\mu \otimes \xi_{\mu5} } (A) &=
  \frac{1}{4} \Tr(\gamma_A^\dagger \gamma_\mu \gamma_{\bar{A}}
  \gamma_{\mu5}^\dagger)
  \\
  M_{\gamma_\mu \otimes \xi_{\mu5}} &= \frac{1}{4!}
  \sum_{\{p_a\} \in \mathcal{P}_\mu}
  \Big[ D_\mu \tilde{D}_{p_1} \tilde{D}_{p_2} \tilde{D}_{p_3}
    \nonumber \\ 
    & + \tilde{D}_{p_1} D_\mu \tilde{D}_{p_2} \tilde{D}_{p_3}
    + \tilde{D}_{p_1} \tilde{D}_{p_2} D_\mu \tilde{D}_{p_3}
    \nonumber \\     
    & + \tilde{D}_{p_1} \tilde{D}_{p_2} \tilde{D}_{p_3} D_\mu \Big]
  \\
    D_{\mu} \chi(y) &= \frac{1}{2} [ V_\mu(y) \chi(y+\hat{\mu})
  - V^\dagger_\mu(y-\hat\mu) \chi(y-\hat\mu)]
\end{align}
where $p_1 \ne p_2 \ne p_3 \ne \mu$, and $\mathcal{P}_\mu$ is
the set of all permutations of $\{p_a | p_a \ne \mu\}$.
Technical details of the derivation of a complete set of irrelevant
operators are explained in Ref.~\cite{ Bae:2008qe}.
All the irrelevant operators have tastes different from $1$ ($\xi_T
\ne 1$), and they contain at least one derivative $D_\mu$, which
leads to higher dimension operators.
As a consequence, their contribution to the chirality vanishes in the
continuum limit $a \rightarrow 0$.
%

% Unique eigenvalue of chirality operators.
%
The recursion relation in Eq.~\eqref{eq:rec-rel-1} is essential to
define the chirality value uniquely for the staggered fermion
formulation.
\begin{align}
  ([\gamma_5 \otimes 1]^{\text{Klu}})^{2n+1} &=
  [\gamma_5 \otimes 1]^{\text{Klu}}
  \label{eq:rr-klu-1}
\end{align}
for all positive $n \in Z$.
Hence, in the case of the Kluberg-Stern operators with the $SU(3)$
projection, we can define the chirality value uniquely without
any ambiguity.
However, in the case of the Golterman operators, it is not possible to
define the chirality value uniquely due to the following
ambiguity:
\begin{align}
  \mathcal{O}^{\text{Gol},n}_{\gamma_5 \times 1} (x)
  &\equiv \sum_A \bar{\chi}(x_A) ([\gamma_5 \otimes 1]^{\text{Gol}})^{2n+1}
  \chi(x_A)
  \\
  [\gamma_5 \otimes 1]^{\text{Gol}}
  &\ne ([\gamma_5 \otimes 1]^{\text{Gol}})^{2n+1}
  \label{eq:rr-gol-2}
  \\
  \mathcal{O}^{\text{Gol},n}_{\gamma_5 \times 1}
  &\ne \mathcal{O}^{\text{Gol},m}_{\gamma_5 \times 1}
  \quad \text{if $n \ne m$} \,.
  \label{eq:u-gol-1}
\end{align}
In addition, the Golterman operator does not satisfy the Ward identity,
while the Kluberg-Stern operator respects it,
\begin{align}
  [\gamma_5 \otimes 1]^{\text{Gol}} [1 \otimes \xi_5]^\text{Gol}
  &\ne [\gamma_5 \otimes \xi_5] = \Gamma_\epsilon
  \label{eq:wi-gol-1}
  \\
  [\gamma_5 \otimes 1]^{\text{Klu}} [1 \otimes \xi_5]^\text{Klu}
  &= [\gamma_5 \otimes \xi_5] = \Gamma_\epsilon
  \label{eq:wi-klu-1}
\end{align}
A rigorous proof is given in Theorem \ref{thm-3}.
However, in the continuum limit, they converge to a unique value:
\begin{align}
  \lim_{a \rightarrow 0} \mathcal{O}^{\text{Gol},n}_{\gamma_5 \times 1}
  &= \lim_{a \rightarrow 0} \mathcal{O}^{\text{Gol},1}_{\gamma_5 \times 1}
  \quad \text{$\forall$ positive $n \in Z$} 
  \nonumber \\
  &= \lim_{a \rightarrow 0} \mathcal{O}^{\text{Klu}}_{\gamma_5 \times 1}
\end{align}
since the contribution from all irrelevant operators vanishes in
the continuum.
We summarize the differences between the Golterman method and the
Kluberg-Stern method in Table \ref{tab:cmp-gol-klu-1}.
%
%----------
% TABLE 11
%----------
%
\begin{table}[!h]
  \caption{ Comparison between the Golterman and Kluberg-Stern
    methods. Here Gol (Klu) represents the Golterman (Kluberg-Stern)
    method. Recursion represents the recursion
    relationship. Uniqueness represents the uniqueness of the chirality
    operator value. Ward Id.~represents Ward identity. The $\bigcirc$
    ($\vartimes$) indicates that a given property is (is not) respected
    by a specific transcription. Ref.~represents key equations given for reference.}
  \label{tab:cmp-gol-klu-1}
  \renewcommand{\arraystretch}{1.2}
  \begin{ruledtabular}
    \begin{tabular}{l | l l | l}
      property  & Gol & Klu & Ref. \\ \hline
      Recursion & $\vartimes$  & $\bigcirc$
      & Eqs.~\eqref{eq:rr-gol-1} and \eqref{eq:rr-klu-1} \\
      Uniqueness & $\vartimes$  & $\bigcirc$
      & Eqs.~\eqref{eq:u-gol-1} and \eqref{eq:rr-klu-1} \\
      Ward Id.  & $\vartimes$  & $\bigcirc$
      & Eqs.~\eqref{eq:wi-gol-1} and \eqref{eq:wi-klu-1} \\
    \end{tabular}
  \end{ruledtabular}
\end{table}

%
%\wlee{EDIT by wlee : end}
%

%------------
% APPENDIX E.
%------------
%\input{app.recur.tex}
%------------
\section{Recursion relationships for chirality operators}
\label{app:recur}

We define the chirality operator
\begin{align}
  & \langle f^s_\alpha | [\gamma_5 \otimes 1] | f^s_\beta \rangle
  \equiv
  \nonumber \\
  & \sum_{x}\sum_{A,B} \; [f^s_\alpha(x_A)]^\dagger
  \overline{(\gamma_5 \otimes 1 )}_{AB} U(x_A,x_B) f^s_\beta(x_B)
  \label{eq:chirality-11}
  \\
  & \overline{(\gamma_S \otimes \xi_T)}_{AB} =
  \frac{1}{4} \Tr(\gamma_A^\dagger \gamma_S \gamma_B \gamma_T^\dagger)
  \\
  & U(x_A,x_B) = \mathbb{P}_{SU(3)} \bigg[ \sum_{p \in \mathcal{C}}
    V(x_A,x_{p_1}) V(x_{p_1},x_{p_2})
    \nonumber \\
    & \qquad \qquad \qquad V(x_{p_2},x_{p_3}) V(x_{p_3},x_{B})
    \bigg]
  \label{eq:link-11}
\end{align}

First let us prove the following theorem.
\begin{thm}
  \begin{align}
    & [\gamma_5 \otimes 1] [\gamma_5 \otimes 1] =
    [1 \otimes 1]
    \label{eq:thm-1}
  \end{align}
  \label{thm-1}
\end{thm}
\begin{proof}
  %%%
  Let us first rewrite $[\gamma_5 \otimes 1]^2 $ as follows,
  \begin{align}
    [\gamma_5 \otimes 1]^2_{AC} & =
    \sum_B \overline{(\gamma_5 \otimes 1 )}_{AB} U(x_A,x_B)
    \nonumber \\
    & \qquad \qquad \cdot \overline{(\gamma_5 \otimes 1 )}_{BC} U(x_B, x_C)
    \nonumber \\
    &= \sum_B [ \overline{(\gamma_5 \otimes 1 )}_{AB}
      \overline{(\gamma_5 \otimes 1 )}_{BC} ]
    \nonumber \\
    & \qquad \qquad \cdot [ U(x_A, x_B) U(x_B, x_C) ]
    \label{eq:id-1}
  \end{align}
  We know that
  \begin{align}
    \overline{(\gamma_5 \otimes 1 )}_{AB}
    & = \frac{1}{4} \Tr(\gamma_A^\dagger \gamma_5 \gamma_B 1 )
    \nonumber \\
    & = \delta_{B\bar{A}} [\eta_1(A) \eta_2(A) \eta_3(A) \eta_4(A)]
    \nonumber \\
    & = \delta_{B\bar{A}} \eta_5(A) \,,
    \label{eq:g5x1:1}
  \end{align}
  where $\bar{A}_\mu = (A_\mu + 1) \mod 2$, and
  \begin{align}
    \eta_\mu(A) &= (-1)^{X_\mu}, \qquad
    \text{ for $\mu = 1,2,3,4$} \,,
    \\
    X_\mu &= \sum_{\nu<\mu} A_\nu,
    \\ \eta_5(A) &= \eta_1(A) \eta_2(A) \eta_3(A)
    \eta_4(A) = (-1)^{A_1 + A_3} \,.
  \end{align}
  Similarly, we find that
  \begin{align}
    \overline{(\gamma_5 \otimes 1 )}_{BC}
    &= \delta_{C\overline{B}} \eta_5(B) \,.
  \end{align}
  Hence, we can rewrite Eq.~\eqref{eq:id-1} as follows,
  \begin{align}
    [\gamma_5 \otimes 1]^2_{AC} &=
    \sum_B [ \delta_{B\bar{A}} \eta_5(A)
      \delta_{C\bar{B}} \eta_5(B) ]
    \nonumber \\
    & \qquad \cdot [ U(x_A, x_B) U(x_B, x_C) ]
        \nonumber \\
        &= \delta_{AC} [U(x_A, x_{\bar{A}}) U(x_{\bar{A}},x_A)] \,,
        \label{eq:id-2}
  \end{align}
  where we use the helpful identity $\eta_5(\bar{A}) = \eta_5(A)$.
  Thanks to the $SU(3)$ projection in Eq.~\eqref{eq:link-11},
  $U(x_{\bar{A}},x_A) = [U(x_A, x_{\bar{A}})]^\dagger \in
  SU(3)$.
  Hence, $[U(x_A, x_{\bar{A}}) U(x_{\bar{A}},x_A)] = 1$.
  Therefore, we can rewrite Eq.~\eqref{eq:id-2} as follows,
  \begin{align}
    [\gamma_5 \otimes 1]^2_{AC} &=
    \delta_{AC} = [1 \otimes 1]_{AC} \,.
  \end{align}
  Hence, we have just proven that $[\gamma_5 \otimes 1]^2 =
  [1 \otimes 1]$.
  (Q.E.D.)
\end{proof}

Using the results of Eq.~\eqref{eq:thm-1}, we can prove the recursion
relationship as follows,
\begin{align}
  [\gamma_5 \otimes 1]^{2n+1} &=
  \left([\gamma_5 \otimes 1]^2 \right)^n
  \cdot  [\gamma_5 \otimes 1] \\
  &= \left([1 \otimes 1]\right)^n
  \cdot  [\gamma_5 \otimes 1] \\
  &= [1 \otimes 1]
  \cdot [\gamma_5 \otimes 1] \\
  &= [\gamma_5 \otimes 1] \,.
\end{align}
Using the results of Eq.~\eqref{eq:thm-1}, we can prove another recursion
relationship as follows,
\begin{align}
  [\gamma_5 \otimes 1]^{2n} &=
  \left([\gamma_5 \otimes 1]^2 \right)^n
  \\
  &= \left([1 \otimes 1]\right)^n
  \\
  &= [1 \otimes 1] \,.
\end{align}

Finally, we can prove the following theorem.
\begin{thm}
  \begin{align}
    & [\frac{1+\gamma_5}{2} \otimes 1]
    [ \frac{1+\gamma_5}{2} \otimes 1] =
    [ \frac{1+\gamma_5}{2} \otimes 1]
    \label{thm-2}
  \end{align}
\end{thm}
\begin{proof}
  \begin{align}
    [\frac{1+\gamma_5}{2} \otimes 1]^2
    &= \frac{1}{4} \left([1 \otimes 1]
    + [\gamma_5 \otimes 1] \right)^2
    \nonumber \\
    &= \frac{1}{4} \left( [1 \otimes 1]
    + 2 [\gamma_5 \otimes 1] + [\gamma_5 \otimes 1]^2
    \right)
    \nonumber \\
    &= \frac{1}{2} \left( [1 \otimes 1]
    + [\gamma_5 \otimes 1] \right)
    \nonumber \\
    &= [\frac{1+\gamma_5}{2} \otimes 1] \,.
  \end{align}
  (Q.E.D.)
\end{proof}
Using Eq.~\eqref{thm-2}, we can prove that for integer $n>0$,
\begin{align}
  [\frac{1+\gamma_5}{2} \otimes 1]^n &=
  [\frac{1+\gamma_5}{2} \otimes 1]
\end{align}
by induction.

At this stage, it will be trivial to prove that
\begin{align}
  [\frac{1+\gamma_5}{2} \otimes 1]
  [\frac{1-\gamma_5}{2} \otimes 1] &=0 \,.
\end{align}
%

%\wlee{wlee: begin editing}

The next two theorems concern the chiral Ward identities.
\begin{thm}
  \begin{align}
    [\gamma_5 \otimes \xi_5]
    &= [\gamma_5 \otimes 1] [1  \otimes \xi_5]
    = [1  \otimes \xi_5] [\gamma_5 \otimes 1]
  \end{align}
  \label{thm-3}
\end{thm}
\begin{proof}
  Using the results of Eq.~\eqref{eq:g5x1:1}, we find that
  \begin{align}
    \overline{(\gamma_5 \otimes 1 )}_{AB}
    & = \delta_{B\bar{A}} \; \eta_5(A) \,,
  \end{align}
  where $\bar{A}_\mu = (A_\mu + 1) \mod 2$.
  Let us rewrite $\overline{( 1 \otimes \xi_5 )}_{AB}$ as follows,
  \begin{align}
    \overline{( 1 \otimes \xi_5 )}_{AB}
    & = \frac{1}{4} \Tr(\gamma_A^\dagger 1 \gamma_B \gamma_5^\dagger  )
    \nonumber \\
    & = \delta_{B\bar{A}}
    [\zeta_1(\bar{A}) \zeta_2(\bar{A}) \zeta_3(\bar{A}) \zeta_4(\bar{A})]
    \nonumber \\
    & = \delta_{B\bar{A}} \; \zeta_5(A) \,,
    \label{eq:1xxi5:1} 
  \end{align}
  where $\bar{A}_\mu = (A_\mu + 1) \mod 2$, and
  \begin{align}
    \zeta_\mu(A) &= (-1)^{Y_\mu}, \qquad
    \text{ for $\mu = 1,2,3,4$} \,,
    \\
    Y_\mu &= \sum_{\nu>\mu} A_\nu,
    \\ \zeta_5(A) &= \zeta_1(A) \zeta_2(A) \zeta_3(A)
    \zeta_4(A) = (-1)^{A_2 + A_4} \,,
    \\
    \zeta_5(\bar{A}) &= \zeta_5(A) \,.
  \end{align}
  Hence, we find that
  \begin{align}
    &[\gamma_5 \otimes 1] [1  \otimes \xi_5] \vert_{AC}
    = \sum_{B} [\gamma_5 \otimes 1]_{AB} [1  \otimes \xi_5]_{BC}
    \nonumber \\
    &= \sum_{B}
    \Big\{ \overline{(\gamma_5 \otimes 1 )}_{AB}
    \overline{( 1 \otimes \xi_5 )}_{BC} \Big\}
             [U(x_A,x_B) U(x_B,x_C)]
    \nonumber \\
    &=  \sum_{B} \Big\{ \delta_{B \bar{A}} \; \eta_5(A) \;
    \delta_{C \bar{B}} \; \zeta_5(C) \Big\}
          [U(x_A,x_{\bar{A}}) U(x_{\bar{A}},x_A)]
    \nonumber \\
    &= \delta_{AC} \; \eta_5(A) \; \zeta_5(A)
    \nonumber \\
     &= \delta_{AC} \; \epsilon(A) = [\gamma_5 \otimes \xi_5]_{AC}
  \end{align}
  This is a proof of the first part of the theorem.
  Similarly,
  \begin{align}
    &[1 \otimes \xi_5] [\gamma_5 \otimes 1]  \vert_{AC}
    = \sum_{B} [1  \otimes \xi_5]_{AB} [\gamma_5 \otimes 1]_{BC} 
    \nonumber \\
    &= \sum_{B}
    \Big\{ \overline{( 1 \otimes \xi_5 )}_{AB}
    \overline{(\gamma_5 \otimes 1 )}_{BC} \Big\}
             [U(x_A,x_B) U(x_B,x_C)]
    \nonumber \\
    &=  \sum_{B} \Big\{ \delta_{B \bar{A}} \; \zeta_5(A) \;
    \delta_{C \bar{B}} \; \eta_5(C) \Big\}
          [U(x_A,x_{\bar{A}}) U(x_{\bar{A}},x_A)]
    \nonumber \\
    &= \delta_{AC} \; \eta_5(A) \; \zeta_5(A)
    \nonumber \\
     &= \delta_{AC} \; \epsilon(A) = [\gamma_5 \otimes \xi_5]_{AC}
  \end{align}
  This is a proof of the second part of the theorem.
  \\(Q.E.D.)
\end{proof}

We prove the Ward identities in
Eqs.~\eqref{eq:wi:u(1)a-2}-\eqref{eq:wi:u(1)a-3} as follows.
\begin{thm}
  \begin{align}
    [\gamma_5 \otimes \xi_5] [\gamma_5 \otimes 1 ]
    &= [\gamma_5 \otimes 1 ] [\gamma_5 \otimes \xi_5]
    = [1 \otimes \xi_5] \,,
    \\
    [\gamma_5 \otimes \xi_5] [1 \otimes \xi_5 ]
    &= [1 \otimes \xi_5 ] [\gamma_5 \otimes \xi_5]
    = [\gamma_5 \otimes 1 ] \,.
  \end{align}
  \label{thm-4}
\end{thm}
\begin{proof}
  Using the results of Theorem \ref{thm-3}, we know
  the following Ward identity:
  \begin{align}
    [\gamma_5 \otimes \xi_5] &= [1  \otimes \xi_5] [\gamma_5 \otimes 1]
    \label{eq:wi-thm-3-1}
  \end{align}
  Let us multiply $[\gamma_5 \otimes 1]$ on both sides of
  Eq.~\eqref{eq:wi-thm-3-1}.
  Then,
  \begin{align} 
    [\gamma_5 \otimes \xi_5] [\gamma_5 \otimes 1]  &=
    [1  \otimes \xi_5] [\gamma_5 \otimes 1]^2
    = [1  \otimes \xi_5] \,.
  \end{align}
  Here we use the recursion relationship in Theorem \ref{thm-1}. 
  Similarly, from the results of Theorem \ref{thm-3}, we know
  that
  \begin{align}
    [\gamma_5 \otimes \xi_5]
    &= [\gamma_5 \otimes 1] [1  \otimes \xi_5]
    \label{eq:wi-thm-3-2}
  \end{align}
  Let us multiply $[\gamma_5 \otimes 1]$ on both sides of
  Eq.~\eqref{eq:wi-thm-3-2}.
  \begin{align}
    [\gamma_5 \otimes 1] [\gamma_5 \otimes \xi_5]
    &= [\gamma_5 \otimes 1]^2 [1  \otimes \xi_5]
    = [1  \otimes \xi_5]
  \end{align}
  Here, we use the recursion relation in Theorem \ref{thm-1}.
  This completes a proof of the first part of Theorem \ref{thm-4}.

  Let us multiply $[1 \otimes \xi_5]$ on both sides of
  Eq.~\eqref{eq:wi-thm-3-2}.
  \begin{align}
    [\gamma_5 \otimes \xi_5] [1  \otimes \xi_5]
    &= [\gamma_5 \otimes 1] [1  \otimes \xi_5]^2
    = [\gamma_5 \otimes 1]
  \end{align}
  Similarly, let us multiply $[1 \otimes \xi_5]$ on both sides of
  Eq.~\eqref{eq:wi-thm-3-1}.
  \begin{align}
    [1 \otimes \xi_5] [\gamma_5 \otimes \xi_5]
    &= [1  \otimes \xi_5]^2 [\gamma_5 \otimes 1]
    = [\gamma_5 \otimes 1]
  \end{align}
  This completes a proof of the second part of Theorem \ref{thm-4}.
  \\ (Q.E.D.)
\end{proof}

%\wlee{wlee: end editing}

%-------------------
% END of APPENDIX E.
%-------------------

%------------
% APPENDIX F.
%------------
%\input{app.leakage.zero.tex}
%------------
\section{Examples for the leakage pattern of zero modes}
\label{app:ex-zero}

Let us begin with the case $Q=-2$.
In Fig.~\ref{fig:qm2-leak-g5-f1set}, we show leakage patterns of the
chirality operator for the first set of zero modes at $Q=-2$.
In Fig.~\ref{fig:qm2-leak-xi5-f1set}, we present the leakage patterns
of the shift operator for the first set of zero modes at $Q=-2$.
By comparing Fig.~\ref{fig:qm2-leak-g5-f1set} with
Fig.~\ref{fig:qm2-leak-xi5-f1set}, we find that the chiral Ward
identities of Eqs.~\eqref{eq:wi-12} and \eqref{eq:wi-13} are
well-respected.

%-----------
% FIGURE 13.
%-----------
%\input{figs/fig_qm2_leak_g5_f1set}
%-----------
\begin{figure}[htb]
%  \vspace*{-5mm}
%  \centering
%  \footnotesize
%  \renewcommand{\arraystretch}{1.2}
%  \renewcommand{\subfigcapskip}{-0.55em}
%
%  \hspace*{-7mm}
  \subfigure[$|\Gamma_5|^i_1 = |\Gamma_5 (\lambda_i, \lambda_{1})| $]{
    \includegraphics[width=\linewidth]{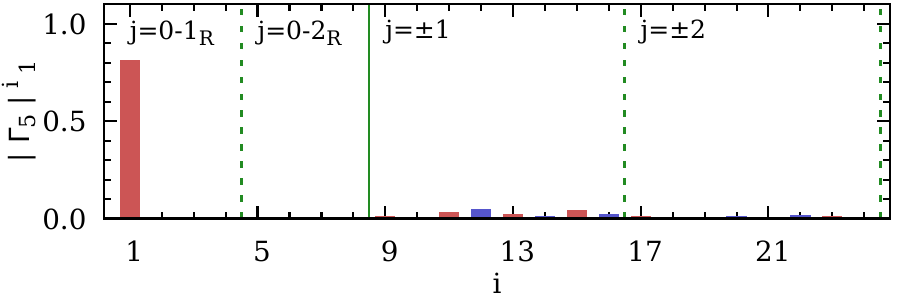}
    \label{sfig:qm2-leak-g5-f1}
  }
  \\
  \subfigure[$|\Gamma_5|^i_3 = |\Gamma_5 (\lambda_i, \lambda_{3})| $]{
    \includegraphics[width=\linewidth]{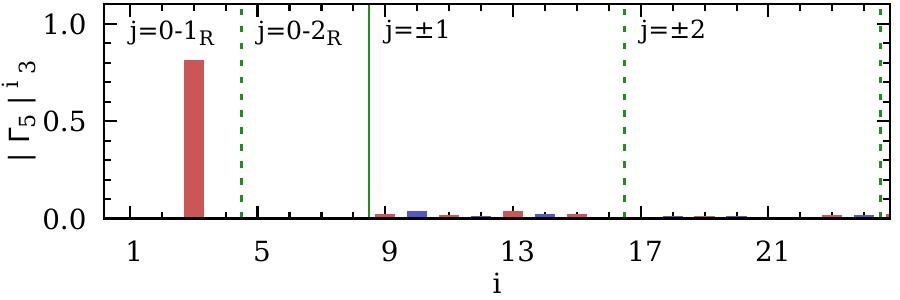}
    \label{sfig:qm2-leak-g5-f3}
  }
  \caption{$[\gamma_5 \otimes 1]$ leakage pattern for the first quartet
      of would-be zero modes at $Q=-2$.}
  \label{fig:qm2-leak-g5-f1set}
\end{figure}
%------------------
% END of FIGURE 13.
%------------------

%-----------
% FIGURE 14.
%-----------
%\input{figs/fig_qm2_leak_xi5_f1set}
%-----------
\begin{figure}[h]
%  \vspace*{-5mm}
%  \centering
%  \footnotesize
%  \renewcommand{\arraystretch}{1.2}
%  \renewcommand{\subfigcapskip}{-0.55em}
%
%  \hspace*{-7mm}
  \subfigure[$|\Xi_5|^i_1 = |\Xi_5 (\lambda_i, \lambda_{1})| $]{
    \includegraphics[width=\linewidth]{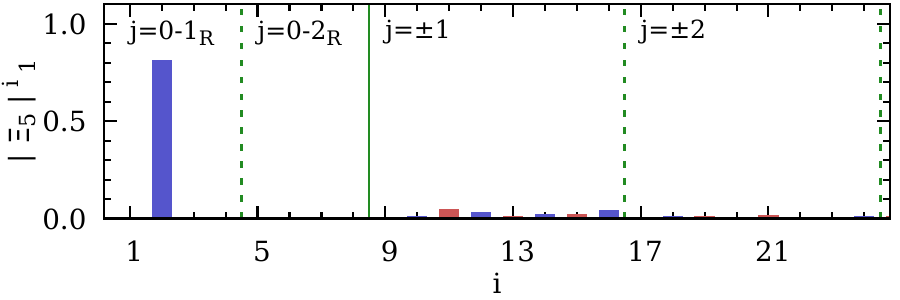}
    \label{sfig:qm2-leak-xi5-f1}
  }
  \\
  \subfigure[$|\Xi_5|^i_3 = |\Xi_5 (\lambda_i, \lambda_{3})| $]{
    \includegraphics[width=\linewidth]{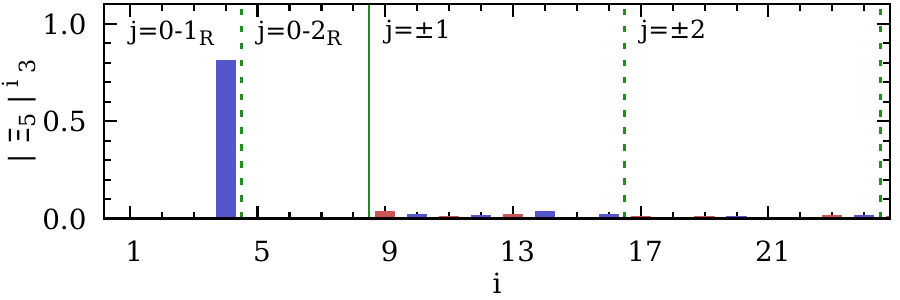}
    \label{sfig:qm2-leak-xi5-f3}
  }
  \caption{$[1 \otimes \xi_5]$ leakage pattern for the first quartet of
    would-be zero modes at $Q=-2$.}
  \label{fig:qm2-leak-xi5-f1set}
\end{figure}
%------------------
% END of FIGURE 14.
%------------------

In Fig.~\ref{fig:qm2-leak-g5-f5set}, we show leakage patterns of
the chirality operator for the second set of zero modes at $Q=-2$.
In Fig.~\ref{fig:qm2-leak-xi5-f5set}, we present the leakage patterns
of the shift operator for the second set of zero modes at $Q=-2$.
By comparing Fig.~\ref{fig:qm2-leak-g5-f5set} with
Fig.~\ref{fig:qm2-leak-xi5-f5set}, we find that the chiral Ward
identities of Eqs.~\eqref{eq:wi-12} and \eqref{eq:wi-13} are
well-preserved.

%-----------
% FIGURE 15.
%-----------
%\input{figs/fig_qm2_leak_g5_f5set}
%-----------
\begin{figure}[h]
%  \vspace*{-5mm}
%  \centering
%  \footnotesize
%  \renewcommand{\arraystretch}{1.2}
%  \renewcommand{\subfigcapskip}{-0.55em}
%
%  \hspace*{-7mm}
  \subfigure[$|\Gamma_5|^i_5 = |\Gamma_5 (\lambda_i, \lambda_{5})| $]{
    \includegraphics[width=\linewidth]{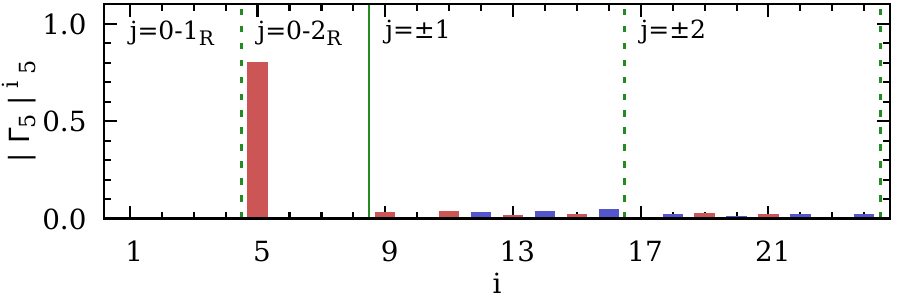}
    \label{sfig:qm2-leak-g5-f5}
  }
  \\
  \subfigure[$|\Gamma_5|^i_7 = |\Gamma_5 (\lambda_i, \lambda_{7})| $]{
    \includegraphics[width=\linewidth]{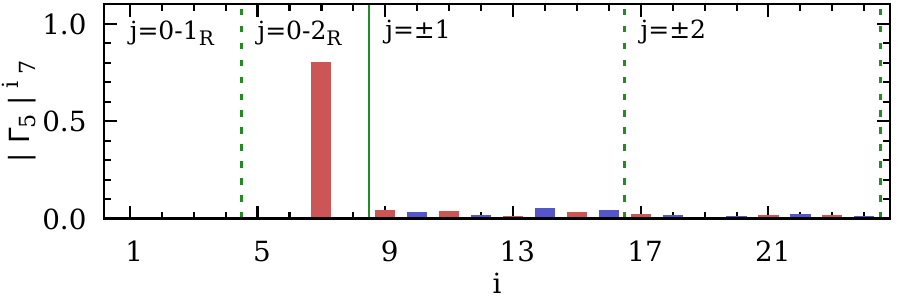}
    \label{sfig:qm2-leak-g5-f7}
  }
  \caption{$[\gamma_5 \otimes 1]$ leakage pattern for the second quartet of
    would-be zero modes at $Q=-2$.}
  \label{fig:qm2-leak-g5-f5set}
\end{figure}
%------------------
% END of FIGURE 15.
%------------------

%-----------
% FIGURE 16.
%-----------
%\input{figs/fig_qm2_leak_xi5_f5set}
%-----------
\begin{figure}[h]
%  \vspace*{-5mm}
%  \centering
%  \footnotesize
%  \renewcommand{\arraystretch}{1.2}
%  \renewcommand{\subfigcapskip}{-0.55em}
%
%  \hspace*{-7mm}
  \subfigure[$|\Xi_5|^i_5 = |\Xi_5 (\lambda_i, \lambda_{5})| $]{
    \includegraphics[width=\linewidth]{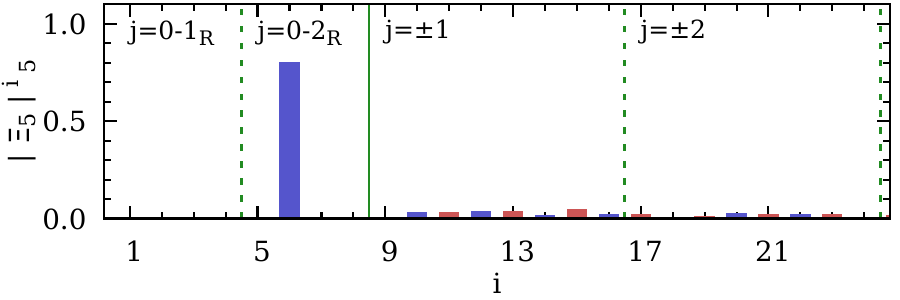}
    \label{sfig:qm2-leak-xi5-f5}
  }
  \\
  \subfigure[$|\Xi_5|^i_7 = |\Xi_5 (\lambda_i, \lambda_{7})| $]{
    \includegraphics[width=\linewidth]{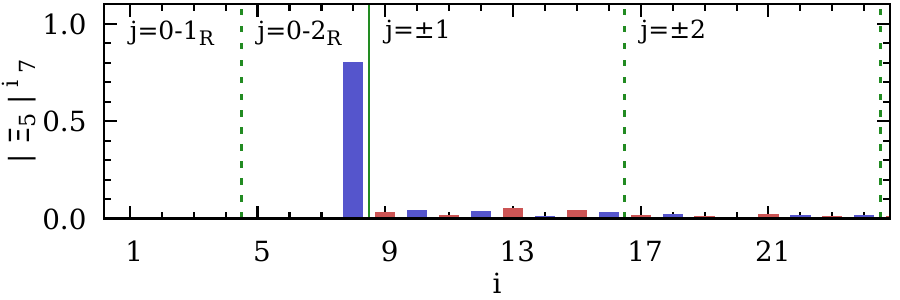}
    \label{sfig:qm2-leak-xi5-f7}
  }
  \caption{$[1 \otimes \xi_5]$ leakage pattern for the second quartet of
    would-be zero modes at $Q=-2$.}
  \label{fig:qm2-leak-xi5-f5set}
\end{figure}
%-------------------
% END of FIGURE 16.
%-------------------

%-----------
% FIGURE 17.
%-----------
%\input{figs/fig_qm3_leak_g5_f9set}
%-----------
\begin{figure}[tbhp]
%  \vspace*{-5mm}
%  \centering
%  \footnotesize
%  \renewcommand{\arraystretch}{1.2}
%  \renewcommand{\subfigcapskip}{-0.55em}
%
%  \hspace*{-7mm}
  \subfigure[$|\Gamma_5|^i_9 = |\Gamma_5 (\lambda_i, \lambda_{9})| $]{
    \includegraphics[width=\linewidth]{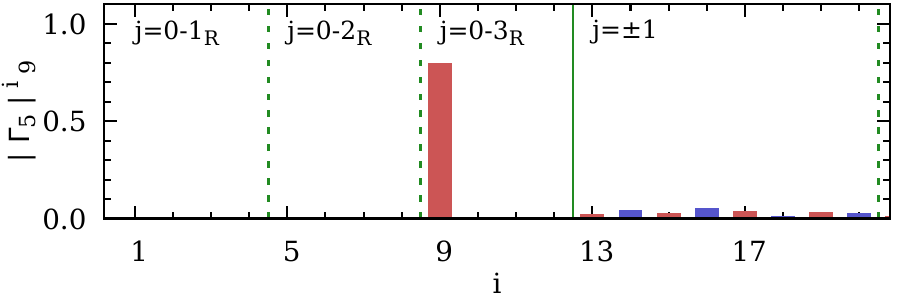}
    \label{sfig:qm3-leak-g5-f9}
  }
  \\
  \subfigure[$|\Gamma_5|^i_{11} = |\Gamma_5 (\lambda_i, \lambda_{11})| $]{
    \includegraphics[width=\linewidth]{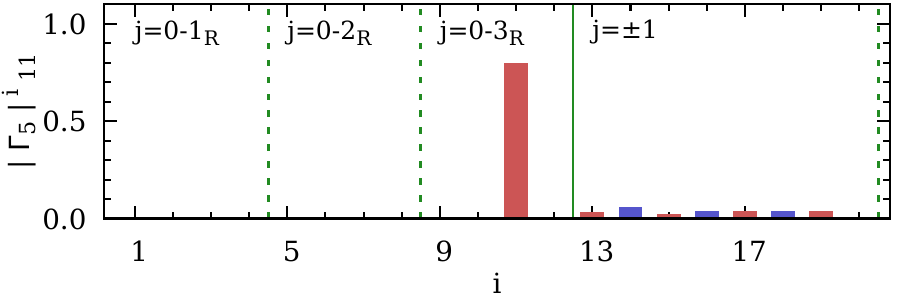}
    \label{sfig:qm3-leak-g5-f11}
  }
  \caption{$[\gamma_5 \otimes 1]$ leakage pattern for the third quartet of
    would-be zero modes at $Q=-3$.}
  \label{fig:qm3-leak-g5-f9set}
\end{figure}
%------------------
% END of FIGURE 17.
%------------------

Now let us consider an example with $Q=-3$.
The leakage patterns for the first and second sets of zero modes are
similar to those at $Q=-2$.
Hence, we choose the third set of zero modes as our example.
In Fig.~\ref{fig:qm3-leak-g5-f9set}, we show leakage patterns of
the chirality operator for the third set of zero modes at $Q=-3$.
In Fig.~\ref{fig:qm3-leak-xi5-f9set}, we present the leakage pattern
of the shift operator for the third set of zero modes at $Q=-3$.
By comparing Fig.~\ref{fig:qm3-leak-g5-f9set} with
Fig.~\ref{fig:qm3-leak-xi5-f9set}, we find that the chiral Ward
identities of Eqs.~\eqref{eq:wi-12} and \eqref{eq:wi-13} are
well-preserved.

%-----------
% FIGURE 18.
%-----------
%\input{figs/fig_qm3_leak_xi5_f9set}
%-----------
\begin{figure}[tbhp]
%  \vspace*{-5mm}
%  \centering
%  \footnotesize
%  \renewcommand{\arraystretch}{1.2}
%  \renewcommand{\subfigcapskip}{-0.55em}
%
%  \hspace*{-7mm}
  \subfigure[$|\Xi_5|^i_9 = |\Xi_5 (\lambda_i, \lambda_{9})| $]{
    \includegraphics[width=\linewidth]{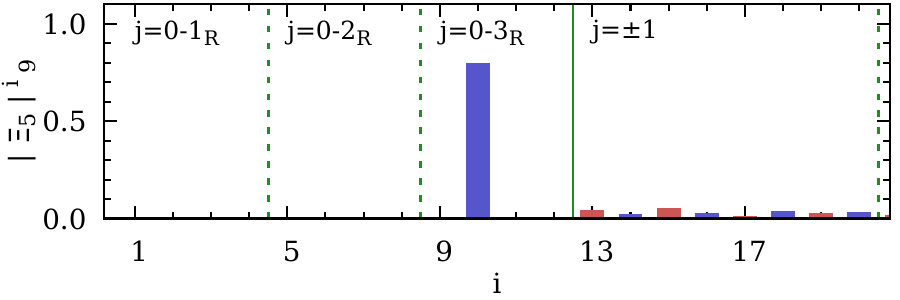}
    \label{sfig:qm3-leak-xi5-f9}
  }
  \\
  \subfigure[$|\Xi_5|^i_{11} = |\Xi_5 (\lambda_i, \lambda_{11})| $]{
    \includegraphics[width=\linewidth]{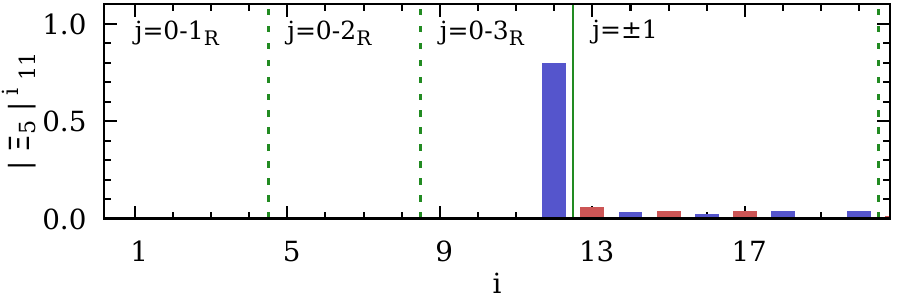}
    \label{sfig:qm3-leak-xi5-f11}
  }
  \caption{$[1 \otimes \xi_5]$ leakage pattern for the third quartet of
    would-be zero modes at $Q=-3$.}
  \label{fig:qm3-leak-xi5-f9set}
\end{figure}
%------------------
% END of FIGURE 18.
%------------------

%-------------------
% END of APPENDIX F.
%-------------------

%------------
% APPENDIX G.
%------------
%\input{app.leakage.nonzero.tex}
%------------
\section{Examples for the leakage pattern of non-zero modes}
\label{app:ex-nonzero}
Let us begin with an example with $Q=0$.
Since the gauge configuration with $Q=0$ usually has no zero mode
($n_- = n_+ = 0$), it is relatively easy to study non-zero modes.
In Fig.~\ref{fig:q0-leak-g5-f1set}, we present leakage patterns of the
chirality operator $\Gamma_5 = [ \gamma_5 \otimes 1 ]$ for non-zero modes
$\{ \lambda_1, \lambda_3, \lambda_5, \lambda_7 \} = \{ \lambda_{j,m} |\;
j=+1,\; m=1,2,3,4 \}$ in the $j = +1$ quartet when $Q=0$.
The results show that the $\Gamma_5$ leakages for non-zero modes
$\lambda_{+1,m}$ mostly go into their parity partners $\{ \lambda_2,
\lambda_4, \lambda_6, \lambda_8 \} = \{ \lambda_{j,m} |\; j=-1,\;
m=1,2,3,4 \}$ in the $j = -1$ quartet.
Meanwhile, the leakages to other quartets such as $j = \pm 2, \pm 3$
are negligibly small compared to those of the $j = -1$ quartet
elements.
This observation is consistent with that for $Q = -1$ in
Fig.~\ref{fig:qm1-leak-g5-f5set}.

%-----------
% FIGURE 19.
%-----------
%\input{figs/fig_q0_leak_g5_f1set}
%-----------
\begin{figure}[t!]
  \subfigure[$|\Gamma_5|^i_1 = |\Gamma_5 (\lambda_i, \lambda_1)| $]{
    \includegraphics[width=\linewidth]{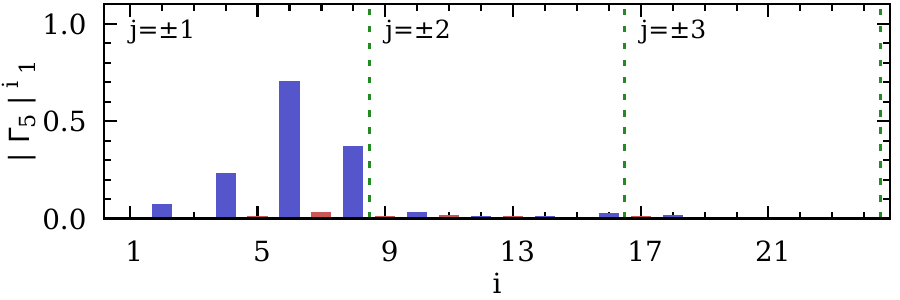}
    \label{sfig:q0-leak-g5-f1}
  }
  \\
  \subfigure[$|\Gamma_5|^i_3 = |\Gamma_5 (\lambda_i, \lambda_3)| $]{
    \includegraphics[width=\linewidth]{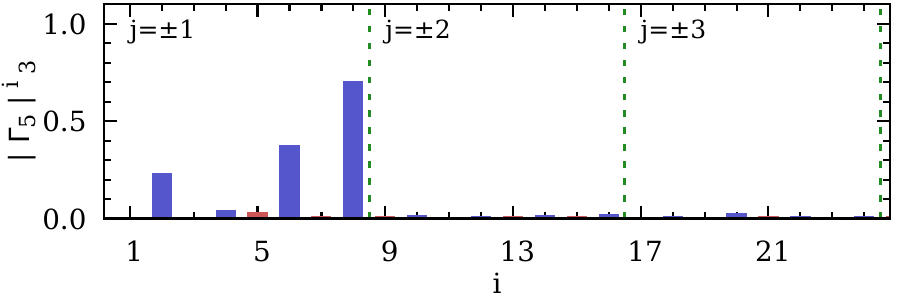}
    \label{sfig:q0-leak-g5-f3}
  }
  \\
  \subfigure[$|\Gamma_5|^i_5 = |\Gamma_5 (\lambda_i, \lambda_5)| $]{
    \includegraphics[width=\linewidth]{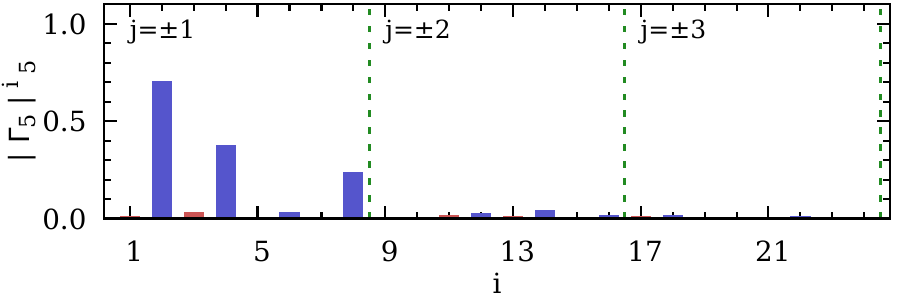}
    \label{sfig:q0-leak-g5-f5}
  }
  \\
  \subfigure[$|\Gamma_5|^i_7 = |\Gamma_5 (\lambda_i, \lambda_{7})|$]{
    \includegraphics[width=\linewidth]{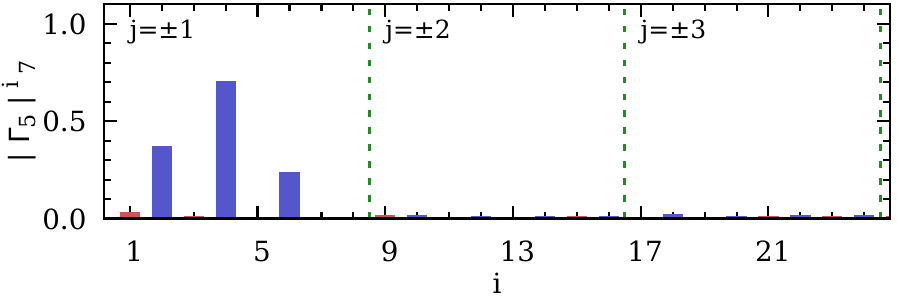}
    \label{sfig:q0-leak-g5-f7}
  }
  \caption{$[\gamma_5 \otimes 1]$ leakage pattern for the first quartet
      of non-zero modes at $Q=0$.}
  \label{fig:q0-leak-g5-f1set}
\end{figure}
%------------------
% END of FIGURE 19.
%------------------

%-----------
% FIGURE 20.
%-----------
%\input{figs/fig_q0_leak_xi5_f1set}
%-----------
\begin{figure}[tbhp]
  \subfigure[$|\Xi_5|^i_1 = |\Xi_5 (\lambda_i, \lambda_1)| $]{
    \includegraphics[width=\linewidth]{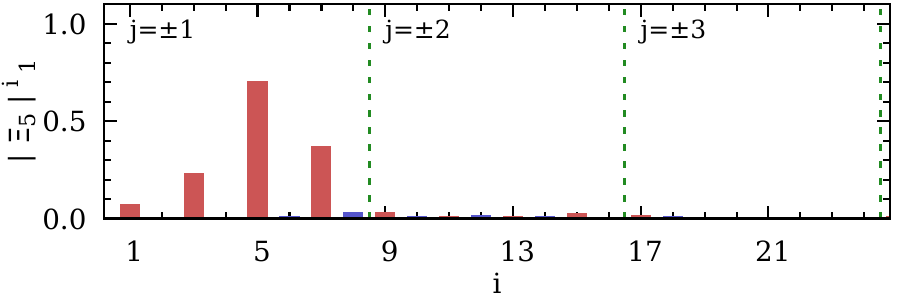}
    \label{sfig:q0-leak-xi5-f1}
  }
  \\
  \subfigure[$|\Xi_5|^i_3 = |\Xi_5 (\lambda_i, \lambda_3)| $]{
    \includegraphics[width=\linewidth]{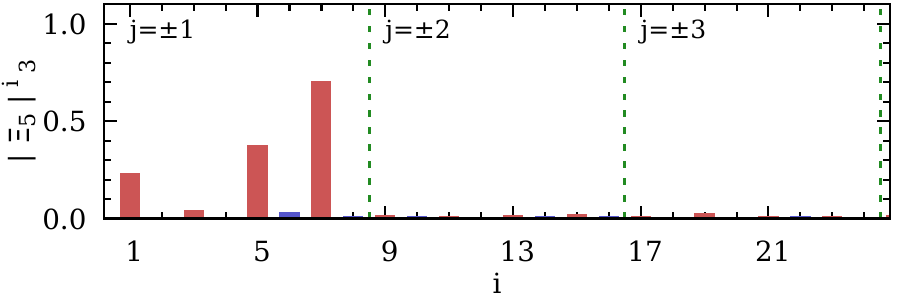}
    \label{sfig:q0-leak-xi5-f3}
  }
  \\
  \subfigure[$|\Xi_5|^i_5 = |\Xi_5 (\lambda_i, \lambda_5)| $]{
    \includegraphics[width=\linewidth]{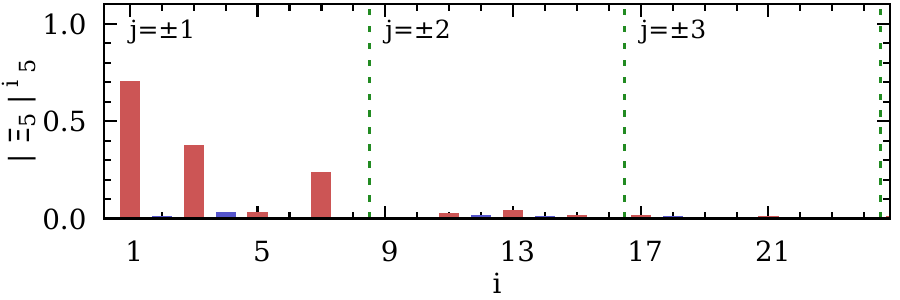}
    \label{sfig:q0-leak-xi5-f5}
  }
  \\
  \subfigure[$|\Xi_5|^i_7 = |\Xi_5 (\lambda_i, \lambda_{7})|$]{
    \includegraphics[width=\linewidth]{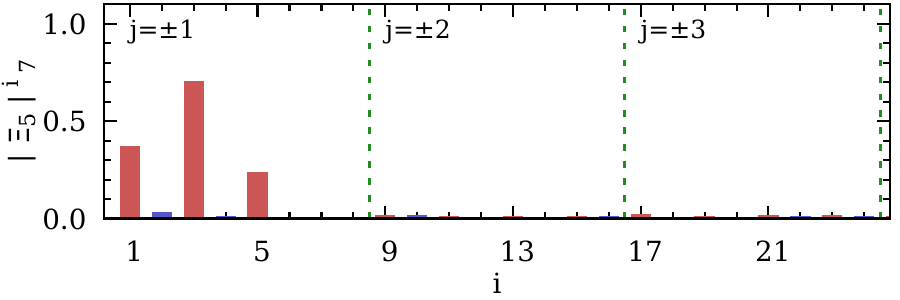}
    \label{sfig:q0-leak-xi5-f7}
  }
  \caption{$[1 \otimes \xi_5]$ leakage pattern for the first quartet of
    non-zero modes at $Q=0$.}
  \label{fig:q0-leak-xi5-f1set}
\end{figure}
%-----------
% FIGURE 20.
%-----------

In Fig.~\ref{fig:q0-leak-xi5-f1set}, we present leakage patterns of
the shift operator $\Xi_5 = [ 1 \otimes \xi_5 ]$ for the non-zero
modes $\{ \lambda_1, \lambda_3, \lambda_5, \lambda_7 \}$ of
$\lambda_{+1,m}$ in the $j = +1$ quartet when $Q=0$.
For the $\Xi_5$ operator, we find the great part of leakages are from
non-zero modes $\lambda_{+1,m}$ to other elements within the $j=+1$
quartet.
Meanwhile, there are only negligible leakages to parity partner
quartet elements ($j = -1$) and other quartets with $j = \pm 2, \pm
3$, and so on.
This observation corresponds to the case $Q = -1$ in
Fig.~\ref{fig:qm1-leak-xi5-f5set}.
We also find that the leakages of $\Gamma_5$ in
Fig.~\ref{fig:q0-leak-g5-f1set} and $\Xi_5$ in
Fig.~\ref{fig:q0-leak-xi5-f1set} are related to each other by the Ward
identity of Eq.~\eqref{eq:ward_jm}.
%

%-----------
% FIGURE 21.
%-----------
%\input{figs/fig_q0_leak_g5_f9set}
%-----------
\begin{figure}[!t]
  \subfigure[$|\Gamma_5|^i_{9} = |\Gamma_5 (\lambda_i, \lambda_{9})| $]{
    \includegraphics[width=\linewidth]{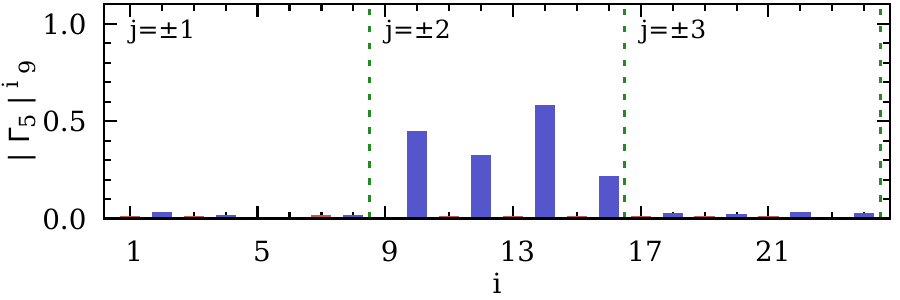}
    \label{sfig:q0-leak-g5-f9}
  }
  \\
  \subfigure[$|\Gamma_5|^i_{11} = |\Gamma_5 (\lambda_i, \lambda_{11})| $]{
    \includegraphics[width=\linewidth]{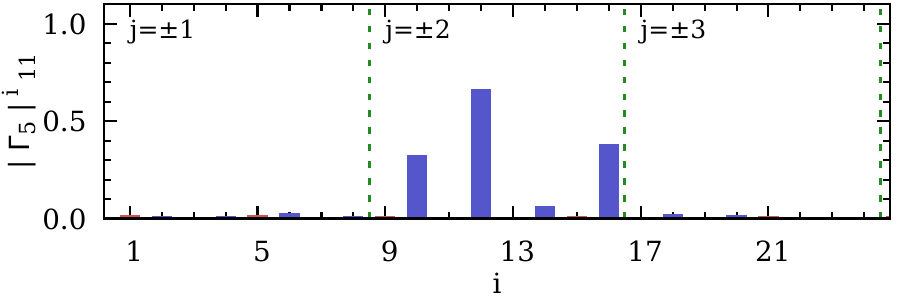}
    \label{sfig:q0-leak-g5-f11}
  }
  \\
  \subfigure[$|\Gamma_5|^i_{13} = |\Gamma_5 (\lambda_i, \lambda_{13})| $]{
    \includegraphics[width=\linewidth]{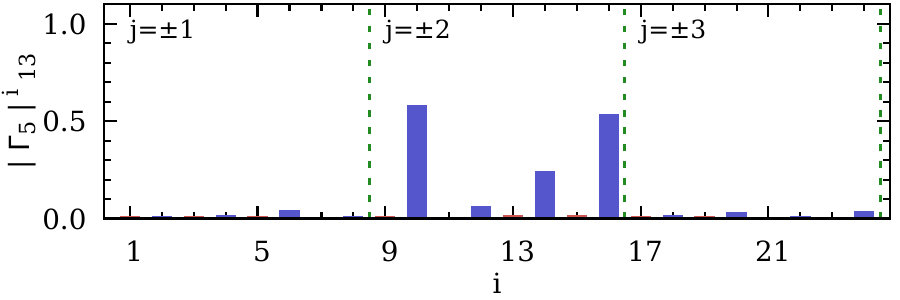}
    \label{sfig:q0-leak-g5-f13}
  }
  \\
  \subfigure[$|\Gamma_5|^i_{15} = |\Gamma_5 (\lambda_i, \lambda_{15})|$]{
    \includegraphics[width=\linewidth]{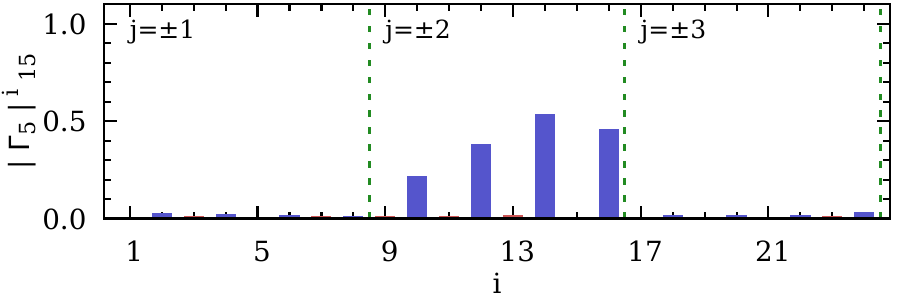}
    \label{sfig:q0-leak-g5-f15}
  }
  \caption{$[\gamma_5 \otimes 1]$ leakage pattern for the second quartet of
    non-zero modes at $Q=0$.}
  \label{fig:q0-leak-g5-f9set}
\end{figure}
%------------------
% END of FIGURE 21.
%------------------

%-----------
% FIGURE 22.
%-----------
%\input{figs/fig_q0_leak_xi5_f9set}
%-----------
\begin{figure}[tbhp]
  \subfigure[$|\Xi_5|^i_{9} = |\Xi_5 (\lambda_i, \lambda_{9})| $]{
    \includegraphics[width=\linewidth]{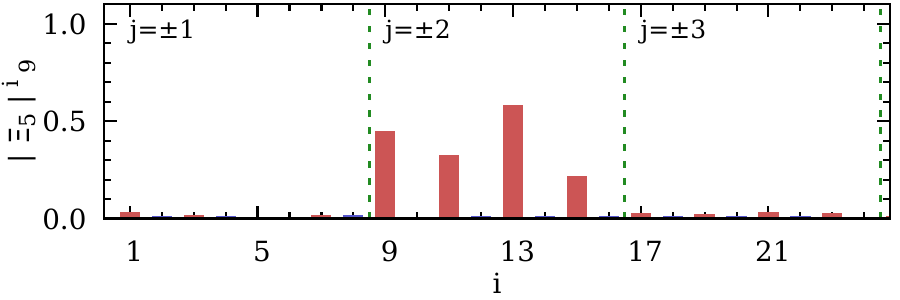}
    \label{sfig:q0-leak-xi5-f9}
  }
  \\
  \subfigure[$|\Xi_5|^i_{11} = |\Xi_5 (\lambda_i, \lambda_{11})| $]{
    \includegraphics[width=\linewidth]{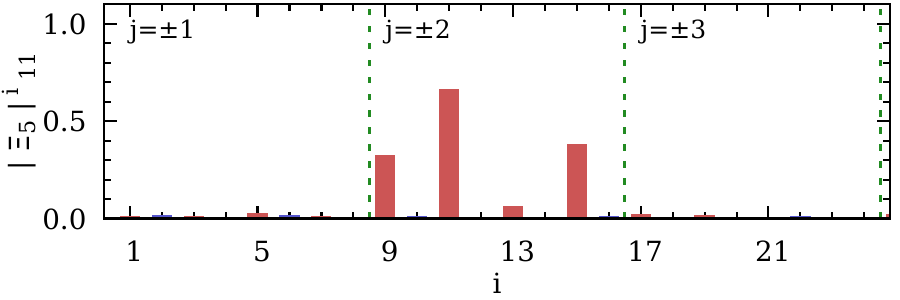}
    \label{sfig:q0-leak-xi5-f11}
  }
  \\
  \subfigure[$|\Xi_5|^i_{13} = |\Xi_5 (\lambda_i, \lambda_{13})| $]{
    \includegraphics[width=\linewidth]{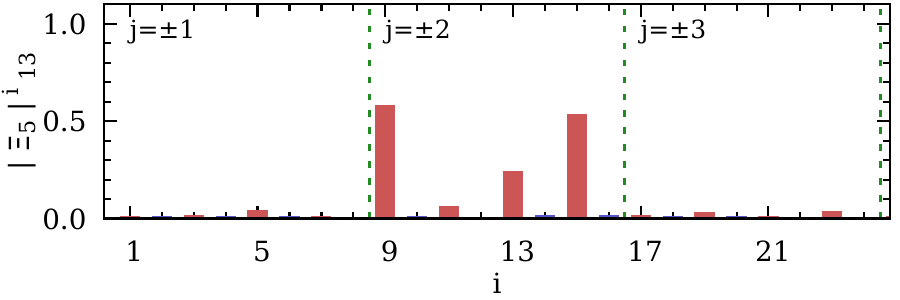}
    \label{sfig:q0-leak-xi5-f13}
  }
  \\
  \subfigure[$|\Xi_5|^i_{15} = |\Xi_5 (\lambda_i, \lambda_{15})|$]{
    \includegraphics[width=\linewidth]{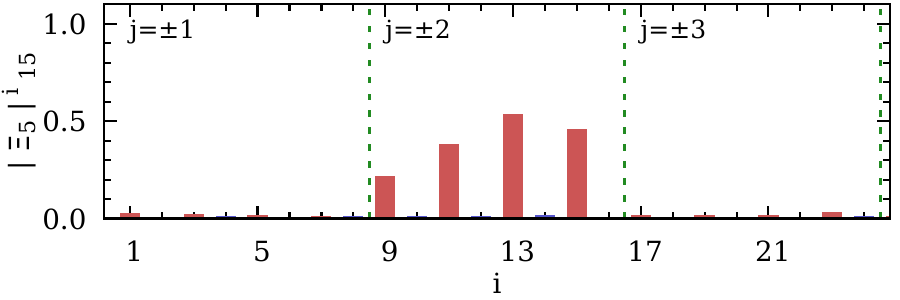}
    \label{sfig:q0-leak-xi5-f15}
  }
  \caption{$[1 \otimes \xi_5]$ leakage pattern for the second quartet of
    non-zero modes at $Q=0$.}
  \label{fig:q0-leak-xi5-f9set}
\end{figure}
%------------------
% END of FIGURE 22.
%------------------

In Figs.~\ref{fig:q0-leak-g5-f9set} and \ref{fig:q0-leak-xi5-f9set},
we present leakage patterns of the $\Gamma_5$ and $\Xi_5$ operators,
respectively, for non-zero modes $\{ \lambda_9, \lambda_{11},
\lambda_{13}, \lambda_{15} \} = \{ \lambda_{j,m} |\; j=+2,\; m=1,2,3,4
\}$ in the $j = +2$ quartet when $Q=0$.
Similar to the above cases for $j = +1$, $\Gamma_5$ leakages for
non-zero modes of $j = +2$ mostly go to their parity partner quartet
elements of $j = -2$: $\{ \lambda_{10}, \lambda_{12}, \lambda_{14},
\lambda_{16} \} = \{ \lambda_{j,m} |\; j=-2,\; m=1,2,3,4 \}$, and
$\Xi_5$ leakages mostly go to members within the $j = +2$ quartet: $\{
\lambda_9, \lambda_{11}, \lambda_{13}, \lambda_{15} \}$.
There are only negligible leakages to other quartets for both
operators.
%

%-----------
% FIGURE 23.
%-----------
%\input{figs/fig_qm2_leak_g5_f9set}
%------------------
\begin{figure}[h]
  \subfigure[$|\Gamma_5|^i_{9} = |\Gamma_5 (\lambda_i, \lambda_{9})| $]{
    \includegraphics[width=\linewidth]{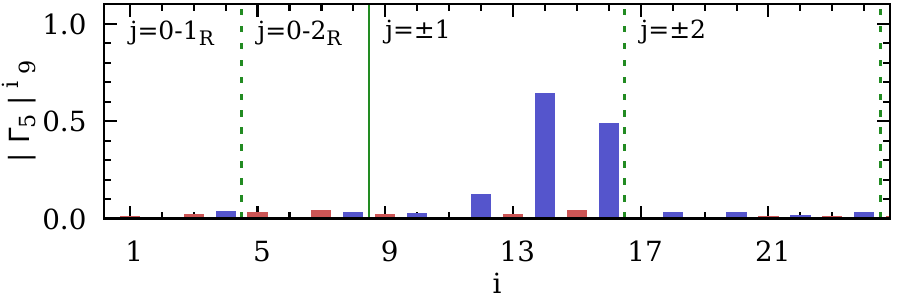}
    \label{sfig:qm2-leak-g5-f9}
  }
  \\
  \subfigure[$|\Gamma_5|^i_{11} = |\Gamma_5 (\lambda_i, \lambda_{11})| $]{
    \includegraphics[width=\linewidth]{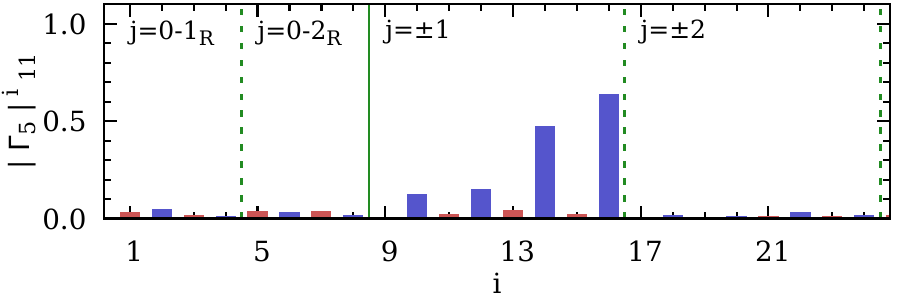}
    \label{sfig:qm2-leak-g5-f11}
  }
  \\
  \subfigure[$|\Gamma_5|^i_{13} = |\Gamma_5 (\lambda_i, \lambda_{13})| $]{
    \includegraphics[width=\linewidth]{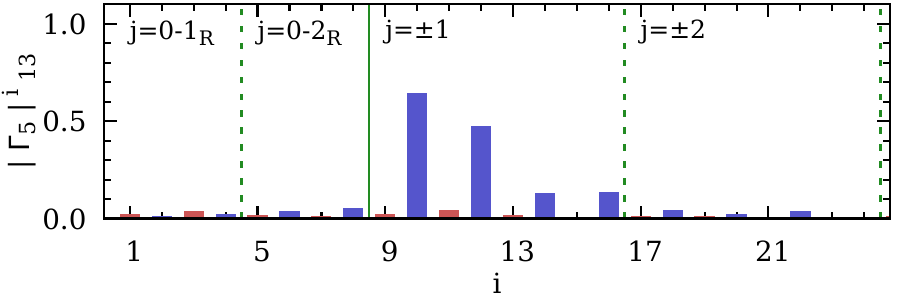}
    \label{sfig:qm2-leak-g5-f13}
  }
  \\
  \subfigure[$|\Gamma_5|^i_{15} = |\Gamma_5 (\lambda_i, \lambda_{15})|$]{
    \includegraphics[width=\linewidth]{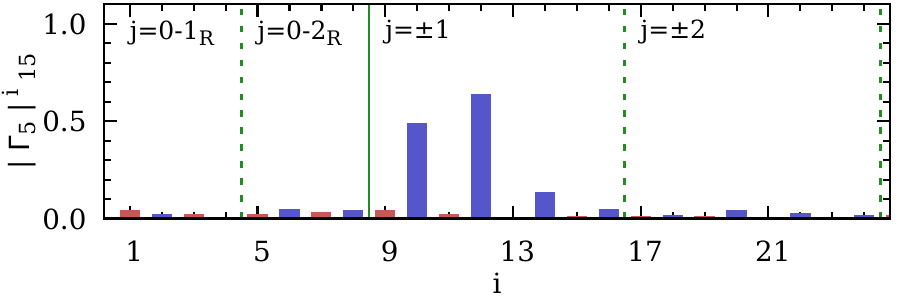}
    \label{sfig:qm2-leak-g5-f15}
  }
  \caption{$[\gamma_5 \otimes 1]$ leakage pattern for the first quartet of
    non-zero modes at $Q=-2$.}
  \label{fig:qm2-leak-g5-f9set}
\end{figure}
%------------------
% END of FIGURE 23.
%------------------

Now let us examine the leakage patterns when would-be zero modes exist
($Q \ne 0$).
In Figs.~\ref{fig:qm2-leak-g5-f9set} and \ref{fig:qm2-leak-xi5-f9set},
we present leakage patterns of the $\Gamma_5$ and $\Xi_5$ operators,
respectively, for non-zero modes $\{ \lambda_9, \lambda_{11},
\lambda_{13}, \lambda_{15} \}$ in the $j = +1$ quartet when $Q = -2$.
There are two quartets of right-handed would-be zero modes where $j =
0-1_R$ and $0-2_R$, which corresponds to $n_-=0$ and $n_+=2$ with $Q =
-2$ by the index theorem ($Q = n_- - n_+$).

%-----------
% FIGURE 24.
%-----------
%\input{figs/fig_qm2_leak_xi5_f9set}
%-----------
\begin{figure}[h]
  \subfigure[$|\Xi_5|^i_{9} = |\Xi_5 (\lambda_i, \lambda_{9})| $]{
    \includegraphics[width=\linewidth]{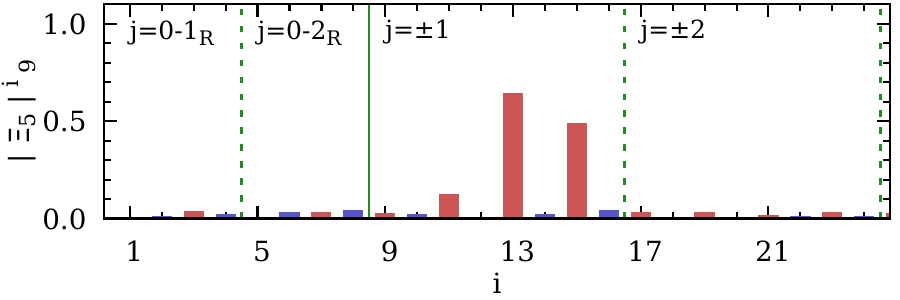}
    \label{sfig:qm2-leak-xi5-f9}
  }
  \\
  \subfigure[$|\Xi_5|^i_{11} = |\Xi_5 (\lambda_i, \lambda_{11})| $]{
    \includegraphics[width=\linewidth]{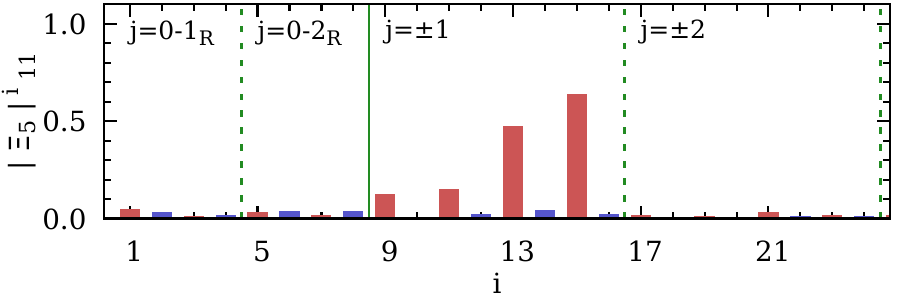}
    \label{sfig:qm2-leak-xi5-f11}
  }
  \\
  \subfigure[$|\Xi_5|^i_{13} = |\Xi_5 (\lambda_i, \lambda_{13})| $]{
    \includegraphics[width=\linewidth]{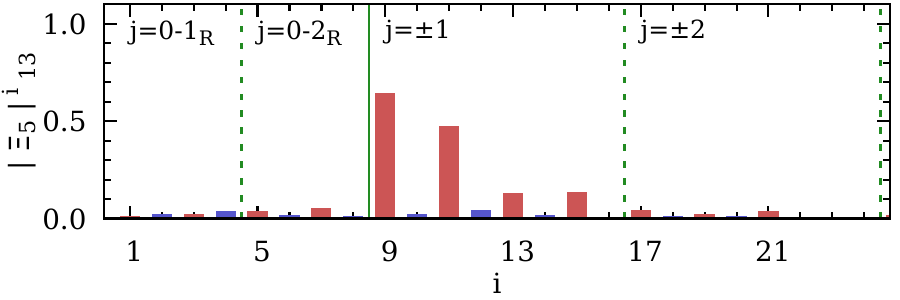}
    \label{sfig:qm2-leak-xi5-f13}
  }
  \\
  \subfigure[$|\Xi_5|^i_{15} = |\Xi_5 (\lambda_i, \lambda_{15})|$]{
    \includegraphics[width=\linewidth]{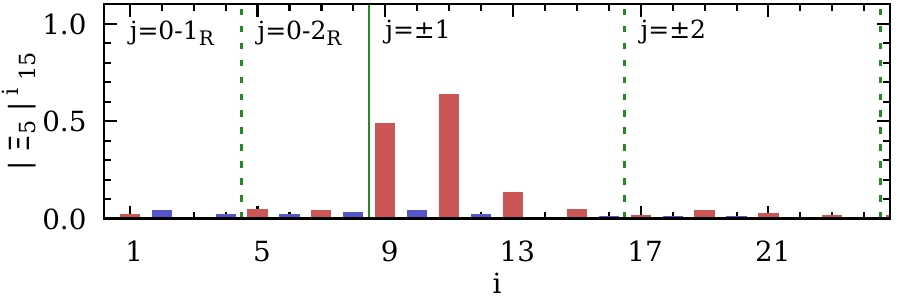}
    \label{sfig:qm2-leak-xi5-f15}
  }
  \caption{$[1 \otimes \xi_5]$ leakage pattern for the first quartet of
    non-zero modes at $Q=-2$.}
  \label{fig:qm2-leak-xi5-f9set}
\end{figure}
%------------------
% END of FIGURE 24.
%------------------

%
As in the cases $Q = -1$ (Figs.~\ref{fig:qm1-leak-g5-f5set} and
\ref{fig:qm1-leak-xi5-f5set}) and $Q = 0$
(Figs.~\ref{fig:q0-leak-g5-f1set} and \ref{fig:q0-leak-xi5-f1set}),
$\Gamma_5$ leakages from non-zero modes of $j = +1$ mostly go to the
parity partner $j = -1$ quartet, and $\Xi_5$ leakages from non-zero
modes of $j = +1$ mostly go within the $j = +1$ quartet itself.
Leakages to other non-zero mode quartets and would-be zero mode quartets
are negligibly small.
We also find that the Ward identity between the two leakage patterns
holds.
%

%-----------
% FIGURE 25.
%-----------
%\input{figs/fig_qm3_leak_g5_f13set}
%-----------
\begin{figure}[t]
  \subfigure[$|\Gamma_5|^i_{13} = |\Gamma_5 (\lambda_i, \lambda_{13})| $]{
    \includegraphics[width=\linewidth]{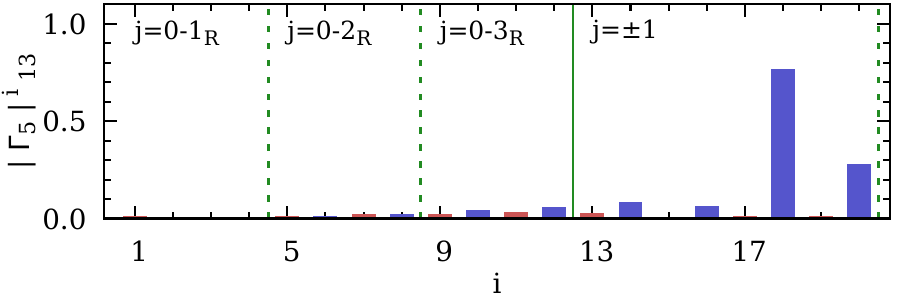}
    \label{sfig:qm3-leak-g5-f13}
  }
  \\
  \subfigure[$|\Gamma_5|^i_{15} = |\Gamma_5 (\lambda_i, \lambda_{15})| $]{
    \includegraphics[width=\linewidth]{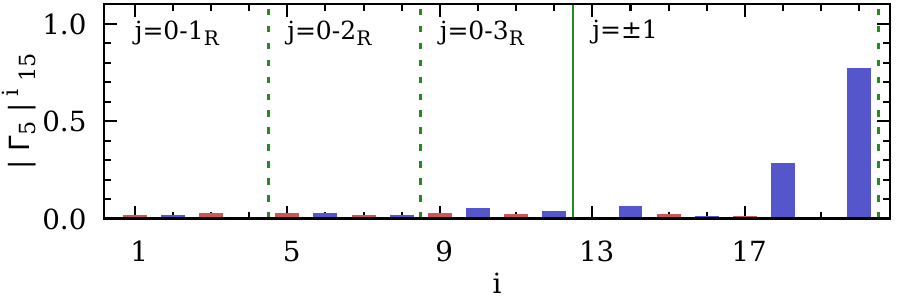}
    \label{sfig:qm3-leak-g5-f15}
  }
  \\
  \subfigure[$|\Gamma_5|^i_{17} = |\Gamma_5 (\lambda_i, \lambda_{17})| $]{
    \includegraphics[width=\linewidth]{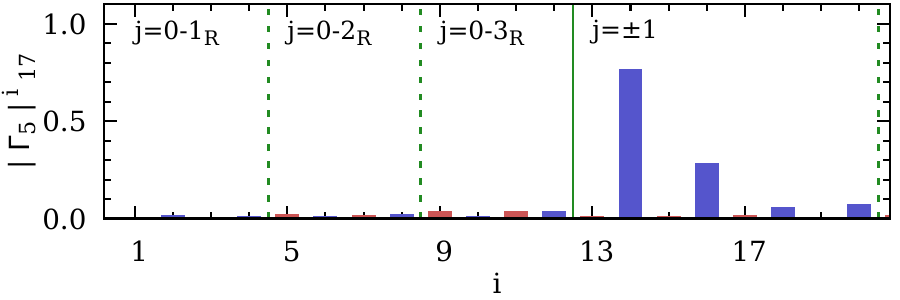}
    \label{sfig:qm3-leak-g5-f17}
  }
  \\
  \subfigure[$|\Gamma_5|^i_{19} = |\Gamma_5 (\lambda_i, \lambda_{19})|$]{
    \includegraphics[width=\linewidth]{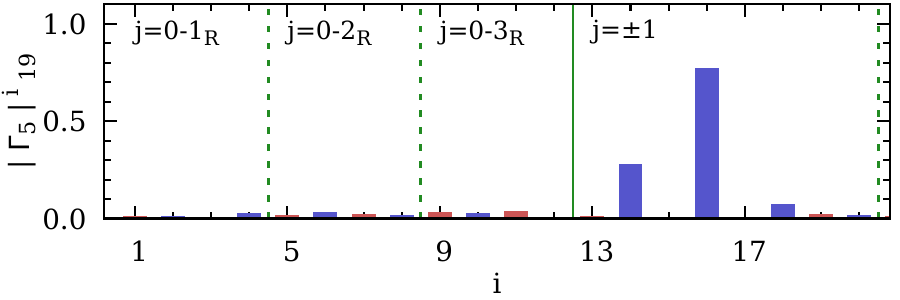}
    \label{sfig:qm3-leak-g5-f19}
  }
  \caption{$[\gamma_5 \otimes 1]$ leakage pattern for the first quartet of
    non-zero modes at $Q=-3$.}
  \label{fig:qm3-leak-g5-f13set}
\end{figure}
%------------------
% END of FIGURE 25.
%------------------

%-----------
% FIGURE 26.
%-----------
%\input{figs/fig_qm3_leak_xi5_f13set}
%-----------
\begin{figure}[t]
  \subfigure[$|\Xi_5|^i_{13} = |\Xi_5 (\lambda_i, \lambda_{13})| $]{
    \includegraphics[width=\linewidth]{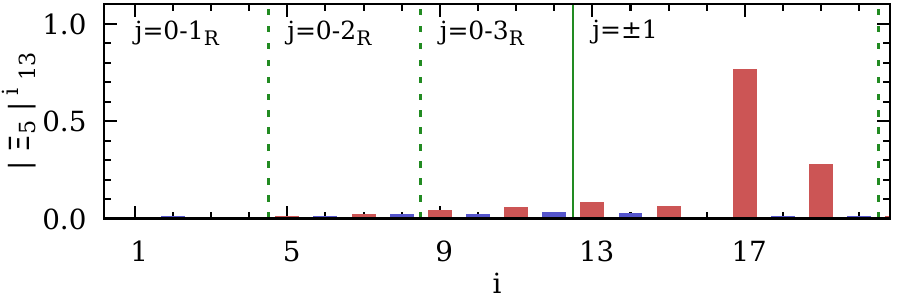}
    \label{sfig:qm3-leak-xi5-f13}
  }
  \\
  \subfigure[$|\Xi_5|^i_{15} = |\Xi_5 (\lambda_i, \lambda_{15})| $]{
    \includegraphics[width=\linewidth]{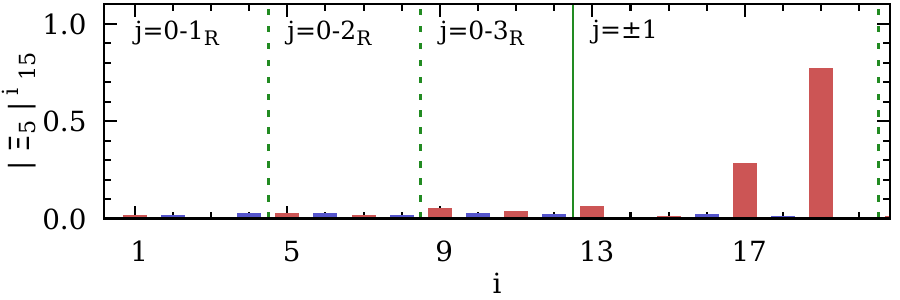}
    \label{sfig:qm3-leak-xi5-f15}
  }
  \\
  \subfigure[$|\Xi_5|^i_{17} = |\Xi_5 (\lambda_i, \lambda_{17})| $]{
    \includegraphics[width=\linewidth]{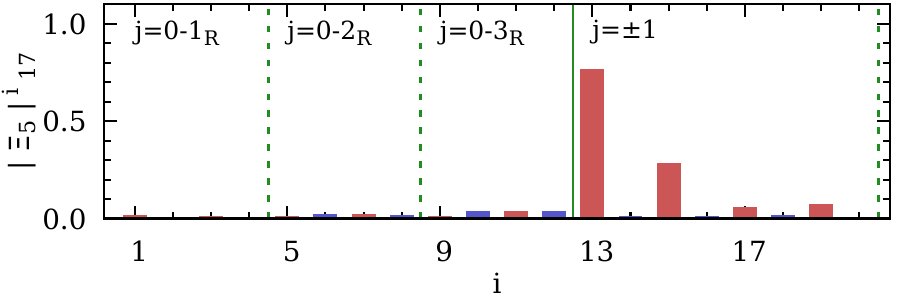}
    \label{sfig:qm3-leak-xi5-f17}
  }
  \\
  \subfigure[$|\Xi_5|^i_{19} = |\Xi_5 (\lambda_i, \lambda_{19})|$]{
    \includegraphics[width=\linewidth]{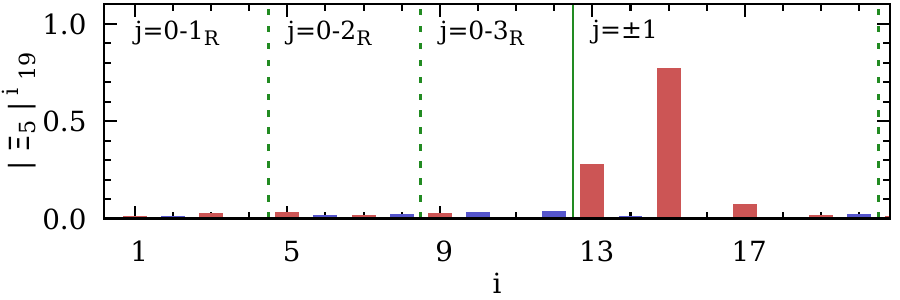}
    \label{sfig:qm3-leak-xi5-f19}
  }
  \caption{$[1 \otimes \xi_5]$ leakage pattern for the first quartet of
    non-zero modes at $Q=-3$.}
  \label{fig:qm3-leak-xi5-f13set}
\end{figure}
%------------------
% END of FIGURE 26.
%------------------

In Figs.~\ref{fig:qm3-leak-g5-f13set} and
\ref{fig:qm3-leak-xi5-f13set}, we present leakage patterns for the
$\Gamma_5$ and $\Xi_5$ operators, respectively, for non-zero modes $\{
\lambda_{13}, \lambda_{15}, \lambda_{17}, \lambda_{19} \}$ in the $j =
+1$ quartet when $Q = -3$.
Their leakage patterns are also consistent with those for $Q=0,\,-1,\,
-2$ in our previous discussion.

%-------------------
% END of APPENDIX G.
%-------------------

%-----------
% reference
%-----------
\bibliography{ref} %%% ref.bib file

\begin{thebibliography}{63}%
\makeatletter
\providecommand \@ifxundefined [1]{%
 \@ifx{#1\undefined}
}%
\providecommand \@ifnum [1]{%
 \ifnum #1\expandafter \@firstoftwo
 \else \expandafter \@secondoftwo
 \fi
}%
\providecommand \@ifx [1]{%
 \ifx #1\expandafter \@firstoftwo
 \else \expandafter \@secondoftwo
 \fi
}%
\providecommand \natexlab [1]{#1}%
\providecommand \enquote  [1]{``#1''}%
\providecommand \bibnamefont  [1]{#1}%
\providecommand \bibfnamefont [1]{#1}%
\providecommand \citenamefont [1]{#1}%
\providecommand \href@noop [0]{\@secondoftwo}%
\providecommand \href [0]{\begingroup \@sanitize@url \@href}%
\providecommand \@href[1]{\@@startlink{#1}\@@href}%
\providecommand \@@href[1]{\endgroup#1\@@endlink}%
\providecommand \@sanitize@url [0]{\catcode `\\12\catcode `\$12\catcode
  `\&12\catcode `\#12\catcode `\^12\catcode `\_12\catcode `\%12\relax}%
\providecommand \@@startlink[1]{}%
\providecommand \@@endlink[0]{}%
\providecommand \url  [0]{\begingroup\@sanitize@url \@url }%
\providecommand \@url [1]{\endgroup\@href {#1}{\urlprefix }}%
\providecommand \urlprefix  [0]{URL }%
\providecommand \Eprint [0]{\href }%
\providecommand \doibase [0]{https://doi.org/}%
\providecommand \selectlanguage [0]{\@gobble}%
\providecommand \bibinfo  [0]{\@secondoftwo}%
\providecommand \bibfield  [0]{\@secondoftwo}%
\providecommand \translation [1]{[#1]}%
\providecommand \BibitemOpen [0]{}%
\providecommand \bibitemStop [0]{}%
\providecommand \bibitemNoStop [0]{.\EOS\space}%
\providecommand \EOS [0]{\spacefactor3000\relax}%
\providecommand \BibitemShut  [1]{\csname bibitem#1\endcsname}%
\let\auto@bib@innerbib\@empty
%</preamble>
\bibitem [{\citenamefont {Atiyah}\ and\ \citenamefont
  {Singer}(1963)}]{Atiyah:1963zz}%
  \BibitemOpen
  \bibfield  {author} {\bibinfo {author} {\bibfnamefont {M.~F.}\ \bibnamefont
  {Atiyah}}\ and\ \bibinfo {author} {\bibfnamefont {I.~M.}\ \bibnamefont
  {Singer}},\ }\bibfield  {title} {\bibinfo {title} {The index of elliptic
  operators on compact manifolds},\ }\href
  {https://doi.org/10.1090/S0002-9904-1963-10957-X} {\bibfield  {journal}
  {\bibinfo  {journal} {Bull. Am. Math. Soc.}\ }\textbf {\bibinfo {volume}
  {69}},\ \bibinfo {pages} {422} (\bibinfo {year} {1963})}\BibitemShut
  {NoStop}%
%%CITATION = BAMOA,69,422;%%
\bibitem [{\citenamefont {Banks}\ and\ \citenamefont
  {Casher}(1980)}]{Banks:1979yr}%
  \BibitemOpen
  \bibfield  {author} {\bibinfo {author} {\bibfnamefont {T.}~\bibnamefont
  {Banks}}\ and\ \bibinfo {author} {\bibfnamefont {A.}~\bibnamefont {Casher}},\
  }\bibfield  {title} {\bibinfo {title} {Chiral symmetry breaking in confining
  theories},\ }\href {https://doi.org/10.1016/0550-3213(80)90255-2} {\bibfield
  {journal} {\bibinfo  {journal} {Nucl. Phys.}\ }\textbf {\bibinfo {volume}
  {B169}},\ \bibinfo {pages} {103} (\bibinfo {year} {1980})}\BibitemShut
  {NoStop}%
%%CITATION = NUPHA,B169,103;%%
\bibitem [{\citenamefont {Shuryak}\ and\ \citenamefont
  {Verbaarschot}(1993)}]{Shuryak:1992pi}%
  \BibitemOpen
  \bibfield  {author} {\bibinfo {author} {\bibfnamefont {E.~V.}\ \bibnamefont
  {Shuryak}}\ and\ \bibinfo {author} {\bibfnamefont {J.~J.~M.}\ \bibnamefont
  {Verbaarschot}},\ }\bibfield  {title} {\bibinfo {title} {{Random matrix
  theory and spectral sum rules for the Dirac operator in QCD}},\ }\href
  {https://doi.org/10.1016/0375-9474(93)90098-I} {\bibfield  {journal}
  {\bibinfo  {journal} {Nucl. Phys.}\ }\textbf {\bibinfo {volume} {A560}},\
  \bibinfo {pages} {306} (\bibinfo {year} {1993})},\ \Eprint
  {https://arxiv.org/abs/hep-th/9212088} {arXiv:hep-th/9212088 [hep-th]}
  \BibitemShut {NoStop}%
%%CITATION = HEP-TH/9212088;%%
\bibitem [{\citenamefont {Leutwyler}\ and\ \citenamefont
  {Smilga}(1992)}]{Leutwyler:1992yt}%
  \BibitemOpen
  \bibfield  {author} {\bibinfo {author} {\bibfnamefont {H.}~\bibnamefont
  {Leutwyler}}\ and\ \bibinfo {author} {\bibfnamefont {A.~V.}\ \bibnamefont
  {Smilga}},\ }\bibfield  {title} {\bibinfo {title} {{Spectrum of Dirac
  operator and role of winding number in QCD}},\ }\href
  {https://doi.org/10.1103/PhysRevD.46.5607} {\bibfield  {journal} {\bibinfo
  {journal} {Phys. Rev.}\ }\textbf {\bibinfo {volume} {D46}},\ \bibinfo {pages}
  {5607} (\bibinfo {year} {1992})}\BibitemShut {NoStop}%
%%CITATION = PHRVA,D46,5607;%%
\bibitem [{\citenamefont {Hasenfratz}\ and\ \citenamefont
  {Knechtli}(2001)}]{Hasenfratz:2001hp}%
  \BibitemOpen
  \bibfield  {author} {\bibinfo {author} {\bibfnamefont {A.}~\bibnamefont
  {Hasenfratz}}\ and\ \bibinfo {author} {\bibfnamefont {F.}~\bibnamefont
  {Knechtli}},\ }\bibfield  {title} {\bibinfo {title} {{Flavor symmetry and the
  static potential with hypercubic blocking}},\ }\href
  {https://doi.org/10.1103/PhysRevD.64.034504} {\bibfield  {journal} {\bibinfo
  {journal} {Phys. Rev.}\ }\textbf {\bibinfo {volume} {D64}},\ \bibinfo {pages}
  {034504} (\bibinfo {year} {2001})},\ \Eprint
  {https://arxiv.org/abs/hep-lat/0103029} {arXiv:hep-lat/0103029 [hep-lat]}
  \BibitemShut {NoStop}%
%%CITATION = HEP-LAT/0103029;%%
\bibitem [{\citenamefont {Lee}\ and\ \citenamefont
  {Sharpe}(2002)}]{Lee:2002ui}%
  \BibitemOpen
  \bibfield  {author} {\bibinfo {author} {\bibfnamefont {W.-j.}\ \bibnamefont
  {Lee}}\ and\ \bibinfo {author} {\bibfnamefont {S.~R.}\ \bibnamefont
  {Sharpe}},\ }\bibfield  {title} {\bibinfo {title} {{One loop matching
  coefficients for improved staggered bilinears}},\ }\href
  {https://doi.org/10.1103/PhysRevD.66.114501} {\bibfield  {journal} {\bibinfo
  {journal} {Phys. Rev.}\ }\textbf {\bibinfo {volume} {D66}},\ \bibinfo {pages}
  {114501} (\bibinfo {year} {2002})},\ \Eprint
  {https://arxiv.org/abs/hep-lat/0208018} {arXiv:hep-lat/0208018 [hep-lat]}
  \BibitemShut {NoStop}%
%%CITATION = HEP-LAT/0208018;%%
\bibitem [{\citenamefont {Follana}\ \emph {et~al.}(2007)\citenamefont
  {Follana}, \citenamefont {Mason}, \citenamefont {Davies}, \citenamefont
  {Hornbostel}, \citenamefont {Lepage}, \citenamefont {Shigemitsu},
  \citenamefont {Trottier},\ and\ \citenamefont {Wong}}]{Follana:2006rc}%
  \BibitemOpen
  \bibfield  {author} {\bibinfo {author} {\bibfnamefont {E.}~\bibnamefont
  {Follana}}, \bibinfo {author} {\bibfnamefont {Q.}~\bibnamefont {Mason}},
  \bibinfo {author} {\bibfnamefont {C.}~\bibnamefont {Davies}}, \bibinfo
  {author} {\bibfnamefont {K.}~\bibnamefont {Hornbostel}}, \bibinfo {author}
  {\bibfnamefont {G.~P.}\ \bibnamefont {Lepage}}, \bibinfo {author}
  {\bibfnamefont {J.}~\bibnamefont {Shigemitsu}}, \bibinfo {author}
  {\bibfnamefont {H.}~\bibnamefont {Trottier}},\ and\ \bibinfo {author}
  {\bibfnamefont {K.}~\bibnamefont {Wong}} (\bibinfo {collaboration} {HPQCD,
  UKQCD}),\ }\bibfield  {title} {\bibinfo {title} {{Highly improved staggered
  quarks on the lattice, with applications to charm physics}},\ }\href
  {https://doi.org/10.1103/PhysRevD.75.054502} {\bibfield  {journal} {\bibinfo
  {journal} {Phys. Rev.}\ }\textbf {\bibinfo {volume} {D75}},\ \bibinfo {pages}
  {054502} (\bibinfo {year} {2007})},\ \Eprint
  {https://arxiv.org/abs/hep-lat/0610092} {arXiv:hep-lat/0610092 [hep-lat]}
  \BibitemShut {NoStop}%
%%CITATION = HEP-LAT/0610092;%%
\bibitem [{\citenamefont {Follana}\ \emph {et~al.}(2004)\citenamefont
  {Follana}, \citenamefont {Hart},\ and\ \citenamefont
  {Davies}}]{Follana:2004sz}%
  \BibitemOpen
  \bibfield  {author} {\bibinfo {author} {\bibfnamefont {E.}~\bibnamefont
  {Follana}}, \bibinfo {author} {\bibfnamefont {A.}~\bibnamefont {Hart}},\ and\
  \bibinfo {author} {\bibfnamefont {C.~T.~H.}\ \bibnamefont {Davies}} (\bibinfo
  {collaboration} {HPQCD, UKQCD}),\ }\bibfield  {title} {\bibinfo {title} {{The
  Index theorem and universality properties of the low-lying eigenvalues of
  improved staggered quarks}},\ }\href
  {https://doi.org/10.1103/PhysRevLett.93.241601} {\bibfield  {journal}
  {\bibinfo  {journal} {Phys. Rev. Lett.}\ }\textbf {\bibinfo {volume} {93}},\
  \bibinfo {pages} {241601} (\bibinfo {year} {2004})},\ \Eprint
  {https://arxiv.org/abs/hep-lat/0406010} {arXiv:hep-lat/0406010 [hep-lat]}
  \BibitemShut {NoStop}%
%%CITATION = HEP-LAT/0406010;%%
\bibitem [{\citenamefont {Follana}\ \emph {et~al.}(2005)\citenamefont
  {Follana}, \citenamefont {Hart}, \citenamefont {Davies},\ and\ \citenamefont
  {Mason}}]{Follana:2005km}%
  \BibitemOpen
  \bibfield  {author} {\bibinfo {author} {\bibfnamefont {E.}~\bibnamefont
  {Follana}}, \bibinfo {author} {\bibfnamefont {A.}~\bibnamefont {Hart}},
  \bibinfo {author} {\bibfnamefont {C.~T.~H.}\ \bibnamefont {Davies}},\ and\
  \bibinfo {author} {\bibfnamefont {Q.}~\bibnamefont {Mason}} (\bibinfo
  {collaboration} {HPQCD, UKQCD}),\ }\bibfield  {title} {\bibinfo {title} {{The
  Low-lying Dirac spectrum of staggered quarks}},\ }\href
  {https://doi.org/10.1103/PhysRevD.72.054501} {\bibfield  {journal} {\bibinfo
  {journal} {Phys. Rev.}\ }\textbf {\bibinfo {volume} {D72}},\ \bibinfo {pages}
  {054501} (\bibinfo {year} {2005})},\ \Eprint
  {https://arxiv.org/abs/hep-lat/0507011} {arXiv:hep-lat/0507011 [hep-lat]}
  \BibitemShut {NoStop}%
%%CITATION = HEP-LAT/0507011;%%
\bibitem [{\citenamefont {Durr}\ \emph {et~al.}(2004)\citenamefont {Durr},
  \citenamefont {Hoelbling},\ and\ \citenamefont {Wenger}}]{Durr:2004as}%
  \BibitemOpen
  \bibfield  {author} {\bibinfo {author} {\bibfnamefont {S.}~\bibnamefont
  {Durr}}, \bibinfo {author} {\bibfnamefont {C.}~\bibnamefont {Hoelbling}},\
  and\ \bibinfo {author} {\bibfnamefont {U.}~\bibnamefont {Wenger}},\
  }\bibfield  {title} {\bibinfo {title} {{Staggered eigenvalue mimicry}},\
  }\href {https://doi.org/10.1103/PhysRevD.70.094502} {\bibfield  {journal}
  {\bibinfo  {journal} {Phys. Rev.}\ }\textbf {\bibinfo {volume} {D70}},\
  \bibinfo {pages} {094502} (\bibinfo {year} {2004})},\ \Eprint
  {https://arxiv.org/abs/hep-lat/0406027} {arXiv:hep-lat/0406027 [hep-lat]}
  \BibitemShut {NoStop}%
%%CITATION = HEP-LAT/0406027;%%
\bibitem [{\citenamefont {Azcoiti}\ \emph {et~al.}(2015)\citenamefont
  {Azcoiti}, \citenamefont {Di~Carlo}, \citenamefont {Follana},\ and\
  \citenamefont {Vaquero}}]{Azcoiti:2014pfa}%
  \BibitemOpen
  \bibfield  {author} {\bibinfo {author} {\bibfnamefont {V.}~\bibnamefont
  {Azcoiti}}, \bibinfo {author} {\bibfnamefont {G.}~\bibnamefont {Di~Carlo}},
  \bibinfo {author} {\bibfnamefont {E.}~\bibnamefont {Follana}},\ and\ \bibinfo
  {author} {\bibfnamefont {A.}~\bibnamefont {Vaquero}},\ }\bibfield  {title}
  {\bibinfo {title} {{Topological Index Theorem on the Lattice through the
  Spectral Flow of Staggered Fermions}},\ }\href
  {https://doi.org/10.1016/j.physletb.2015.03.049} {\bibfield  {journal}
  {\bibinfo  {journal} {Phys. Lett.}\ }\textbf {\bibinfo {volume} {B744}},\
  \bibinfo {pages} {303} (\bibinfo {year} {2015})},\ \Eprint
  {https://arxiv.org/abs/1410.5733} {arXiv:1410.5733 [hep-lat]} \BibitemShut
  {NoStop}%
%%CITATION = ARXIV:1410.5733;%%
\bibitem [{\citenamefont {Durr}(2013)}]{Durr:2013gp}%
  \BibitemOpen
  \bibfield  {author} {\bibinfo {author} {\bibfnamefont {S.}~\bibnamefont
  {Durr}},\ }\bibfield  {title} {\bibinfo {title} {{Taste-split staggered
  actions: eigenvalues, chiralities and Symanzik improvement}},\ }\href
  {https://doi.org/10.1103/PhysRevD.87.114501} {\bibfield  {journal} {\bibinfo
  {journal} {Phys. Rev.}\ }\textbf {\bibinfo {volume} {D87}},\ \bibinfo {pages}
  {114501} (\bibinfo {year} {2013})},\ \Eprint
  {https://arxiv.org/abs/1302.0773} {arXiv:1302.0773 [hep-lat]} \BibitemShut
  {NoStop}%
%%CITATION = ARXIV:1302.0773;%%
\bibitem [{\citenamefont {Adams}(2010)}]{Adams:2009eb}%
  \BibitemOpen
  \bibfield  {author} {\bibinfo {author} {\bibfnamefont {D.~H.}\ \bibnamefont
  {Adams}},\ }\bibfield  {title} {\bibinfo {title} {{Theoretical foundation for
  the Index Theorem on the lattice with staggered fermions}},\ }\href
  {https://doi.org/10.1103/PhysRevLett.104.141602} {\bibfield  {journal}
  {\bibinfo  {journal} {Phys. Rev. Lett.}\ }\textbf {\bibinfo {volume} {104}},\
  \bibinfo {pages} {141602} (\bibinfo {year} {2010})},\ \Eprint
  {https://arxiv.org/abs/0912.2850} {arXiv:0912.2850 [hep-lat]} \BibitemShut
  {NoStop}%
%%CITATION = ARXIV:0912.2850;%%
\bibitem [{\citenamefont {Cundy}\ \emph {et~al.}(2016)\citenamefont {Cundy},
  \citenamefont {Jeong},\ and\ \citenamefont {Lee}}]{Cundy:2016tmw}%
  \BibitemOpen
  \bibfield  {author} {\bibinfo {author} {\bibfnamefont {N.~D.}\ \bibnamefont
  {Cundy}}, \bibinfo {author} {\bibfnamefont {H.}~\bibnamefont {Jeong}},\ and\
  \bibinfo {author} {\bibfnamefont {W.}~\bibnamefont {Lee}},\ }\bibfield
  {title} {\bibinfo {title} {{Calculation of strange and light quark condensate
  using improved staggered fermions and overlap fermions}},\ }\bibfield
  {booktitle} {\emph {\bibinfo {booktitle} {{Proceedings, 33rd International
  Symposium on Lattice Field Theory (Lattice 2015): Kobe, Japan, July 14-18,
  2015}}},\ }\href {https://doi.org/10.22323/1.251.0066} {\bibfield  {journal}
  {\bibinfo  {journal} {PoS}\ }\textbf {\bibinfo {volume} {LATTICE2015}},\
  \bibinfo {pages} {066} (\bibinfo {year} {2016})}\BibitemShut {NoStop}%
%%CITATION = POSCI,LATTICE2015,066;%%
\bibitem [{\citenamefont {Jeong}\ \emph {et~al.}(2017)\citenamefont {Jeong},
  \citenamefont {Jwa}, \citenamefont {Kim}, \citenamefont {Kim}, \citenamefont
  {Lee}, \citenamefont {Lee},\ and\ \citenamefont {Pak}}]{Jeong:2017kst}%
  \BibitemOpen
  \bibfield  {author} {\bibinfo {author} {\bibfnamefont {H.}~\bibnamefont
  {Jeong}}, \bibinfo {author} {\bibfnamefont {S.}~\bibnamefont {Jwa}}, \bibinfo
  {author} {\bibfnamefont {J.}~\bibnamefont {Kim}}, \bibinfo {author}
  {\bibfnamefont {S.}~\bibnamefont {Kim}}, \bibinfo {author} {\bibfnamefont
  {S.}~\bibnamefont {Lee}}, \bibinfo {author} {\bibfnamefont {W.}~\bibnamefont
  {Lee}},\ and\ \bibinfo {author} {\bibfnamefont {J.}~\bibnamefont {Pak}}
  (\bibinfo {collaboration} {SWME}),\ }\bibfield  {title} {\bibinfo {title}
  {{How to identify zero modes for improved staggered fermions}},\ }in\
  \href@noop {} {\emph {\bibinfo {booktitle} {{Proceedings, 35th International
  Symposium on Lattice Field Theory (Lattice 2017): Granada, Spain, June 18-24,
  2017}}}}\ (\bibinfo {year} {2017})\ \Eprint
  {https://arxiv.org/abs/1711.01826} {arXiv:1711.01826 [hep-lat]} \BibitemShut
  {NoStop}%
%%CITATION = ARXIV:1711.01826;%%
\bibitem [{\citenamefont {Jeong}\ \emph {et~al.}(2019)\citenamefont {Jeong},
  \citenamefont {Jung}, \citenamefont {Kim}, \citenamefont {Lee},\ and\
  \citenamefont {Pak}}]{Jeong:2020map}%
  \BibitemOpen
  \bibfield  {author} {\bibinfo {author} {\bibfnamefont {H.}~\bibnamefont
  {Jeong}}, \bibinfo {author} {\bibfnamefont {C.}~\bibnamefont {Jung}},
  \bibinfo {author} {\bibfnamefont {S.}~\bibnamefont {Kim}}, \bibinfo {author}
  {\bibfnamefont {W.}~\bibnamefont {Lee}},\ and\ \bibinfo {author}
  {\bibfnamefont {J.}~\bibnamefont {Pak}} (\bibinfo {collaboration} {SWME}),\
  }\bibfield  {title} {\bibinfo {title} {{Chiral Ward identities for Dirac
  eigenmodes with staggered fermions}},\ }\bibfield  {booktitle} {\emph
  {\bibinfo {booktitle} {{37th International Symposium on Lattice Field Theory
  (Lattice 2019) Wuhan, Hubei, China, June 16-22, 2019}}},\ }\href@noop {}
  {\bibfield  {journal} {\bibinfo  {journal} {PoS}\ }\textbf {\bibinfo {volume}
  {LATTICE2019}},\ \bibinfo {pages} {031} (\bibinfo {year} {2019})},\ \Eprint
  {https://arxiv.org/abs/2001.06568} {arXiv:2001.06568 [hep-lat]} \BibitemShut
  {NoStop}%
%%CITATION = ARXIV:2001.06568;%%
\bibitem [{\citenamefont {Weinberg}(2013)}]{Weinberg:1996kr}%
  \BibitemOpen
  \bibfield  {author} {\bibinfo {author} {\bibfnamefont {S.}~\bibnamefont
  {Weinberg}},\ }\href@noop {} {\emph {\bibinfo {title} {{The quantum theory of
  fields. Vol. 2: Modern applications}}}}\ (\bibinfo  {publisher} {Cambridge
  University Press},\ \bibinfo {year} {2013})\BibitemShut {NoStop}%
%%CITATION = INSPIRE-430948;%%
\bibitem [{\citenamefont {Bazavov}\ \emph {et~al.}(2010)\citenamefont {Bazavov}
  \emph {et~al.}}]{Bazavov:2009bb}%
  \BibitemOpen
  \bibfield  {author} {\bibinfo {author} {\bibfnamefont {A.}~\bibnamefont
  {Bazavov}} \emph {et~al.} (\bibinfo {collaboration} {MILC}),\ }\bibfield
  {title} {\bibinfo {title} {{Nonperturbative QCD Simulations with 2+1 Flavors
  of Improved Staggered Quarks}},\ }\href
  {https://doi.org/10.1103/RevModPhys.82.1349} {\bibfield  {journal} {\bibinfo
  {journal} {Rev. Mod. Phys.}\ }\textbf {\bibinfo {volume} {82}},\ \bibinfo
  {pages} {1349} (\bibinfo {year} {2010})},\ \Eprint
  {https://arxiv.org/abs/0903.3598} {arXiv:0903.3598 [hep-lat]} \BibitemShut
  {NoStop}%
%%CITATION = ARXIV:0903.3598;%%
\bibitem [{\citenamefont {Golterman}\ and\ \citenamefont
  {Smit}(1984)}]{Golterman:1984cy}%
  \BibitemOpen
  \bibfield  {author} {\bibinfo {author} {\bibfnamefont {M.~F.~L.}\
  \bibnamefont {Golterman}}\ and\ \bibinfo {author} {\bibfnamefont
  {J.}~\bibnamefont {Smit}},\ }\bibfield  {title} {\bibinfo {title}
  {{Selfenergy and Flavor Interpretation of Staggered Fermions}},\ }\href
  {https://doi.org/10.1016/0550-3213(84)90424-3} {\bibfield  {journal}
  {\bibinfo  {journal} {Nucl. Phys.}\ }\textbf {\bibinfo {volume} {B245}},\
  \bibinfo {pages} {61} (\bibinfo {year} {1984})}\BibitemShut {NoStop}%
%%CITATION = NUPHA,B245,61;%%
\bibitem [{\citenamefont {Lee}\ and\ \citenamefont
  {Sharpe}(1999)}]{Lee:1999zxa}%
  \BibitemOpen
  \bibfield  {author} {\bibinfo {author} {\bibfnamefont {W.-J.}\ \bibnamefont
  {Lee}}\ and\ \bibinfo {author} {\bibfnamefont {S.~R.}\ \bibnamefont
  {Sharpe}},\ }\bibfield  {title} {\bibinfo {title} {{Partial flavor symmetry
  restoration for chiral staggered fermions}},\ }\href
  {https://doi.org/10.1103/PhysRevD.60.114503} {\bibfield  {journal} {\bibinfo
  {journal} {Phys. Rev.}\ }\textbf {\bibinfo {volume} {D60}},\ \bibinfo {pages}
  {114503} (\bibinfo {year} {1999})},\ \Eprint
  {https://arxiv.org/abs/hep-lat/9905023} {arXiv:hep-lat/9905023 [hep-lat]}
  \BibitemShut {NoStop}%
%%CITATION = HEP-LAT/9905023;%%
\bibitem [{\citenamefont {Lanczos}(1950)}]{Lanczos:1950zz}%
  \BibitemOpen
  \bibfield  {author} {\bibinfo {author} {\bibfnamefont {C.}~\bibnamefont
  {Lanczos}},\ }\bibfield  {title} {\bibinfo {title} {{An iteration method for
  the solution of the eigenvalue problem of linear differential and integral
  operators}},\ }\href {https://doi.org/10.6028/jres.045.026} {\bibfield
  {journal} {\bibinfo  {journal} {J. Res. Natl. Bur. Stand. B Math. Sci.}\
  }\textbf {\bibinfo {volume} {45}},\ \bibinfo {pages} {255} (\bibinfo {year}
  {1950})}\BibitemShut {NoStop}%
%%CITATION = JNBBA,45,255;%%
\bibitem [{\citenamefont {DeGrand}\ and\ \citenamefont
  {Rossi}(1990)}]{DeGrand:1990dk}%
  \BibitemOpen
  \bibfield  {author} {\bibinfo {author} {\bibfnamefont {T.~A.}\ \bibnamefont
  {DeGrand}}\ and\ \bibinfo {author} {\bibfnamefont {P.}~\bibnamefont
  {Rossi}},\ }\bibfield  {title} {\bibinfo {title} {{Conditioning Techniques
  for Dynamical Fermions}},\ }\href
  {https://doi.org/10.1016/0010-4655(90)90006-M} {\bibfield  {journal}
  {\bibinfo  {journal} {Comput. Phys. Commun.}\ }\textbf {\bibinfo {volume}
  {60}},\ \bibinfo {pages} {211} (\bibinfo {year} {1990})}\BibitemShut
  {NoStop}%
%%CITATION = CPHCB,60,211;%%
\bibitem [{\citenamefont {Saad}(1984{\natexlab{a}})}]{Youcef:1984}%
  \BibitemOpen
  \bibfield  {author} {\bibinfo {author} {\bibfnamefont {Y.}~\bibnamefont
  {Saad}},\ }\bibfield  {title} {\bibinfo {title} {{Chebyshev acceleration
  techniques for solving nonsymmetric eigenvalue problems}},\ }\href
  {https://doi.org/10.1090/2007602} {\bibfield  {journal} {\bibinfo  {journal}
  {Math. Comp.}\ }\textbf {\bibinfo {volume} {42}},\ \bibinfo {pages} {567}
  (\bibinfo {year} {1984}{\natexlab{a}})}\BibitemShut {NoStop}%
\bibitem [{\citenamefont {Smit}\ and\ \citenamefont
  {Vink}(1987)}]{Smit:1986fn}%
  \BibitemOpen
  \bibfield  {author} {\bibinfo {author} {\bibfnamefont {J.}~\bibnamefont
  {Smit}}\ and\ \bibinfo {author} {\bibfnamefont {J.~C.}\ \bibnamefont
  {Vink}},\ }\bibfield  {title} {\bibinfo {title} {{Remnants of the Index
  Theorem on the Lattice}},\ }\href
  {https://doi.org/10.1016/0550-3213(87)90451-2} {\bibfield  {journal}
  {\bibinfo  {journal} {Nucl. Phys.}\ }\textbf {\bibinfo {volume} {B286}},\
  \bibinfo {pages} {485} (\bibinfo {year} {1987})}\BibitemShut {NoStop}%
%%CITATION = NUPHA,B286,485;%%
\bibitem [{\citenamefont {Golterman}(1986{\natexlab{a}})}]{Golterman:1985dz}%
  \BibitemOpen
  \bibfield  {author} {\bibinfo {author} {\bibfnamefont {M.~F.~L.}\
  \bibnamefont {Golterman}},\ }\bibfield  {title} {\bibinfo {title} {{STAGGERED
  MESONS}},\ }\href {https://doi.org/10.1016/0550-3213(86)90383-4} {\bibfield
  {journal} {\bibinfo  {journal} {Nucl. Phys.}\ }\textbf {\bibinfo {volume}
  {B273}},\ \bibinfo {pages} {663} (\bibinfo {year}
  {1986}{\natexlab{a}})}\BibitemShut {NoStop}%
%%CITATION = NUPHA,B273,663;%%
\bibitem [{\citenamefont {Golterman}(1986{\natexlab{b}})}]{Golterman:1986jf}%
  \BibitemOpen
  \bibfield  {author} {\bibinfo {author} {\bibfnamefont {M.~F.~L.}\
  \bibnamefont {Golterman}},\ }\bibfield  {title} {\bibinfo {title}
  {{Irreducible Representations of the Staggered Fermion Symmetry Group}},\
  }\href {https://doi.org/10.1016/0550-3213(86)90220-8} {\bibfield  {journal}
  {\bibinfo  {journal} {Nucl. Phys.}\ }\textbf {\bibinfo {volume} {B278}},\
  \bibinfo {pages} {417} (\bibinfo {year} {1986}{\natexlab{b}})}\BibitemShut
  {NoStop}%
%%CITATION = NUPHA,B278,417;%%
\bibitem [{\citenamefont {Kluberg-Stern}\ \emph {et~al.}(1983)\citenamefont
  {Kluberg-Stern}, \citenamefont {Morel}, \citenamefont {Napoly},\ and\
  \citenamefont {Petersson}}]{KlubergStern:1983dg}%
  \BibitemOpen
  \bibfield  {author} {\bibinfo {author} {\bibfnamefont {H.}~\bibnamefont
  {Kluberg-Stern}}, \bibinfo {author} {\bibfnamefont {A.}~\bibnamefont
  {Morel}}, \bibinfo {author} {\bibfnamefont {O.}~\bibnamefont {Napoly}},\ and\
  \bibinfo {author} {\bibfnamefont {B.}~\bibnamefont {Petersson}},\ }\bibfield
  {title} {\bibinfo {title} {{Flavors of Lagrangian Susskind Fermions}},\
  }\href {https://doi.org/10.1016/0550-3213(83)90501-1} {\bibfield  {journal}
  {\bibinfo  {journal} {Nucl. Phys. B}\ }\textbf {\bibinfo {volume} {220}},\
  \bibinfo {pages} {447} (\bibinfo {year} {1983})}\BibitemShut {NoStop}%
\bibitem [{\citenamefont {Verstegen}(1985)}]{Verstegen:1985kt}%
  \BibitemOpen
  \bibfield  {author} {\bibinfo {author} {\bibfnamefont {D.}~\bibnamefont
  {Verstegen}},\ }\bibfield  {title} {\bibinfo {title} {{Symmetry Properties of
  Fermionic Bilinears}},\ }\href {https://doi.org/10.1016/0550-3213(85)90029-X}
  {\bibfield  {journal} {\bibinfo  {journal} {Nucl. Phys.}\ }\textbf {\bibinfo
  {volume} {B249}},\ \bibinfo {pages} {685} (\bibinfo {year}
  {1985})}\BibitemShut {NoStop}%
%%CITATION = NUPHA,B249,685;%%
\bibitem [{\citenamefont {Lee}(2001)}]{Lee:2001hc}%
  \BibitemOpen
  \bibfield  {author} {\bibinfo {author} {\bibfnamefont {W.-J.}\ \bibnamefont
  {Lee}},\ }\bibfield  {title} {\bibinfo {title} {{Perturbative matching of the
  staggered four fermion operators for e-prime / epsilon}},\ }\href
  {https://doi.org/10.1103/PhysRevD.64.054505} {\bibfield  {journal} {\bibinfo
  {journal} {Phys. Rev. D}\ }\textbf {\bibinfo {volume} {64}},\ \bibinfo
  {pages} {054505} (\bibinfo {year} {2001})},\ \Eprint
  {https://arxiv.org/abs/hep-lat/0106005} {arXiv:hep-lat/0106005} \BibitemShut
  {NoStop}%
\bibitem [{\citenamefont {Lee}\ and\ \citenamefont
  {Klomfass}(1995)}]{Lee:1994xs}%
  \BibitemOpen
  \bibfield  {author} {\bibinfo {author} {\bibfnamefont {W.-J.}\ \bibnamefont
  {Lee}}\ and\ \bibinfo {author} {\bibfnamefont {M.}~\bibnamefont {Klomfass}},\
  }\bibfield  {title} {\bibinfo {title} {{One spin trace formalism for B(K)}},\
  }\href {https://doi.org/10.1103/PhysRevD.51.6426} {\bibfield  {journal}
  {\bibinfo  {journal} {Phys. Rev. D}\ }\textbf {\bibinfo {volume} {51}},\
  \bibinfo {pages} {6426} (\bibinfo {year} {1995})},\ \Eprint
  {https://arxiv.org/abs/hep-lat/9412039} {arXiv:hep-lat/9412039} \BibitemShut
  {NoStop}%
\bibitem [{\citenamefont {Lee}(2002)}]{Lee:2002fj}%
  \BibitemOpen
  \bibfield  {author} {\bibinfo {author} {\bibfnamefont {W.-j.}\ \bibnamefont
  {Lee}},\ }\bibfield  {title} {\bibinfo {title} {{Perturbative improvement of
  staggered fermions using fat links}},\ }\href
  {https://doi.org/10.1103/PhysRevD.66.114504} {\bibfield  {journal} {\bibinfo
  {journal} {Phys. Rev.}\ }\textbf {\bibinfo {volume} {D66}},\ \bibinfo {pages}
  {114504} (\bibinfo {year} {2002})},\ \Eprint
  {https://arxiv.org/abs/hep-lat/0208032} {arXiv:hep-lat/0208032 [hep-lat]}
  \BibitemShut {NoStop}%
%%CITATION = HEP-LAT/0208032;%%
\bibitem [{\citenamefont {Orginos}\ \emph {et~al.}(1999)\citenamefont
  {Orginos}, \citenamefont {Toussaint},\ and\ \citenamefont
  {Sugar}}]{Orginos:1999cr}%
  \BibitemOpen
  \bibfield  {author} {\bibinfo {author} {\bibfnamefont {K.}~\bibnamefont
  {Orginos}}, \bibinfo {author} {\bibfnamefont {D.}~\bibnamefont {Toussaint}},\
  and\ \bibinfo {author} {\bibfnamefont {R.~L.}\ \bibnamefont {Sugar}}
  (\bibinfo {collaboration} {MILC}),\ }\bibfield  {title} {\bibinfo {title}
  {{Variants of fattening and flavor symmetry restoration}},\ }\href
  {https://doi.org/10.1103/PhysRevD.60.054503} {\bibfield  {journal} {\bibinfo
  {journal} {Phys. Rev.}\ }\textbf {\bibinfo {volume} {D60}},\ \bibinfo {pages}
  {054503} (\bibinfo {year} {1999})},\ \Eprint
  {https://arxiv.org/abs/hep-lat/9903032} {arXiv:hep-lat/9903032 [hep-lat]}
  \BibitemShut {NoStop}%
%%CITATION = HEP-LAT/9903032;%%
\bibitem [{\citenamefont {Lepage}(1999)}]{Lepage:1998vj}%
  \BibitemOpen
  \bibfield  {author} {\bibinfo {author} {\bibfnamefont {G.~P.}\ \bibnamefont
  {Lepage}},\ }\bibfield  {title} {\bibinfo {title} {{Flavor symmetry
  restoration and Symanzik improvement for staggered quarks}},\ }\href
  {https://doi.org/10.1103/PhysRevD.59.074502} {\bibfield  {journal} {\bibinfo
  {journal} {Phys. Rev.}\ }\textbf {\bibinfo {volume} {D59}},\ \bibinfo {pages}
  {074502} (\bibinfo {year} {1999})},\ \Eprint
  {https://arxiv.org/abs/hep-lat/9809157} {arXiv:hep-lat/9809157 [hep-lat]}
  \BibitemShut {NoStop}%
%%CITATION = HEP-LAT/9809157;%%
\bibitem [{\citenamefont {Luscher}\ and\ \citenamefont
  {Weisz}(1985{\natexlab{a}})}]{Luscher:1984xn}%
  \BibitemOpen
  \bibfield  {author} {\bibinfo {author} {\bibfnamefont {M.}~\bibnamefont
  {Luscher}}\ and\ \bibinfo {author} {\bibfnamefont {P.}~\bibnamefont
  {Weisz}},\ }\bibfield  {title} {\bibinfo {title} {{On-Shell Improved Lattice
  Gauge Theories}},\ }\href {https://doi.org/10.1007/BF01206178} {\bibfield
  {journal} {\bibinfo  {journal} {Commun. Math. Phys.}\ }\textbf {\bibinfo
  {volume} {97}},\ \bibinfo {pages} {59} (\bibinfo {year}
  {1985}{\natexlab{a}})},\ \bibinfo {note} {[Erratum: Commun. Math.
  Phys.98,433(1985)]}\BibitemShut {NoStop}%
%%CITATION = CMPHA,97,59;%%
\bibitem [{\citenamefont {Luscher}\ and\ \citenamefont
  {Weisz}(1985{\natexlab{b}})}]{Luscher:1985zq}%
  \BibitemOpen
  \bibfield  {author} {\bibinfo {author} {\bibfnamefont {M.}~\bibnamefont
  {Luscher}}\ and\ \bibinfo {author} {\bibfnamefont {P.}~\bibnamefont
  {Weisz}},\ }\bibfield  {title} {\bibinfo {title} {{Computation of the Action
  for On-Shell Improved Lattice Gauge Theories at Weak Coupling}},\ }\href
  {https://doi.org/10.1016/0370-2693(85)90966-9} {\bibfield  {journal}
  {\bibinfo  {journal} {Phys. Lett.}\ }\textbf {\bibinfo {volume} {158B}},\
  \bibinfo {pages} {250} (\bibinfo {year} {1985}{\natexlab{b}})}\BibitemShut
  {NoStop}%
%%CITATION = PHLTA,158B,250;%%
\bibitem [{\citenamefont {Alford}\ \emph {et~al.}(1995)\citenamefont {Alford},
  \citenamefont {Dimm}, \citenamefont {Lepage}, \citenamefont {Hockney},\ and\
  \citenamefont {Mackenzie}}]{Alford:1995hw}%
  \BibitemOpen
  \bibfield  {author} {\bibinfo {author} {\bibfnamefont {M.~G.}\ \bibnamefont
  {Alford}}, \bibinfo {author} {\bibfnamefont {W.}~\bibnamefont {Dimm}},
  \bibinfo {author} {\bibfnamefont {G.~P.}\ \bibnamefont {Lepage}}, \bibinfo
  {author} {\bibfnamefont {G.}~\bibnamefont {Hockney}},\ and\ \bibinfo {author}
  {\bibfnamefont {P.~B.}\ \bibnamefont {Mackenzie}},\ }\bibfield  {title}
  {\bibinfo {title} {{Lattice QCD on small computers}},\ }\href
  {https://doi.org/10.1016/0370-2693(95)01131-9} {\bibfield  {journal}
  {\bibinfo  {journal} {Phys. Lett.}\ }\textbf {\bibinfo {volume} {B361}},\
  \bibinfo {pages} {87} (\bibinfo {year} {1995})},\ \Eprint
  {https://arxiv.org/abs/hep-lat/9507010} {arXiv:hep-lat/9507010 [hep-lat]}
  \BibitemShut {NoStop}%
%%CITATION = HEP-LAT/9507010;%%
\bibitem [{\citenamefont {Bonnet}\ \emph {et~al.}(2002)\citenamefont {Bonnet},
  \citenamefont {Leinweber}, \citenamefont {Williams},\ and\ \citenamefont
  {Zanotti}}]{Bonnet:2001rc}%
  \BibitemOpen
  \bibfield  {author} {\bibinfo {author} {\bibfnamefont {F.~D.~R.}\
  \bibnamefont {Bonnet}}, \bibinfo {author} {\bibfnamefont {D.~B.}\
  \bibnamefont {Leinweber}}, \bibinfo {author} {\bibfnamefont {A.~G.}\
  \bibnamefont {Williams}},\ and\ \bibinfo {author} {\bibfnamefont {J.~M.}\
  \bibnamefont {Zanotti}},\ }\bibfield  {title} {\bibinfo {title} {{Improved
  smoothing algorithms for lattice gauge theory}},\ }\href
  {https://doi.org/10.1103/PhysRevD.65.114510} {\bibfield  {journal} {\bibinfo
  {journal} {Phys. Rev.}\ }\textbf {\bibinfo {volume} {D65}},\ \bibinfo {pages}
  {114510} (\bibinfo {year} {2002})},\ \Eprint
  {https://arxiv.org/abs/hep-lat/0106023} {arXiv:hep-lat/0106023 [hep-lat]}
  \BibitemShut {NoStop}%
%%CITATION = HEP-LAT/0106023;%%
\bibitem [{\citenamefont {Kim}\ \emph {et~al.}(2011)\citenamefont {Kim},
  \citenamefont {Lee},\ and\ \citenamefont {Sharpe}}]{Kim:2011pz}%
  \BibitemOpen
  \bibfield  {author} {\bibinfo {author} {\bibfnamefont {J.}~\bibnamefont
  {Kim}}, \bibinfo {author} {\bibfnamefont {W.}~\bibnamefont {Lee}},\ and\
  \bibinfo {author} {\bibfnamefont {S.~R.}\ \bibnamefont {Sharpe}},\ }\bibfield
   {title} {\bibinfo {title} {{One-loop matching of improved four-fermion
  staggered operators with an improved gluon action}},\ }\href
  {https://doi.org/10.1103/PhysRevD.83.094503} {\bibfield  {journal} {\bibinfo
  {journal} {Phys. Rev.}\ }\textbf {\bibinfo {volume} {D83}},\ \bibinfo {pages}
  {094503} (\bibinfo {year} {2011})},\ \Eprint
  {https://arxiv.org/abs/1102.1774} {arXiv:1102.1774 [hep-lat]} \BibitemShut
  {NoStop}%
%%CITATION = ARXIV:1102.1774;%%
\bibitem [{\citenamefont {Kim}\ \emph {et~al.}(2010)\citenamefont {Kim},
  \citenamefont {Lee},\ and\ \citenamefont {Sharpe}}]{Kim:2010fj}%
  \BibitemOpen
  \bibfield  {author} {\bibinfo {author} {\bibfnamefont {J.}~\bibnamefont
  {Kim}}, \bibinfo {author} {\bibfnamefont {W.}~\bibnamefont {Lee}},\ and\
  \bibinfo {author} {\bibfnamefont {S.~R.}\ \bibnamefont {Sharpe}},\ }\bibfield
   {title} {\bibinfo {title} {{One-loop matching factors for staggered bilinear
  operators with improved gauge actions}},\ }\href
  {https://doi.org/10.1103/PhysRevD.81.114503} {\bibfield  {journal} {\bibinfo
  {journal} {Phys. Rev.}\ }\textbf {\bibinfo {volume} {D81}},\ \bibinfo {pages}
  {114503} (\bibinfo {year} {2010})},\ \Eprint
  {https://arxiv.org/abs/1004.4039} {arXiv:1004.4039 [hep-lat]} \BibitemShut
  {NoStop}%
%%CITATION = ARXIV:1004.4039;%%
\bibitem [{\citenamefont {de~Forcrand}\ \emph {et~al.}(1997)\citenamefont
  {de~Forcrand} \emph {et~al.}}]{deForcrand:1997esx}%
  \BibitemOpen
  \bibfield  {author} {\bibinfo {author} {\bibfnamefont {P.}~\bibnamefont
  {de~Forcrand}} \emph {et~al.},\ }\bibfield  {title} {\bibinfo {title}
  {{Topology of the SU(2) vacuum: A Lattice study using improved cooling}},\
  }\href {https://doi.org/10.1016/S0550-3213(97)00275-7} {\bibfield  {journal}
  {\bibinfo  {journal} {Nucl. Phys.}\ }\textbf {\bibinfo {volume} {B499}},\
  \bibinfo {pages} {409} (\bibinfo {year} {1997})},\ \Eprint
  {https://arxiv.org/abs/hep-lat/9701012} {arXiv:hep-lat/9701012 [hep-lat]}
  \BibitemShut {NoStop}%
%%CITATION = HEP-LAT/9701012;%%
\bibitem [{\citenamefont {de~Forcrand}\ \emph {et~al.}(1996)\citenamefont
  {de~Forcrand} \emph {et~al.}}]{deForcrand:1995qq}%
  \BibitemOpen
  \bibfield  {author} {\bibinfo {author} {\bibfnamefont {P.}~\bibnamefont
  {de~Forcrand}} \emph {et~al.},\ }\bibfield  {title} {\bibinfo {title}
  {{Improved cooling algorithm for gauge theories}},\ }\bibfield  {booktitle}
  {\emph {\bibinfo {booktitle} {{Lattice '95. Proceedings, International
  Symposium on Lattice Field Theory, Melbourne, Australia, July 11-15,
  1995}}},\ }\href {https://doi.org/10.1016/0920-5632(96)00172-7} {\bibfield
  {journal} {\bibinfo  {journal} {Nucl. Phys. Proc. Suppl.}\ }\textbf {\bibinfo
  {volume} {47}},\ \bibinfo {pages} {777} (\bibinfo {year} {1996})},\ \Eprint
  {https://arxiv.org/abs/hep-lat/9509064} {arXiv:hep-lat/9509064 [hep-lat]}
  \BibitemShut {NoStop}%
%%CITATION = HEP-LAT/9509064;%%
\bibitem [{\citenamefont {Cichy}\ \emph {et~al.}(2014)\citenamefont {Cichy}
  \emph {et~al.}}]{Cichy:2014qta}%
  \BibitemOpen
  \bibfield  {author} {\bibinfo {author} {\bibfnamefont {K.}~\bibnamefont
  {Cichy}} \emph {et~al.},\ }\bibfield  {title} {\bibinfo {title} {{Comparison
  of different lattice definitions of the topological charge}},\ }\bibfield
  {booktitle} {\emph {\bibinfo {booktitle} {{Proceedings, 32nd International
  Symposium on Lattice Field Theory (Lattice 2014): Brookhaven, NY, USA, June
  23-28, 2014}}},\ }\href {https://doi.org/10.22323/1.214.0075} {\bibfield
  {journal} {\bibinfo  {journal} {PoS}\ }\textbf {\bibinfo {volume}
  {LATTICE2014}},\ \bibinfo {pages} {075} (\bibinfo {year} {2014})},\ \Eprint
  {https://arxiv.org/abs/1411.1205} {arXiv:1411.1205 [hep-lat]} \BibitemShut
  {NoStop}%
%%CITATION = ARXIV:1411.1205;%%
\bibitem [{\citenamefont {Hasenfratz}\ and\ \citenamefont
  {Nieter}(1998)}]{Hasenfratz:1998qk}%
  \BibitemOpen
  \bibfield  {author} {\bibinfo {author} {\bibfnamefont {A.}~\bibnamefont
  {Hasenfratz}}\ and\ \bibinfo {author} {\bibfnamefont {C.}~\bibnamefont
  {Nieter}},\ }\bibfield  {title} {\bibinfo {title} {{Instanton content of the
  SU(3) vacuum}},\ }\href {https://doi.org/10.1016/S0370-2693(98)01058-2}
  {\bibfield  {journal} {\bibinfo  {journal} {Phys. Lett.}\ }\textbf {\bibinfo
  {volume} {B439}},\ \bibinfo {pages} {366} (\bibinfo {year} {1998})},\ \Eprint
  {https://arxiv.org/abs/hep-lat/9806026} {arXiv:hep-lat/9806026 [hep-lat]}
  \BibitemShut {NoStop}%
%%CITATION = HEP-LAT/9806026;%%
\bibitem [{\citenamefont {Falcioni}\ \emph {et~al.}(1985)\citenamefont
  {Falcioni}, \citenamefont {Paciello}, \citenamefont {Parisi},\ and\
  \citenamefont {Taglienti}}]{Falcioni:1984ei}%
  \BibitemOpen
  \bibfield  {author} {\bibinfo {author} {\bibfnamefont {M.}~\bibnamefont
  {Falcioni}}, \bibinfo {author} {\bibfnamefont {M.~L.}\ \bibnamefont
  {Paciello}}, \bibinfo {author} {\bibfnamefont {G.}~\bibnamefont {Parisi}},\
  and\ \bibinfo {author} {\bibfnamefont {B.}~\bibnamefont {Taglienti}},\
  }\bibfield  {title} {\bibinfo {title} {{AGAIN ON SU(3) GLUEBALL MASS}},\
  }\href {https://doi.org/10.1016/0550-3213(85)90280-9} {\bibfield  {journal}
  {\bibinfo  {journal} {Nucl. Phys.}\ }\textbf {\bibinfo {volume} {B251}},\
  \bibinfo {pages} {624} (\bibinfo {year} {1985})}\BibitemShut {NoStop}%
%%CITATION = NUPHA,B251,624;%%
\bibitem [{\citenamefont {Donald}\ \emph {et~al.}(2011)\citenamefont {Donald},
  \citenamefont {Davies}, \citenamefont {Follana},\ and\ \citenamefont
  {Kronfeld}}]{Donald:2011if}%
  \BibitemOpen
  \bibfield  {author} {\bibinfo {author} {\bibfnamefont {G.~C.}\ \bibnamefont
  {Donald}}, \bibinfo {author} {\bibfnamefont {C.~T.~H.}\ \bibnamefont
  {Davies}}, \bibinfo {author} {\bibfnamefont {E.}~\bibnamefont {Follana}},\
  and\ \bibinfo {author} {\bibfnamefont {A.~S.}\ \bibnamefont {Kronfeld}},\
  }\bibfield  {title} {\bibinfo {title} {{Staggered fermions, zero modes, and
  flavor-singlet mesons}},\ }\href {https://doi.org/10.1103/PhysRevD.84.054504}
  {\bibfield  {journal} {\bibinfo  {journal} {Phys. Rev.}\ }\textbf {\bibinfo
  {volume} {D84}},\ \bibinfo {pages} {054504} (\bibinfo {year} {2011})},\
  \Eprint {https://arxiv.org/abs/1106.2412} {arXiv:1106.2412 [hep-lat]}
  \BibitemShut {NoStop}%
%%CITATION = ARXIV:1106.2412;%%
\bibitem [{\citenamefont {Goodfellow}\ \emph {et~al.}(2016)\citenamefont
  {Goodfellow}, \citenamefont {Bengio},\ and\ \citenamefont
  {Courville}}]{Goodfellow-et-al-2016}%
  \BibitemOpen
  \bibfield  {author} {\bibinfo {author} {\bibfnamefont {I.}~\bibnamefont
  {Goodfellow}}, \bibinfo {author} {\bibfnamefont {Y.}~\bibnamefont {Bengio}},\
  and\ \bibinfo {author} {\bibfnamefont {A.}~\bibnamefont {Courville}},\
  }\href@noop {} {\emph {\bibinfo {title} {Deep Learning}}}\ (\bibinfo
  {publisher} {MIT Press},\ \bibinfo {year} {2016})\ \bibinfo {note}
  {\url{http://www.deeplearningbook.org}}\BibitemShut {NoStop}%
\bibitem [{\citenamefont {Chollet}\ \emph {et~al.}(2015)\citenamefont {Chollet}
  \emph {et~al.}}]{chollet2015keras}%
  \BibitemOpen
  \bibfield  {author} {\bibinfo {author} {\bibfnamefont {F.}~\bibnamefont
  {Chollet}} \emph {et~al.},\ }\href@noop {} {\bibinfo {title} {Keras}},\
  \bibinfo {howpublished} {\url{https://keras.io}} (\bibinfo {year}
  {2015})\BibitemShut {NoStop}%
\bibitem [{\citenamefont {Kingma}\ and\ \citenamefont
  {Ba}(2014)}]{Kingma:2014vow}%
  \BibitemOpen
  \bibfield  {author} {\bibinfo {author} {\bibfnamefont {D.~P.}\ \bibnamefont
  {Kingma}}\ and\ \bibinfo {author} {\bibfnamefont {J.}~\bibnamefont {Ba}},\
  }\bibfield  {title} {\bibinfo {title} {{Adam: A Method for Stochastic
  Optimization}},\ }\Eprint {https://arxiv.org/abs/1412.6980} {arXiv:1412.6980
  [cs.LG]}  (\bibinfo {year} {2014})\BibitemShut {NoStop}%
\bibitem [{Note1()}]{Note1}%
  \BibitemOpen
  \bibinfo {note} {Popular and basic loss functions such as the mean squared
  error (MSE) and mean absolute error (MAE) are usually used for regression
  problems. On the contrary, the categorical cross-entropy loss function is
  most applicable to multi-class classification problems.}\BibitemShut {Stop}%
\bibitem [{Note2()}]{Note2}%
  \BibitemOpen
  \bibinfo {note} {Popular optimization methods available in the market are
  stochastic gradient descent, AdaGrad, RMSprop, and Adam \cite
  {Kingma:2014vow}.}\BibitemShut {Stop}%
\bibitem [{Note3()}]{Note3}%
  \BibitemOpen
  \bibinfo {note} {Popular activation functions are the $\protect \qopname
  \relax o{tanh}$, sigmoid, and ReLU. Here, we make use of ReLU for the hidden
  layers since it is the simplest and fastest. The softmax function is
  essential for the output layer of the multi-class
  classification.}\BibitemShut {Stop}%
\bibitem [{\citenamefont {Jeong}\ \emph {et~al.}()\citenamefont {Jeong},
  \citenamefont {Lee},\ and\ \citenamefont {Lee}}]{skLee:prep}%
  \BibitemOpen
  \bibfield  {author} {\bibinfo {author} {\bibfnamefont {H.}~\bibnamefont
  {Jeong}}, \bibinfo {author} {\bibfnamefont {S.}~\bibnamefont {Lee}},\ and\
  \bibinfo {author} {\bibfnamefont {W.}~\bibnamefont {Lee}},\ }\bibinfo {note}
  {in preparation}\BibitemShut {NoStop}%
\bibitem [{\citenamefont {Aoki}\ \emph {et~al.}(2008)\citenamefont {Aoki} \emph
  {et~al.}}]{Aoki:2007xm}%
  \BibitemOpen
  \bibfield  {author} {\bibinfo {author} {\bibfnamefont {Y.}~\bibnamefont
  {Aoki}} \emph {et~al.},\ }\bibfield  {title} {\bibinfo {title}
  {{Non-perturbative renormalization of quark bilinear operators and B(K) using
  domain wall fermions}},\ }\href {https://doi.org/10.1103/PhysRevD.78.054510}
  {\bibfield  {journal} {\bibinfo  {journal} {Phys. Rev.}\ }\textbf {\bibinfo
  {volume} {D78}},\ \bibinfo {pages} {054510} (\bibinfo {year} {2008})},\
  \Eprint {https://arxiv.org/abs/0712.1061} {arXiv:0712.1061 [hep-lat]}
  \BibitemShut {NoStop}%
%%CITATION = ARXIV:0712.1061;%%
\bibitem [{\citenamefont {Sturm}\ \emph {et~al.}(2009)\citenamefont {Sturm},
  \citenamefont {Aoki}, \citenamefont {Christ}, \citenamefont {Izubuchi},
  \citenamefont {Sachrajda},\ and\ \citenamefont {Soni}}]{Sturm:2009kb}%
  \BibitemOpen
  \bibfield  {author} {\bibinfo {author} {\bibfnamefont {C.}~\bibnamefont
  {Sturm}}, \bibinfo {author} {\bibfnamefont {Y.}~\bibnamefont {Aoki}},
  \bibinfo {author} {\bibfnamefont {N.~H.}\ \bibnamefont {Christ}}, \bibinfo
  {author} {\bibfnamefont {T.}~\bibnamefont {Izubuchi}}, \bibinfo {author}
  {\bibfnamefont {C.~T.~C.}\ \bibnamefont {Sachrajda}},\ and\ \bibinfo {author}
  {\bibfnamefont {A.}~\bibnamefont {Soni}},\ }\bibfield  {title} {\bibinfo
  {title} {{Renormalization of quark bilinear operators in a
  momentum-subtraction scheme with a nonexceptional subtraction point}},\
  }\href {https://doi.org/10.1103/PhysRevD.80.014501} {\bibfield  {journal}
  {\bibinfo  {journal} {Phys. Rev.}\ }\textbf {\bibinfo {volume} {D80}},\
  \bibinfo {pages} {014501} (\bibinfo {year} {2009})},\ \Eprint
  {https://arxiv.org/abs/0901.2599} {arXiv:0901.2599 [hep-ph]} \BibitemShut
  {NoStop}%
%%CITATION = ARXIV:0901.2599;%%
\bibitem [{\citenamefont {Lehoucq}\ and\ \citenamefont
  {Sorensen}(1996)}]{Lehoucq96deflationtechniques}%
  \BibitemOpen
  \bibfield  {author} {\bibinfo {author} {\bibfnamefont {R.}~\bibnamefont
  {Lehoucq}}\ and\ \bibinfo {author} {\bibfnamefont {D.~C.}\ \bibnamefont
  {Sorensen}},\ }\bibfield  {title} {\bibinfo {title} {Deflation techniques for
  an implicitly re-started arnoldi iteration},\ }\href@noop {} {\bibfield
  {journal} {\bibinfo  {journal} {SIAM J. Matrix Anal. Appl}\ }\textbf
  {\bibinfo {volume} {17}},\ \bibinfo {pages} {789} (\bibinfo {year}
  {1996})}\BibitemShut {NoStop}%
\bibitem [{\citenamefont {Saad}(1984{\natexlab{b}})}]{MR736453}%
  \BibitemOpen
  \bibfield  {author} {\bibinfo {author} {\bibfnamefont {Y.}~\bibnamefont
  {Saad}},\ }\bibfield  {title} {\bibinfo {title} {Chebyshev acceleration
  techniques for solving nonsymmetric eigenvalue problems},\ }\href
  {https://doi.org/10.2307/2007602} {\bibfield  {journal} {\bibinfo  {journal}
  {Math. Comp.}\ }\textbf {\bibinfo {volume} {42}},\ \bibinfo {pages} {567}
  (\bibinfo {year} {1984}{\natexlab{b}})}\BibitemShut {NoStop}%
\bibitem [{\citenamefont {Cullum}\ and\ \citenamefont
  {Willoughby}(1985)}]{Cullum:1985aa}%
  \BibitemOpen
  \bibfield  {author} {\bibinfo {author} {\bibfnamefont {J.}~\bibnamefont
  {Cullum}}\ and\ \bibinfo {author} {\bibfnamefont {R.}~\bibnamefont
  {Willoughby}},\ }\bibfield  {title} {\bibinfo {title} {A survey of lanczos
  procedures for very large real ‘symmetric’ eigenvalue problems},\ }\href
  {https://doi.org/10.1016/0377-0427(85)90006-8} {\bibfield  {journal}
  {\bibinfo  {journal} {Journal of Computational and Applied Mathematics}\
  }\textbf {\bibinfo {volume} {s 12–13}},\ \bibinfo {pages} {37–60}
  (\bibinfo {year} {1985})}\BibitemShut {NoStop}%
\bibitem [{\citenamefont {Cullum}\ and\ \citenamefont
  {Willoughby}(1981)}]{CULLUM1981329}%
  \BibitemOpen
  \bibfield  {author} {\bibinfo {author} {\bibfnamefont {J.}~\bibnamefont
  {Cullum}}\ and\ \bibinfo {author} {\bibfnamefont {R.~A.}\ \bibnamefont
  {Willoughby}},\ }\bibfield  {title} {\bibinfo {title} {Computing eigenvalues
  of very large symmetric matrices—an implementation of a lanczos algorithm
  with no reorthogonalization},\ }\href
  {https://doi.org/https://doi.org/10.1016/0021-9991(81)90056-5} {\bibfield
  {journal} {\bibinfo  {journal} {Journal of Computational Physics}\ }\textbf
  {\bibinfo {volume} {44}},\ \bibinfo {pages} {329 } (\bibinfo {year}
  {1981})}\BibitemShut {NoStop}%
\bibitem [{\citenamefont {Cullum}\ and\ \citenamefont
  {Willoughby}(2002)}]{Cullum:2002aa}%
  \BibitemOpen
  \bibfield  {author} {\bibinfo {author} {\bibfnamefont {J.}~\bibnamefont
  {Cullum}}\ and\ \bibinfo {author} {\bibfnamefont {R.}~\bibnamefont
  {Willoughby}},\ }\bibfield  {title} {\bibinfo {title} {Lanczos algorithms for
  large symmetric eigenvalue computations. vol. i: Theory},\ }\href
  {https://doi.org/10.1137/1.9780898719192} {\bibfield  {journal} {\bibinfo
  {journal} {Classics in Applied Mathematics}\ }\textbf {\bibinfo {volume} {I}}
  (\bibinfo {year} {2002})}\BibitemShut {NoStop}%
\bibitem [{\citenamefont {Clark}\ \emph {et~al.}(2018)\citenamefont {Clark},
  \citenamefont {Jung},\ and\ \citenamefont {Lehner}}]{Clark:2017wom}%
  \BibitemOpen
  \bibfield  {author} {\bibinfo {author} {\bibfnamefont {M.~A.}\ \bibnamefont
  {Clark}}, \bibinfo {author} {\bibfnamefont {C.}~\bibnamefont {Jung}},\ and\
  \bibinfo {author} {\bibfnamefont {C.}~\bibnamefont {Lehner}},\ }\bibfield
  {title} {\bibinfo {title} {{Multi-Grid Lanczos}},\ }\bibfield  {booktitle}
  {\emph {\bibinfo {booktitle} {{Proceedings, 35th International Symposium on
  Lattice Field Theory (Lattice 2017): Granada, Spain, June 18-24, 2017}}},\
  }\href {https://doi.org/10.1051/epjconf/201817514023} {\bibfield  {journal}
  {\bibinfo  {journal} {EPJ Web Conf.}\ }\textbf {\bibinfo {volume} {175}},\
  \bibinfo {pages} {14023} (\bibinfo {year} {2018})},\ \Eprint
  {https://arxiv.org/abs/1710.06884} {arXiv:1710.06884 [hep-lat]} \BibitemShut
  {NoStop}%
%%CITATION = ARXIV:1710.06884;%%
\bibitem [{\citenamefont {Jang}\ and\ \citenamefont
  {Jung}(2019)}]{Jang:2019roq}%
  \BibitemOpen
  \bibfield  {author} {\bibinfo {author} {\bibfnamefont {Y.-C.}\ \bibnamefont
  {Jang}}\ and\ \bibinfo {author} {\bibfnamefont {C.}~\bibnamefont {Jung}},\
  }\bibfield  {title} {\bibinfo {title} {{Split Grid and Block Lanczos
  Algorithm for Efficient Eigenpair Generation}},\ }\bibfield  {booktitle}
  {\emph {\bibinfo {booktitle} {{Proceedings, 36th International Symposium on
  Lattice Field Theory (Lattice 2018): East Lansing, MI, United States, July
  22-28, 2018}}},\ }\href {https://doi.org/10.22323/1.334.0309} {\bibfield
  {journal} {\bibinfo  {journal} {PoS}\ }\textbf {\bibinfo {volume}
  {LATTICE2018}},\ \bibinfo {pages} {309} (\bibinfo {year} {2019})}\BibitemShut
  {NoStop}%
%%CITATION = POSCI,LATTICE2018,309;%%
\bibitem [{\citenamefont {Luscher}(2007)}]{Luscher:2007se}%
  \BibitemOpen
  \bibfield  {author} {\bibinfo {author} {\bibfnamefont {M.}~\bibnamefont
  {Luscher}},\ }\bibfield  {title} {\bibinfo {title} {{Local coherence and
  deflation of the low quark modes in lattice QCD}},\ }\href
  {https://doi.org/10.1088/1126-6708/2007/07/081} {\bibfield  {journal}
  {\bibinfo  {journal} {JHEP}\ }\textbf {\bibinfo {volume} {07}},\ \bibinfo
  {pages} {081}},\ \Eprint {https://arxiv.org/abs/0706.2298} {arXiv:0706.2298
  [hep-lat]} \BibitemShut {NoStop}%
\bibitem [{\citenamefont {Bae}\ \emph {et~al.}(2008)\citenamefont {Bae},
  \citenamefont {Adams}, \citenamefont {Jung}, \citenamefont {Kim},
  \citenamefont {Kim}, \citenamefont {Kim}, \citenamefont {Lee},\ and\
  \citenamefont {Sharpe}}]{Bae:2008qe}%
  \BibitemOpen
  \bibfield  {author} {\bibinfo {author} {\bibfnamefont {T.}~\bibnamefont
  {Bae}}, \bibinfo {author} {\bibfnamefont {D.~H.}\ \bibnamefont {Adams}},
  \bibinfo {author} {\bibfnamefont {C.}~\bibnamefont {Jung}}, \bibinfo {author}
  {\bibfnamefont {H.-J.}\ \bibnamefont {Kim}}, \bibinfo {author} {\bibfnamefont
  {J.}~\bibnamefont {Kim}}, \bibinfo {author} {\bibfnamefont {K.}~\bibnamefont
  {Kim}}, \bibinfo {author} {\bibfnamefont {W.}~\bibnamefont {Lee}},\ and\
  \bibinfo {author} {\bibfnamefont {S.~R.}\ \bibnamefont {Sharpe}},\ }\bibfield
   {title} {\bibinfo {title} {{Taste symmetry breaking with HYP-smeared
  staggered fermions}},\ }\href {https://doi.org/10.1103/PhysRevD.77.094508}
  {\bibfield  {journal} {\bibinfo  {journal} {Phys. Rev. D}\ }\textbf {\bibinfo
  {volume} {77}},\ \bibinfo {pages} {094508} (\bibinfo {year} {2008})},\
  \Eprint {https://arxiv.org/abs/0801.3000} {arXiv:0801.3000 [hep-lat]}
  \BibitemShut {NoStop}%
\end{thebibliography}%
\end{document}